\documentclass[a4paper,12pt,oneside,reqno]{amsart}

\usepackage{textcomp}
\usepackage{hyperref}
\usepackage[headinclude,DIV=13]{typearea}
\areaset{15.1cm}{25.0cm}
\usepackage{amsfonts,amssymb,amsmath,amsthm,bbm,longtable,verbatim}
\usepackage{dsfont}
\usepackage[latin1] {inputenc}
\usepackage{graphicx, psfrag}
\usepackage{color}
\usepackage{bold-extra}
\usepackage{orcidlink}

\newcommand{\hypgeo}[2]{%
  \operatorname{%
    {\vphantom{\mathnormal{F}}}_{#1}%
    \kern-\scriptspace
    \mathnormal{F}_{#2}%
  }%
}

\newtheorem{df}{Definition}[section]
\newtheorem{prop}{Proposition}[section]
\newtheorem{theo}{Theorem}[section]
\newtheorem{lemma}{Lemma}[section]

\newtheorem{cor}{Corollary}[section]

\newtheorem{rkk}{Remark}[section]

\newtheorem{claim}{Claim}[section]
\newtheorem{consequence}{Consequence}[section]

\theoremstyle{remark}

\def\paragraph#1{\noindent \textbf{#1}}

\numberwithin{equation}{section}

\newcommand{\D}{{\mathbb D}}

\newcommand{\R}{{\mathbb R}}

\DeclareGraphicsExtensions{.pdf,.jpg,.png}
\definecolor{darkolivegreen}{rgb}{0.33, 0.42, 0.18}
\definecolor{britishracinggreen}{rgb}{0.0, 0.26, 0.15}
\definecolor{ao(english)}{rgb}{0.0, 0.5, 0.0}
\definecolor{azure(colorwheel)}{rgb}{0.0, 0.5, 1.0}
\definecolor{bleudefrance}{rgb}{0.19, 0.55, 0.91}

\newcommand{\bb}{\mathbb}

\def\D{\mathbb{D}}
\def\Ti#1/{\tilde{#1}}

\def \expect {{\bf E}}

\def\Z{\mathbb{Z}}

\begin{document}

\title{Complex Generalized Integral Means Spectrum of  Drifted Whole-Plane SLE \& LLE}
\date{}
 \author[Bertrand Duplantier]{Bertrand Duplantier\,\orcidlink{0000-0002-0765-8476}}
      \address{Bertrand Duplantier\\Universit\'e Paris-Saclay\\CNRS\\ CEA\\Institut de Physique Th\'{e}orique
\\91191\\ Gif-sur-Yvette Cedex\\ France}
\email{bertrand.duplantier@ipht.fr}
   \author[Yong Han]{Yong Han\,\orcidlink{0000-0001-6149-3724}}
      \address{Yong Han\\ College of Mathematics and Statistics\\ Shenzhen University\\ Shenzhen 518060\\ Guangdong\\ P. R. China}
\email{hanyong@szu.edu.cn}
\author[Chi Nguyen]{Chi Nguyen}
\address{Chi Nguyen \\ Ho Chi Minh University of Transport \\Department of Information Technology \\ 2 Vo Oanh St.\\ Ward 25\\ Binh Thanh District\\ Ho Chi Minh City\\ Vietnam}
\email{ntpchi@gmail.com}  
\author[Michel Zinsmeister]{Michel Zinsmeister}
   \address{Michel Zinsmeister\\Institut Denis Poisson\\ Universit\'e d'Orl\'eans\\ B\^atiment de math\'ematiques\\ rue de Chartres B.P. 6759-F-45067 Orl\'eans Cedex 2, France}
     \email{zins@univ-orleans.fr}        
\begin{abstract} We present new exact results for the complex generalized integral means spectrum (in the sense of \cite{DHLZ}) for two kinds of whole-plane Loewner evolutions driven by a L\'{e}vy process:
\begin{itemize}
\item[(1)] The case of a L\'{e}vy process with continuous trajectories, which corresponds to Schramm-Loewner evolution SLE$_\kappa$ with a \emph{drift}  term in the Brownian driving function. There is no known result for its standard integral means spectrum, and we show that a natural path to access it goes through the introduction of the complex generalized integral means spectrum, which is obtained via the so-called Liouville quantum gravity.
\item[(2)] The case of symmetric L\'{e}vy processes for which we generalize results by Loutsenko and Yermolayeva (\cite{LY4,LY1,LY2,LY3}).
\end{itemize}
\end{abstract}
\maketitle
\centerline{\emph{Dedicated to the Memory of Krzysztof Gawedzki}}
\section{Introduction}
{\sl ``Il apparut que, entre deux v\'erit\'es du domaine r\'eel, le chemin le plus facile et le plus court passe bien souvent par le domaine complexe.''[It came to appear that, between two truths of the real domain, the easiest and shortest path quite often passes through the complex domain] (Paul Painlev\'e, 1900) \cite{Painleve1900}.}\\

More than two decades ago, Oded Schramm \cite{OS} introduced his celebrated theory of random growth processes SLE$_\kappa$. As an example, in the so-called chordal case in the half-plane $\mathbb H$, it consists of the one-parameter family of Loewner processes driven on the real line $\partial \mathbb H$ by $\sqrt{\kappa}B_t$, where $\kappa$ is a nonnegative number and $B_t$ is  standard one-dimensional Brownian motion. This is the unique family of random processes satisfying a certain Markov property with continuous driving function, that is symmetric with respect to the imaginary axis. This theory may be generalized along two directions:
\begin{enumerate}
\item One can drop symmetry with respect to the imaginary axis: one then considers SLE$_\kappa$ with a \emph{drift} term, e.g., the chordal Loewner process driven by a random function of the form
$$\lambda(t)=\sqrt{\kappa}B_t+at,$$
where $B_t$ is as before standard one-dimensional Brownian motion and $a\in \R$.
\item One can drop the continuity assumption while keeping symmetry: the process so obtained is Loewner evolution driven by a L\'evy process, called LLE (for L\'evy-Loewner evolution).
\end{enumerate}
Notice that the first class of continuous drifted processes coincides with the whole class of LLE processes with continuous trajectories. For $\kappa=0$, the Loewner process generated by $\lambda(t)=at$ becomes deterministic. Several deterministic chordal Loewner processes, driven by Lip-$1/2$ functions, were investigated in \cite{zbMATH02085869,zbMATH02199893,lind2005,lind2010}.

In this paper, we shall consider both extended classes in the \emph{whole-plane} case. In order to understand the multifractal spectra of these processes, such as their integral means spectra (\emph{ims}), and in the spirit of references  \cite{DNNZ,DHLZ,ho:tel-01581324,LY4,LY1,LY2,LY3}, we shall first investigate the cases for which the expected complex moments,  
$$\mathbb{E}\left[\vert f'(z)^p\vert \left \vert\left(\frac{z}{f(z)}\right)^q\right \vert\right],\,\, p,q \in \mathbb C,$$
 may be computed explicitly, to become  part of \emph{integrable probability}. Here $f$ stands for the time $0$  whole-plane map from $\mathbb D$ to the slit plane in the corresponding Loewner process. Note that \emph{complex} values of $(p,q)$ are considered here in the case of whole-plane SLE with drift. In agreement with the citation by P. Painlev\'e above,  the suggested passage by the complex plane will help us discover the precise form of the associated integral means spectrum in the case of SLE$_\kappa$ with drift, via its complex and generalized versions \cite{DHLZ}. 
 
In Section 2, we shall make use of the so-called {\it Liouville quantum gravity} and {\it Coulomb gas}  techniques, in the spirit of \cite{Duplantier00,PhysRevLett.89.264101,MR2112128,DMS14}, to (non-rigorously) derive the full complex generalized integral means spectrum $\beta_1(p,q;\kappa,a)$ of whole-plane SLE$_\kappa$ with drift $a$ and for $(p,q)\in \mathbb C^2$. Section 3 covers SLE integrable cases, which are rigorously solved on a two-dimensional sub-manifold of $\mathbb C^2$ which generalizes the integrable parabolae of \cite{DHLZ} and \cite{ho:tel-01581324},  and successfully compared with the previous claims. For the generalized spectrum of LLE processes studied in Section 4, we shall concentrate on $(p=2,q\in \mathbb R)$ integrable cases, which can be solved analytically by closing some recursions between Fourier modes, in the spirit of \cite{LY3}. 
 
The remainder of the present introductory Section 1 is devoted to providing precise definitions, and as a warm-up, to computing the complex generalized spectrum of the logarithmic spiral, in the first $(p,q)\in \mathbb C^2$, $\kappa=0,a\neq 0$ non-trivial case.
 
\subsection{Interior whole-plane SLE}
SLE is a particular case of a growth process called the Loewner process, of which several variants exist, known as {\it chordal}, {\it radial}, {\it dipolar}, or {\it whole-plane} \cite{lawler,beliaev2019}. In this work we will consider the {\it interior whole-plane} case, which is determined by a {\it driving function} $\lambda:[0,+\infty)\to \partial \mathbb D:=\{z\in \mathbb C:\vert z\vert=1\}$ obtained as follows.
Let us start by defining $\gamma :\,[0,+\infty)\to \mathbb C$ to be a 
continuous 
 function such that $\lim_{t\to+\infty} \vert \gamma (t)\vert=+\infty$ and $\gamma (t)\neq 0,\forall t\geq 0$. 
Then, for each $t>0$, the slit domain $\Omega_t=\mathbb C\backslash\gamma ([t,\infty))$ is a simply connected domain containing $0$. By the Riemann Mapping Theorem, there  exists a unique conformal map $f_t:\mathbb D\to\Omega_t$  such that $f_t(0)=0$ and $f_t'(0)>0$. By the Caratheodory convergence theorem, $f_t$ converges to $f_0$, the Riemann mapping of $\Omega_0$, as $t\to 0$. We may assume without loss of generality that $f_0'(0)=1$ and, by re-parametrizing the curve if necessary, choose the normalization $f_t'(0)=e^t$. Loewner's  theorem asserts that there exists a continuous function $\lambda$ taking values in the unit circle such that
\begin{equation}
\label{eq:Loewner}
\frac{\partial}{\partial t} f_t(z)=z\frac{\partial}{\partial z} f_t(z)\frac{\lambda(t)+z}{\lambda(t)-z},\,\,\,\,\lim\limits_{t\rightarrow+\infty} f_t(e^{-t}z)=z,\forall z \in \mathbb{D}.
\end{equation}
The Loewner method can be reversed: given a continuous function $\lambda:[0,+\infty)\rightarrow \partial \bb D$, the partial differential equation \eqref{eq:Loewner} has a unique solution $f_t(z)$, which is a conformal map from $\mathbb{D}$ onto a domain $\Omega_t$, and the corresponding family $(\Omega_t)_t$ is increasing in $t$. Nevertheless the domains $\Omega_t$ need not be slit domains as in the example above.

Whole-plane SLE$_\kappa$ is the process driven by  
$$\lambda(t)=e^{i\sqrt{\kappa}B_t},$$
where $\kappa\in[0,+\infty)$ and $B_t$ is standard one-dimensional Brownian motion. Note that when $\kappa=0$, $f_t(z)=\frac{e^tz}{(1-z)^2}$ is the solution to \eqref{eq:Loewner}, so that $f_0$ is the Koebe function.  Thus, as $\kappa \to 0^+$, whole-plane SLE$_\kappa$ may be seen as a stochastic perturbation of the Koebe map. 

In this work, we generalize SLE by adding a drift term to Brownian motion, with a driving function defined as 
\begin{equation}\label{driving}\lambda(t):=e^{i(\sqrt{\kappa}B_t+at)}, a\in \mathbb R.
\end{equation}
The process driven by $\lambda(t)$ then appears for small $\kappa$ as a stochastic perturbation of the $(\kappa=0,\,\, a\neq 0)$ case of the \emph{logarithmic spiral}. 
\subsection{Complex generalized integral means spectrum}
Let $f$ be a conformal map from $\mathbb{D}$ to $\mathbb{C}$ with $f(0)=0,f'(0)=1$. The generalized integral means spectrum of $f$ was originally defined in \cite{DHLZ} as follows: for any pair of real numbers $(p,q)$, define the integral moments, for $r\in [0,1)$,
\begin{equation}\label{moments}
M_f(p,q):=\int_{0}^{2\pi}r^q\frac{\vert f'(re^{i\theta})\vert^p}{\vert f(re^{i\theta})\vert^q}d\theta,\,\,r\in [0,1).
\end{equation}
The {\it generalized integral means spectrum is then defined as 
$$\beta_f(p,q):=\limsup_{r\to 1^{-}}\left[\log M_f(p,q)/\log \left((1-r)^{-1}\right)\right].$$
If the limit exists, then 
\begin{equation}M_f(p,q)\stackrel{\cdot}{\sim}\left(1-r\right)^{-\beta_f(p,q)}, \,\,\,r\to 1^-,
\label{defbeta}
\end{equation}
where the notation `$\stackrel{\cdot}{\sim}$' between two quantities stands for the equivalence of the logarithms of these  quantities (i.e., their ratio tends to 1) \cite{DHLZ}.} 

One recovers for $q=0$ the \emph{standard integral means spectrum}, $\beta_f(p):=\beta_f(p,q=0)$, which is related by various Legendre transformations to the so-called \emph{multifractal spectra} \cite{Mand,HentschelP,FP,1986PhRvA..33.1141H,PhysRevA.34.1601}, like those governing the moments of the harmonic measure or the continuum of its local singularities \cite{Makanaliz,MR2450237}. 

 For a random simply connected domain as arising from a whole-plane Loewner process with a random driving function like SLE, the question whether the equivalence \eqref{defbeta}
  holds \emph{almost surely} is notoriously difficult. Earlier works dealt with the `expected spectrum' for Brownian motion \cite{0965.60071,1999PhRvL..82..880D}, self-avoiding walk \cite{1999PhRvL..82..880D}, percolation \cite{1999PhRvL..82.3940D,2008PhRvL.101n4102A}, and SLE \cite{Duplantier00,Duplantier03,PhysRevLett.88.055506,MR2112128,DupLH,PhysRevLett.95.170602,2007JPhA...40.2165R,BS,DNNZ,MR3638311} as well as with the expected generalized spectrum of whole-plane SLE \cite{DHLZ}. The almost sure case was solved only recently for the standard spectrum of chordal SLE by Gwynne, Miller and Sun \cite{gwynne2018} by using the so-called {\it imaginary geometry} of Miller and Sheffield. (See also the earlier works  \cite{0911.3983} for the SLE a.s.~tip spectrum,  \cite{Alberts2016} for the SLE a.s.~boundary spectrum, and the recent work \cite{Schoug20} for the SLE$_\kappa(\rho)$ a.s.~boundary spectrum, where imaginary geometry was also used.) 
  
The case of complex moments corresponds to the mixed multifractal spectrum of the harmonic measure and logarithmic rotations of the conformal map \cite{dis}. It was studied in expectation in Refs. \cite{PhysRevLett.89.264101,2008NuPhB.802..494D,1751-8121-41-28-285006} for the chordal and radial SLE cases. We shall consider here the whole-plane spectrum defined in expectation for complex moments,
 \begin{equation}\label{eq:cgims}
\int_{r\partial\mathbb{D}}\mathbb{E}\left[\vert f'(z)^p\vert\left \vert\left(\frac{z}{f(z)}\right)^q\right \vert\right] |dz|
\stackrel{\cdot}{\sim} \left(\frac{1}{1-r}\right)^{\beta(p,q)},\,\,\,p,q\in \mathbb C.
\end{equation}
It is then natural to introduce the one-point function
 \begin{equation}\label{eq:G}
 G(z):=\mathbb{E}\left[\vert f'(z)^p\vert\left \vert\left(\frac{z}{f(z)}\right)^q\right \vert\right],\,\,\,p,q\in \mathbb C.
 \end{equation}
The setting chosen in \eqref{eq:cgims} and \eqref{eq:G} allows for complex values $p,q \in \mathbb C$,  which we shall need to study the drift case. In the more general case of L\'evy processes,  we shall see that their defining properties  are exactly those needed to obtain a PDE satisfied by $G$ \eqref{eq:G}, as initiated in Refs. \cite{PhysRevLett.88.055506,BS} and further developed in Refs. \cite{DNNZ,MR3638311,DHLZ} and \cite{LY4,LY1,LY2,LY3}.

\subsection{Interior-exterior duality}
As mentioned in \cite{DHLZ}, it is interesting to remark that the map $\widehat f$, $$\zeta\in\mathbb C\setminus \overline{\mathbb D} \mapsto \widehat f (\zeta):=1/f(1/\zeta),$$  is just the {\it exterior} whole-plane map from $\mathbb C\setminus\overline{\mathbb D}$ to the slit plane considered in Ref. \cite{BS} by Beliaev and Smirnov and in Ref. \cite{MR3638311}. We identically have for $0<r<1$ and $p\in \mathbb R$,
\begin{equation}\label{q2p}
\int_{r^{-1}\partial\mathbb D} \mathbb E\left(\vert \widehat f'(\zeta)\vert^p \right)|d\zeta|=r^{2p-2}\int_{r\partial\mathbb D} \mathbb E\left(\frac{\vert f'(z)\vert^p}{\vert f(z)\vert^{2p}}\right)|dz|. 
\end{equation} 
We thus see that the standard integral mean of order $(p,q=0)$ for the exterior whole-plane map studied in  \cite{BS,MR3638311} {\it coincides} (up to an irrelevant power of $r$) with the $(p,q')$ integral mean  for $q'=2p$, for the interior whole-plane map.   
\begin{rkk}\label{qq'duality} \emph{Interior-Exterior Duality}. By conformal inversion, we have for any $p,q\in \mathbb C$,
\begin{equation}\label{qq'}
\int_{r^{-1}\partial\mathbb D} \mathbb E\left(\left\vert \frac{\widehat f'(\zeta)^p}{\widehat  f(\zeta)^{q}}\right\vert \right)|d\zeta|=r^{2\Re p-2}\int_{r\partial\mathbb D} \mathbb E\left(\left\vert\frac{f'(z)^p}{f(z)^{2p-q}}\right\vert \right)|dz|, 
\end{equation} 
so that the $(p,q)$ exterior integral means spectrum coincides with the $(p,q')$ interior integral means spectrum for $q'-p=p-q$. In particular, the $(p\in \mathbb R,q'=0)$ interior derivative moments studied in Ref. \cite{DNNZ} correspond to the $(p, q=2p)$ mixed moments of the exterior map. 
\end{rkk}
\subsection{Generalized spectrum for the logarithmic spiral}
In this section we give an example of a generalized integral means spectrum, which is deterministic and corresponds to the $\kappa=0$ case of the drifted SLE$_\kappa$. It is nothing but the logarithmic spiral with parameter $a\in \mathbb{R}$ (Fig. \ref{fig:spiral}), i.e., the curve parametrized by 
\begin{equation}\label{spiral}
\gamma(t)=\exp\left[(1+ia)t\right],\,\,\,t\in \mathbb R.
\end{equation}
\begin{figure}[htbp!]
\begin{center}
\includegraphics[width=7cm]{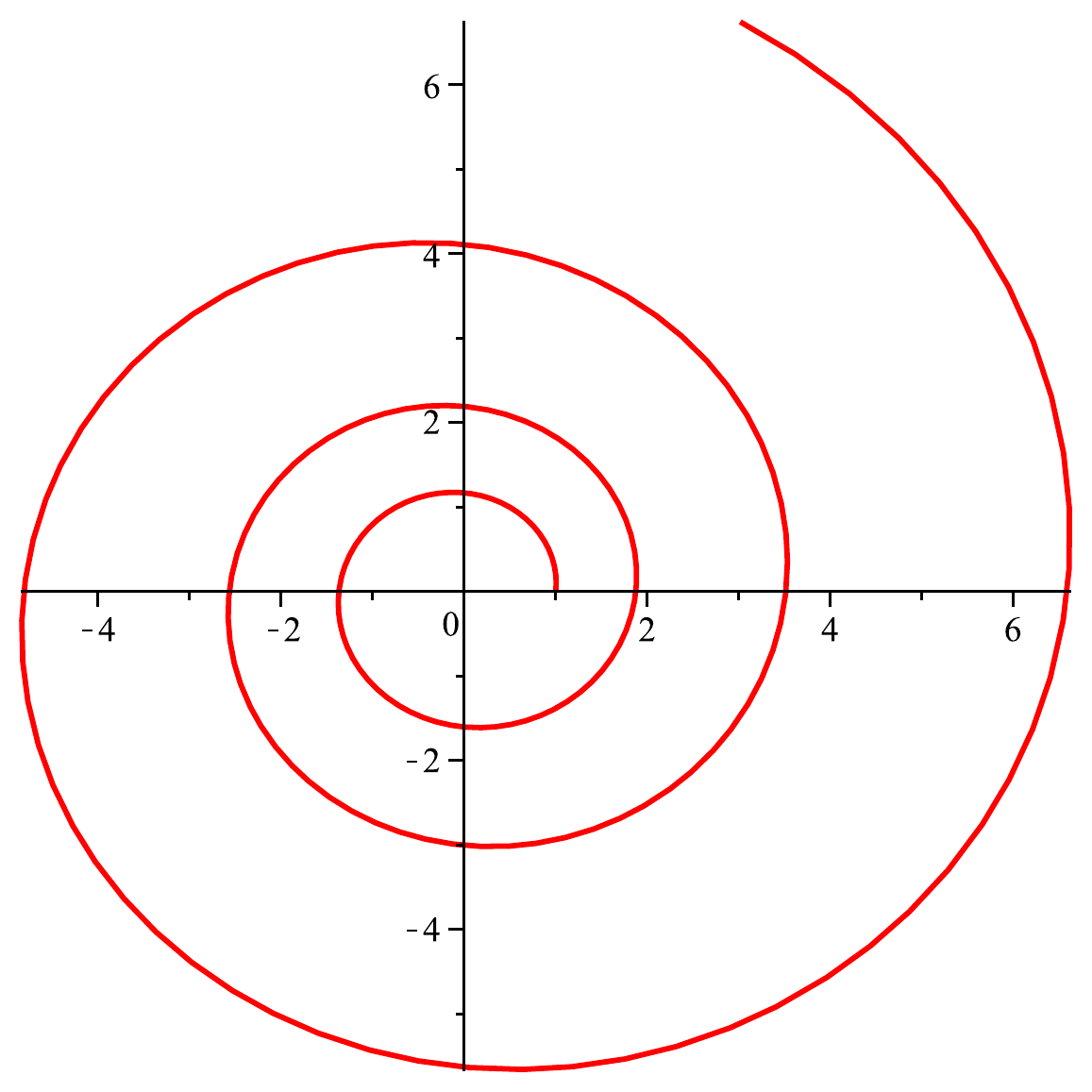}
\end{center}
\caption{Logarithmic spiral $\gamma(t)$ \eqref{spiral} for $a=+5$, restricted to $t\geq 0$.}
\label{fig:spiral}
\end{figure}
\subsubsection{Loewner process for the logarithmic spiral}
Let us define, as before, $\Omega_t:=\mathbb{C}\setminus \gamma[t,+\infty)$  and let $f_t:\mathbb{D}\rightarrow \Omega_t$  be the associated Riemann map, i.e., the conformal map such that 
\begin{equation}
f_t(0)=0,f_t'(0)>0.
\end{equation}
By the Koebe distorsion theorem, $\lim_{t\rightarrow-\infty}f_t'(0)=0$ and $\lim_{t\rightarrow+\infty}f_t'(0)=+\infty$. Then, there exists $t_0$ such that $f_{t_0}'(0)=1$. One also has that $f_{t_0}(e^{i\theta_0})=\gamma(t_0)=e^{(1+ia)t_0}$ for some $\theta_0\in[0,2\pi).$ Consider now the function $\tilde f_t$ defined by $$\tilde{f}_t(z):=e^{(1+ia)t}f_{t_0}(e^{-iat}z).$$ We have
 $\tilde{f}_{t}(0)=0, \tilde{f}_{t}'(0)=e^t,$ and 
 \begin{equation} \label{tildef-gamma}
 \tilde{f}_t(e^{i(\theta_0+at)})=e^{(1+ia)(t+t_0)}=\gamma(t+t_0).
 \end{equation}
 Hence from \eqref{tildef-gamma}, $\tilde{f}_{t}:\mathbb{D}\rightarrow \Omega_{t+t_0}$ is the Loewner process corresponding  to the curve  
$$
\tilde{\gamma}(t):=\gamma(t+t_0),\,\,t\in\mathbb{R},
$$
with  the associated driving function $\tilde\lambda(t):=e^{i(\theta_0+at)}$.

Define then the curve, $$\eta(t):=e^{-i\theta_0}\tilde{\gamma}(t),\,\,\,t\in\mathbb{R},$$ and the conformal map,   
$$h_t(z):=e^{-i\theta_0}\tilde{f}_t(e^{i\theta_0}z).$$ One still has  $h_{t}(0)=0, h_{t}'(0)=e^t,$ and $h_t(e^{iat})=\eta(t)$, so that $h_t$ is the Loewner map corresponding to $\eta(t)$ and the associated process is driven by $\lambda(t):=e^{iat}$.\\
Notice that the curve $\eta$ is obtained by a time-translation and a rotation of the logarithmic spiral $\gamma$. Thus the integral means spectrum will be the same for the time zero Loewner maps $h_0$ and $\tilde f_0$.

\subsubsection{Complex generalized spectrum for the complete logarithmic spiral} We first focus on the \emph{complete} spiral, for which we first establish the following theorem.
\begin{theo}\label{lem:spiralc}
The complex generalized integral means spectrum of  the \emph{complete} logarithmic spiral $\gamma(t)=e^{(1+ia)t},\,t\in \mathbb R$,  is given, for $p,q\in \mathbb C$, by
\begin{eqnarray}\label{betaspiralcomp}  \beta(p,q;a)=
\sup\left\{0, \beta_1(p,q;0,a),\beta_2(p,q;0,a)\right\},            
\end{eqnarray}
  where 
 \begin{equation}\begin{split}\label{beta1spiralc}
& \beta_1(p,q;0,a):=2\Re{\left(\frac{p-q}{1-ia}\right)}+\Re{p}-1,\\ 
 &\beta_2(p,q;0,a):=-2\Re{\left(\frac{p-q}{1-ia}\right)}+\Re{p}-1.
 \end{split}
  \end{equation}
 \end{theo}
\begin{proof}
Let us define on the unit disk $\mathbb D$, the Moebius map $\xi: z\mapsto \xi(z):=i\frac{1-z}{1+z}$, and consider the function $\Phi$  defined on $\mathbb D$ as, 
$$
\Phi(z):=\exp\left[\frac{2}{1-ia}\log \xi(z)\right]=\left(i\frac{1-z}{1+z}\right)^{\frac{2(1+ia)}{1+a^2}},\,\,\, z\in \mathbb D.
$$
Define also the strip domain ${\mathbb S}_{\pi}:=\{x+iy: x\in \mathbb R,\,\, 0<y<\pi\}$. We know that $z\mapsto \xi(z)$  conformally maps $\mathbb D$ onto upper half-plane $\mathbb H$, 
while $z\mapsto \log(z)$  conformally maps $\mathbb H$ onto the strip ${\mathbb S}_\pi$. Lastly, $z\mapsto \exp\left(\frac{2}{1-ia}z\right)$ conformally maps the strip domain ${\mathbb S}_{\pi}$ onto $\mathbb C\setminus \gamma$, with a cut along the whole logarithmic spiral $\gamma:=\{\gamma(t)=e^{(1+ia)t}, t\in \R\}$. 
Consequently, $\Phi$ is a conformal map from the unit disk $\mathbb D$ to the complement of the \emph{whole} logarithmic spiral $\gamma$, with $\Phi(1)=0, \Phi(-1)=\infty$. 

It enjoys the useful property,
\begin{equation}\label{eq:derivPhi}
\Phi'(z)=\frac{2}{1-ia}\big(\log \xi(z)\big)'\Phi(z)=-\frac{4}{1-ia}\frac{\Phi(z)}{1-z^2}.
\end{equation}
Owing to \eqref{eq:derivPhi},the complex mixed moments of $\Phi$ read
\begin{equation}\label{eq:moments}
  \frac{\Phi'(z)^p}{\Phi(z)^q}=\left(-\frac{4}{1-ia}\frac{1}{1-z^2}\right)^p\, \frac{1}{\Phi(z)^{q-p}}, 
  \end{equation}
so that 
   \begin{equation}\label{eq:modmoments}
  \left\vert\frac{\Phi'(z)^p}{\Phi(z)^q}\right\vert= \left\vert\left(-\frac{4}{1-ia}\right)^p\right\vert\times \frac{\left\vert\Phi(z)^{p-q}\right\vert}{\left\vert(1-z^2)^p\right\vert}. 
  \end{equation}
We have explicitly 
  \begin{equation}\label{eq:1-z2}
\left\vert(1-z^2)^p\right\vert=\left\vert1-z\right\vert^{\Re p}\left\vert1+z\right\vert^{\Re p} e^{-\Im p\,\arg(1-z^2)},
 \end{equation}
 and 
 \begin{equation}\label{phipq}
 \left\vert\Phi(z)^{p-q}\right\vert=\exp\Re\left[\frac{2(p-q)}{1-ia}\log \xi(z)\right].
  \end{equation}
 Setting $b=b(p,q):=\frac{2(p-q)}{1-ia}$, we have $\Re\left[b\log \xi(z)\right]=\Re b\, \log |\xi(z)|-\Im b\, \Im \log \xi(z)$, and since $\log \xi(z)\in {\mathbb S}_\pi$, its imaginary part stays \emph{bounded}. 
 We thus have the following (logarithmic) equivalence near the two possible singular points $z=\pm1$,
  \begin{equation}\label{eq:phipqlog}
 \left\vert\Phi(z)^{p-q}\right\vert \stackrel{\cdot}{\sim} |\xi(z)|^{\Re b(p,q)}= \left\vert\frac{1-z}{1+z}\right\vert^{\Re \frac{2(p-q)}{1-ia}}.
  \end{equation}
 Using \eqref{eq:modmoments}, \eqref{eq:1-z2}, and \eqref{eq:phipqlog}, we finally arrive at 
   \begin{equation}\label{eq:asymp}
  \left\vert\frac{\Phi'(z)^p}{\Phi(z)^q}\right\vert \stackrel{\cdot}{\sim} \left\vert\frac{1-z}{1+z}\right\vert^{\Re \frac{2(p-q)}{1-ia}} \left\vert1-z\right\vert^{-\Re p}\left\vert1+z\right\vert^{-\Re p}.
  \end{equation}
\emph{Behaviour near infinity and near the origin.}  For $z=re^{i\theta}$ near $z=-1$ (point at $\infty$ on the spiral), $|1+z|^2= r^2+2r\cos\theta+1$ behaves like $(1-r)^2+(\pi-\theta)^{2}$. Similarly, near $z=+1$ (point $0$ on the spiral),
 $|1-z|^2$ behaves like $(1-r)^2+\theta^{2}$. The integral of \eqref{eq:asymp} along the circle $|z|=r$ for $r\to 1^-$ 
 is thus dominated near $z=-1$  by the contribution of the angular neighbourhood of $\theta=\pi$, while  near $z=+1$ it is symmetrically dominated by that  of the angular neighbourhood of $\theta=0$. From the explicit form of the integrand \eqref{eq:asymp}, we readily obtain  the overall asymptotic behaviour as $r\to 1^-$ of the integral means,
    \begin{equation}\label{eq:integral}
 \int_{r\partial \mathbb D} \left\vert\frac{\Phi'(z)^p}{\Phi(z)^q}\right\vert |dz| \stackrel{\cdot}{\sim} \left(1-r\right)^{-\beta_{\Phi}(p,q;a)},\,\,\, r\to 1^-,
  \end{equation}
where the integral means spectrum $\beta_{\Phi}$ is given by the largest exponent, 
  \begin{eqnarray}\label{max12}
\beta_{\Phi}(p,q;a):=\beta_1(p,q;0,a)\vee \beta_2(p,q;0,a)\vee 0,
\end{eqnarray}  with the two dual spectra defined as,
\begin{eqnarray}
  \label{beta1.0c}
 && \beta_1(p,q;0,a):=2\Re \frac{p-q}{1-ia}+\Re p-1,\\ \label{beta2.0c}
 && \beta_2(p,q;0,a):=2\Re \frac{q-p}{1-ia}+\Re p-1.
 \end{eqnarray}
 \end{proof}
\begin{rkk} \label{rk:inf-tip} \emph{Singularity localization.}
Exponent $\beta_1$ is associated with the singularity near $z=-1$ on $\mathbb D$ in \eqref{eq:asymp}, i.e., at \emph{infinity} on the spiral, while $\beta_2$ corresponds to that near $z=+1$, i.e., 
near the tip at origin $0$, around which the spiral \emph{indefinitely winds}. 
\end{rkk}
\begin{rkk} \emph{Conformal invariance by inversion and duality.}
The full logarithmic spiral is conformally invariant under the complex inversion, $z\mapsto 1/z$, since $1/\gamma(t)=\gamma(-t)$, and $t\in \mathbb R$.  This inversion exchanges the roles of origin and infinity, and maps 
the interior of $\mathbb D$ to its exterior. The complex generalized integral means spectrum then obeys the \emph{duality property} \eqref{qq'duality}.  Spectra  \eqref{beta1.0c} and \eqref{beta2.0c} are indeed dual of each other under the corresponding exchange $q-p\mapsto p-q$, resulting in the expected invariance under duality of the integral means spectrum $\beta_{\Phi}$ \eqref{max12} for the complete logarithmic spiral.  
\end{rkk}   
\subsubsection{Complex generalized spectrum of the half spiral} Consider now $h_0(z)$, the conformal map corresponding to the whole-plane Loewner process driven by 
$e^{iat}$, stopped at time $t=0$, the image of which, $\gamma(t)=e^{(1+ia)t},\, t\geq 0$,  we may call the \emph{half spiral} (Fig. \ref{fig:spiral}). The complex generalized integral means spectrum of the half logarithmic spiral is given by the following theorem. (See Fig. \ref{fig:spiralp-q}.)
\begin{theo}\label{theo:spiralc} The \emph{complex} generalized integral means spectrum of $h_0$, where $h_t$ is the whole-plane Loewner process driven by $\lambda(t)=e^{iat}$, and whose trace is 
the \emph{half} logarithmic spiral $\gamma(t)=e^{(1+ia)t},\, t\geq 0$,  is given, for $p,q\in \mathbb C$, by
 \begin{eqnarray}\label{betaspiral}  \beta(p,q;\kappa=0,a)=
  \sup\left\{-\Re{p}-1,0, \beta_1(p,q;0,a)=2\Re{\left(\frac{p-q}{1-ia}\right)}+\Re{p}-1\right\}.            
  \end{eqnarray}
  \end{theo}
  From this, one immediately deduces the following corollary, which yields the \emph{real} generalized integral means spectrum of the half spiral. 
\begin{cor}\label{theo:spiral} The \emph{real} generalized integral means spectrum of $h_0$, where $h_t$ is the whole-plane Loewner process driven by $\lambda(t)=e^{iat}$, and whose trace is 
the \emph{half} logarithmic spiral $\gamma(t)=e^{(1+ia)t},\, t\geq 0$,  is given, for $p,q\in \mathbb R$, by
 \begin{eqnarray}\label{betaspiralreal}  \beta(p,q;\kappa=0,a)=
  \sup\left\{-p-1,0, \beta_1(p,q;0,a)=2\frac{p-q}{1+a^2}+p-1\right\}.            
 \end{eqnarray}
 \end{cor}
This result for the real case, $p,q\in \mathbb R$,  is illustrated in Fig. \ref{fig:spiralp-q}.
\begin{figure}[tb]
\begin{center}
\includegraphics[width=9cm]{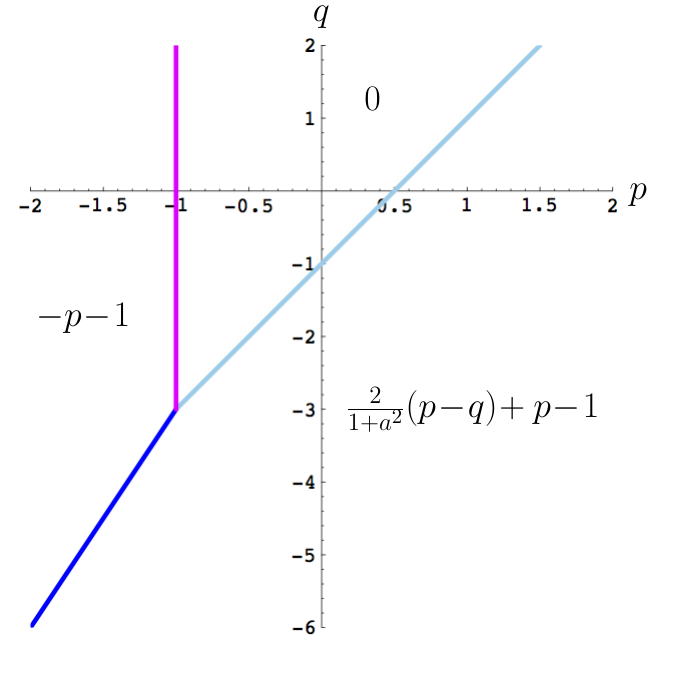}
\end{center}
\caption{The three phases of the generalized integral means spectrum of the logarithmic spiral, with $\beta_{\mathrm{tip}}(p;\kappa=0)=-p-1$, $\beta_0(p;\kappa=0)=0$, 
$\beta_1(p,q;\kappa=0,a)=\frac{2}{1+a^2}(p-q)+p-1$ (Corollary \ref{theo:spiral}).}
\label{fig:spiralp-q}
\end{figure}
\begin{proof}  
{\it $\bullet$ Behaviour near infinity.} For $t\geq 0$,  the half spiral and whole spiral are identical, thus 
have the same spectrum near infinity. So we use the conformal map $\Phi$ to calculate the integral means spectrum near $\infty$, i.e., by considering the mixed moments \eqref{eq:asymp} for $z\to -1$ only, 
as well as the corresponding contribution to integral \eqref{eq:integral}. Because of Remark \ref{rk:inf-tip}, the associated spectrum is $\beta_1$ \eqref{beta1.0c}.\\
 {\it $\bullet$ Behaviour near the tip.} 
  Let $\phi(z):=\frac{z}{(1-z)^2}, z\in \mathbb D$, with $\phi(-1)=-\frac{1}{4},\phi(0)=0,\phi(1)=\infty$, be the \emph{Koebe function}, conformally  mapping the unit disk to the straight cut plane as $\mathbb{D}\mapsto \mathbb{C}\setminus (-\infty,-\frac{1}{4}]$. Let $g$ be the conformal map from $\mathbb{C}\setminus (-\infty,-\frac{1}{4}]$ to the plane cut by the half spiral, $\Omega_0:=\mathbb{C}\setminus \gamma[0,\infty)$, with $g(0)=0,g(-\frac{1}{4})=\gamma(0)=1$. Then $h_0=g\circ \phi$. Notice that both $g$ and $g'$  are bounded near $\phi(-1)=-\frac{1}{4}$, hence also $h_0$ near $z=-1$.  
  Let us define  $r \partial \mathbb D_\varepsilon:= \{z: |z|=r, |1+z|< \varepsilon\}$, for some fixed $\varepsilon$ such that $1-r <\varepsilon<1$, as the neighbourhood along the circle $r\partial \mathbb D$ of
   the pre-image $z=-1$ by $h_0$ of the half spiral tip $\gamma(0)=1$. In this domain, we have the logarithmic equivalence, as $r \to 1^-$,  
$$\int_{r \partial \mathbb D_\varepsilon} \left\vert\frac{h_0'(z)^p}{h_0(z)^q}\right\vert|dz| \stackrel{\cdot}{\sim} \int_{r \partial \mathbb D_\varepsilon}|\phi'(z)^p| |dz|,\,\,\,r \to 1^-.$$  We thus obtain that the integral means spectrum near the tip of the half spiral is the same as the \emph{ims}  near the tip of the half line, which is simply, 
  \begin{equation}\label{eq:betatip0}
  \beta_{\mathrm{tip}}(p;\kappa=0):=-\Re p-1.
  \end{equation} 
{\it $\bullet$ Bulk behaviour.} Away from $\infty$ and the tip, the half spiral is rectifiable, and its bulk integral means spectrum is trivial, 
$\beta_0(p;\kappa=0)=0$. 
This ends the proof of Theorem \ref{theo:spiralc}.
\end{proof}

\section{Complex generalized spectrum of drifted whole-plane SLE}
\subsection{Introduction}In this section, we will predict the exact form of the generalized integral means spectrum $\beta_1(p,q;\kappa,a)$ associated with the whole-plane SLE$_\kappa$ with drift $a$. As we shall see, its most symmetric and simplest form is obtained for the complex generalized spectrum  where the exponents are complex variables $p,q\in \mathbb C$. We shall use a non-fully rigorous method inherited from theoretical physics. More specifically, we use two-dimensional quantum gravity where the Euclidean Lebesgue measure is replaced by the Liouville quantum measure. This allows us to compute multifractal exponents in Liouville quantum gravity (LQG) for $p\in \mathbb C$ in the $q=0$ case, and for any $a\in \mathbb R$.  The conversion to the complex multifractal spectrum in the Euclidean plane is then obtained by using the celebrated Knizhnik-Polyakov-Zamolodchikov (KPZ) relation \cite{MR947880,MR1005268,MR981529,2009arXiv0901.0277D,springerlink:10.1007/s00222-010-0308-1bis,rhodes-2008,PhysRevLett.107.131305,DMS14}. The final step to get the complex generalized spectrum for $q\neq 0$ is then obtained via the introduction of the \emph{packing spectrum, 
\begin{equation}\label{eq:packing}
s_1(p,q;\kappa,a):=\beta_1(p,q;\kappa,a)-\Re\, p+1,
\end{equation}
 together with the fact that it is a function of variable $p-q$ only.}
 
\begin{figure}[htbp!]
\begin{center}
\includegraphics[width=7.7329cm]{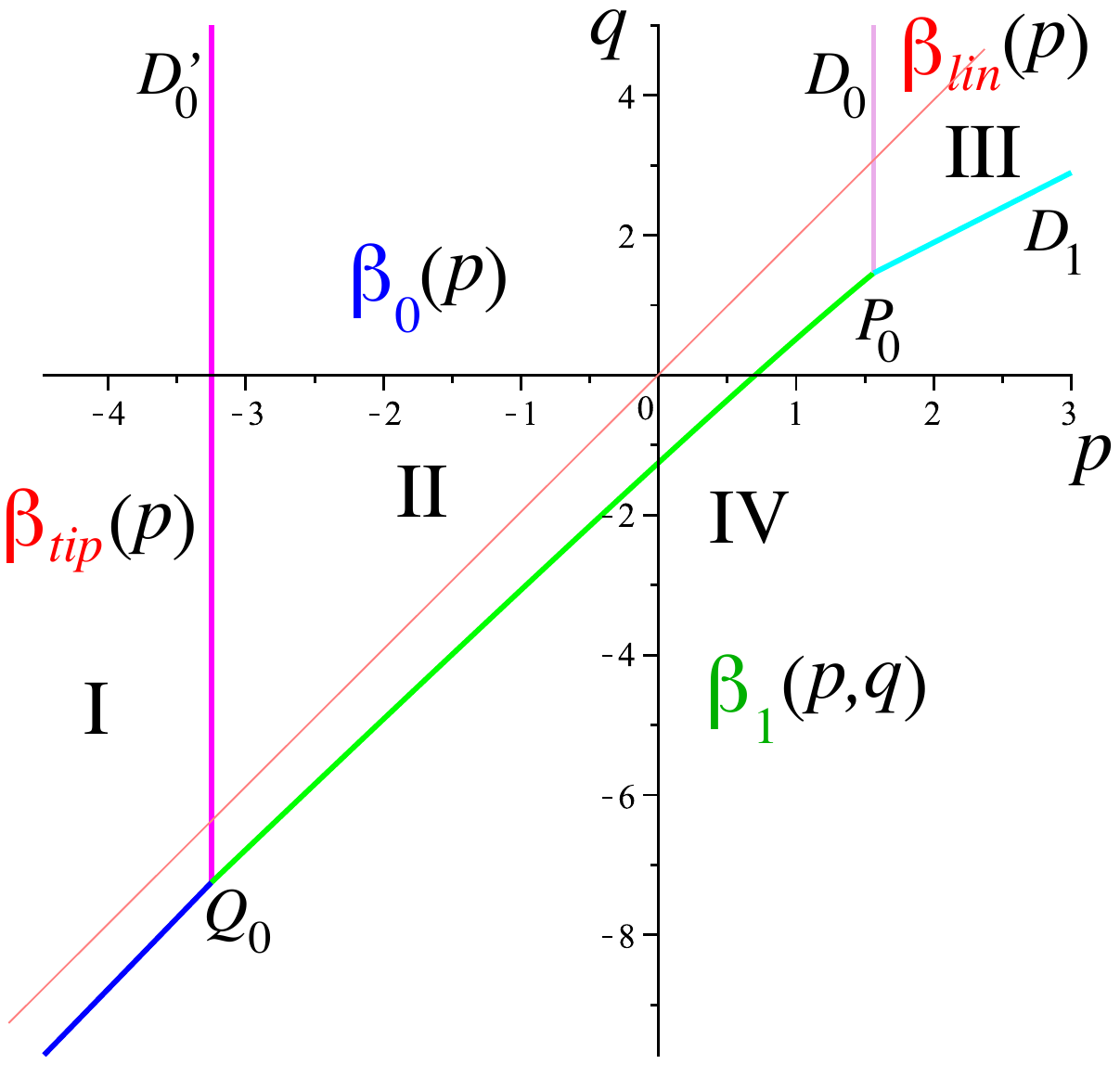}
\end{center}
\caption{Phase transition lines for the generalized integral means spectrum of whole-plane SLE$_\kappa$ with no drift $a=0$. The standard \emph{ims} of the $q=2p$ exterior version crosses phases I, II, III only, while the $q=0$ standard \emph{ims} of the interior version crosses phases  I, II, IV (from Ref. \cite{DHLZ}).}
\label{fig:separatrix}
\end{figure}
\subsection{Driftless and real case} Let us denote by $\beta(p,q;\kappa,a)$ the generalized integral means spectrum of the whole-plane SLE$_\kappa$ with drift coefficient $a$. 
Ref. \cite{DHLZ} studied the $a=0$ case and for $(p,q)\in \mathbb R^2$, for which it is shown that $\beta(p,q;\kappa,a=0)$ has four possible forms, of which three are independent of $q$,
  \begin{eqnarray}
 \beta_{\mathrm{tip}}(p;\kappa) &:=& -p-1+\frac 14\left(4+\kappa-\sqrt{(4+\kappa)^2-8\kappa p}\right),\\
 \beta_0(p;\kappa) &:=& -p+\frac{4+\kappa}{4\kappa}\left(4+\kappa-\sqrt{(4+\kappa)^2-8\kappa p}\right),\\
 \beta_{\mathrm{lin}}(p;\kappa) &:=& p-\frac{(4+\kappa)^2}{16\kappa},\\
 \beta_1(p,q;\kappa,a=0) &:=& p+2(p-q)-\frac 12-\frac 12\sqrt{1+2\kappa(p-q)}.
 \end{eqnarray} 
  The separatrices between the different phases are located as follows \cite[Theorem 1.7]{DHLZ} (See Fig. \ref{fig:separatrix}.) For $p\leq -1-\frac{3\kappa}{8}$ there is a (quartic) curve ending at point $Q_0: p_0'=-1-\frac{3\kappa}{8},q_0'=-2-\frac{7\kappa}{8}$, that separates the half-plane into two parts, $\beta$ being equal to $\beta_{\mathrm{tip}}$ above that curve and to $\beta_1$ below it. In the strip $-1-\frac{3\kappa}{8}\leq p\leq \frac{3(4+\kappa)^2}{32\kappa}$, there is a section of parabola joining $Q_0$ to point $P_0=(p_0,q_0)$, with
  \begin{equation}\label{eq:P0}
  p_0=\frac{3(4+\kappa)^2}{32\kappa},\,\,\,q_0=\frac{(4+\kappa)(8+\kappa)}{16\kappa},
  \end{equation} that separates the strip into two parts, an upper one where $\beta=\beta_0$ and a lower one where $\beta=\beta_1$. Finally the half-plane $p\geq p_0$ is similarly split by the half-line with unit slope starting at $P_0$ into an upper part where $\beta=\beta_{\mathrm{lin}}$, while $\beta=\beta_1$ in the lower part. It should be noticed that the generalized spectrum $\beta$ is not everywhere the maximum of the four spectra listed above \cite{DHLZ}.
  \begin{figure}[!htb]
\begin{center}
\includegraphics[width=7.8329cm]{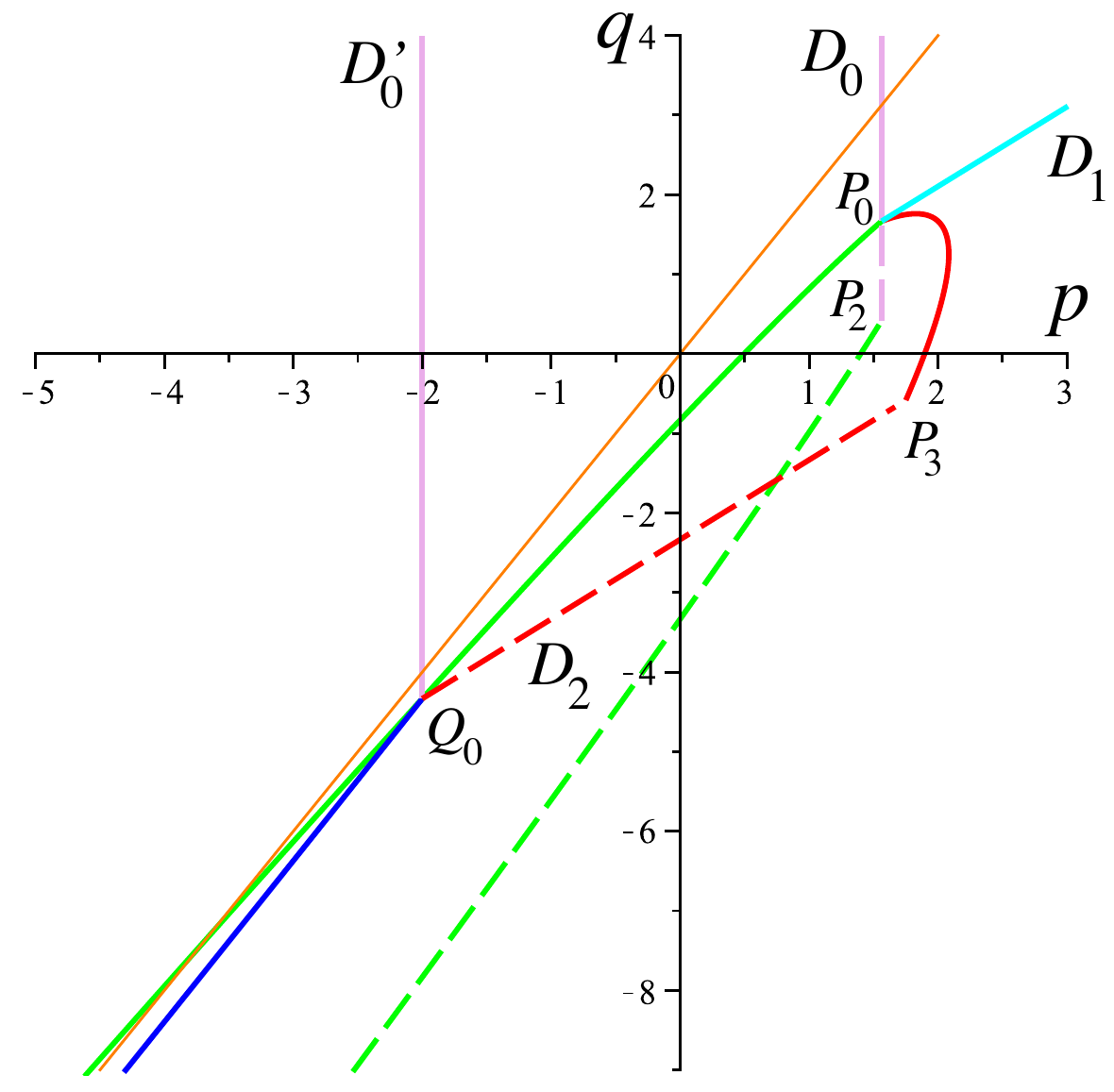}
\end{center}
\caption{Domains of validity of the proofs in the driftless case $a=0$, $\kappa=8/3$  (from  Ref. \cite{DHLZ}).}
\label{fig:raquette}
\end{figure}
The existence of these phase transition lines was established  in \cite{DHLZ} within a connected semi-infinite domain of the $(p,q)$ plane, as indicated in Fig. \ref{fig:raquette}. This domain of validity sweeps the plane from its upper-left part up to a piecewise boundary first made,  for increasing values of $p$, of the dotted green parabola up to its intersection with the straight line $D_2$ of equation $p-q=1+\frac{\kappa}{2}$. It then follows this line up to its intersection  $P_3=\left(1+\frac{2}{\kappa}, \frac{4-\kappa^2}{2\kappa}\right)$ with the red parabola. From there, the boundary is made of the section of red parabola up to point $P_0$ \eqref{eq:P0}, followed by the straight line $D_1$ of equation  $q-p=\frac{16-\kappa^2}{32 \kappa}$. These restrictions to the domain of proof are due to technicalities involved in the proofs \cite{DNNZ,MR3638311,DHLZ}, and the spectrum is supposed to be still given by $\beta_1$ in the whole connected domain located to the right of the piecewise boundary just described. Recent work by Xuan Hieu Ho extends the domain of validity to the whole interior of the red parabola \cite{XHH}. Let us now turn to the complex generalized spectrum of whole-plane SLE for $(p,q)\in \mathbb C^2$, possibly with a drift term. 

\bigskip

\subsection{Complex case with drift}
\begin{claim}\label{claim1}
For $p,q\in \mathbb C$, and $a=0$, the complex spectrum $\beta_1$ of whole-plane SLE can be obtained by  combining Liouville quantum gravity and Coulomb gas methods. It is 
\begin{align}\label{beta1complexe}
&\beta_1(p,q;\kappa,a=0)= s_1(p-q;\kappa,a=0) +
\Re p -1\\
&s_1(p-q;\kappa,a=0)=s_1(\tau):=2\tau+\frac 12-\frac 12\sqrt{1+2\kappa \tau},\\ \label{tbis}
&1+2\kappa \tau:=\frac{1}{2}\left\{1+2\kappa \Re (p-q)+\left |1+2\kappa(p-q)\right |\right\}.
\end{align}
\end{claim}
\begin{claim}\label{claim2}For $a\neq 0$, the complex spectrum $\beta_1$ of whole-plane SLE with drift is given by an  extension of the above proofs, as 
\begin{align}\label{beta1spiralecomplexe}
&\beta_1(p,q;\kappa,a)= s_1(p-q;\kappa,a) +
\Re p -1\\\label{s1comp}&s_1(p-q;\kappa,a)=s_1(\tau):=2\tau+\frac 12-\frac 12\sqrt{1+2\kappa \tau},\\ \label{tter}
&1+2\kappa \tau:=\frac{1}{2}\left\{\Re\left[(1+ia)^2+2\kappa (p-q)\right]+\left |(1+ia)^2+2\kappa(p-q)\right |\right\}.
\end{align}
\end{claim}
\begin{rkk}
In the limit $\kappa\to 0$, the integral means spectrum \eqref{beta1.0c} of the half-spiral is recovered from \eqref{beta1spiralecomplexe} \eqref{s1comp}, by observing  that the expansion to order $O(\kappa)$ of the r.h.s. of \eqref{tter} indeed yields $\tau=\Re{\left(\frac{p-q}{1-ia}\right)}$. 
\end{rkk}
As we shall see in Section \ref{Redcheck}, this complex spectrum yields the correct answer along an \emph{integrable} complex parabola in the complex space $(p,q)\in \mathbb C^2$.

In the \emph{real} moment case, $(p,q)\in \mathbb R^2$, the generalized integral means spectrum $\beta_1(p,q;\kappa,a)$ associated with whole-plane SLE$_\kappa$ with drift $a$ is given  by the explicit formulae:
\begin{eqnarray}\label{beta1spirale}
&&\beta_1(p,q;\kappa,a)= p+2\tau-\frac 12-\frac 12\sqrt{1+2\kappa \tau},\\ \label{t}
&&1+2\kappa \tau:=\frac{1}{2}\left\{1-a^2+2\kappa(p-q)+\sqrt{\left[1-a^2+2\kappa(p-q)\right]^2+4a^2}\right\}.
\end{eqnarray}
\begin{consequence}
Eq. \eqref{t} can be inverted into:
\begin{equation}\label{p-qversust}
p-q=\tau\left(1+\frac{a^2}{1+2\kappa \tau}\right).
\end{equation}
Therefore the phase transition lines in the $(p,p-q)$ plane for $a\neq 0$ are obtained from those for $a=0$ by the non-linear transform,
\begin{eqnarray}\nonumber
&&p\mapsto p,\\ \label{nltp-q} &&p-q=\tau \mapsto p-q=\tau\left(1+\frac{a^2}{1+2\kappa \tau}\right).
\end{eqnarray}
\end{consequence}
In  the work \cite{DHLZ}, the location of the various phase transition lines in the case of whole-plane SLE without drift was established with the help of several master curves: a so-called `red  parabola' where the one-point function $G$ \eqref{eq:G} is integrable, a so-called `green parabola' where the spectrum changes from $\beta_0$ to $\beta_1$, and a `blue quartic' where it changes from $\beta_{\mathrm{tip}}$ to $\beta_1$, as well as several straight lines, like $D_0'$ where the spectrum changes from $\beta_{\mathrm{tip}}$ to $\beta_0$, $D_0$ where it changes from $\beta_0$ to $\beta_{\mathrm{lin}}$, and $D_1$ where it changes from $\beta_{\mathrm{lin}}$ to $\beta_1$  (Fig. \ref{fig:separatrix}). These curves are also instrumental in delimiting the domains of validity of the proofs (Fig. \ref{fig:raquette}). Applying the non-linear transform \eqref{nltp-q} in the $(p,q-p)$ plane to these curves yields the corresponding curves in the case of whole-plane SLE with drift. They are illustrated in Figs. \ref{Fig4-5} and \ref{Fig6}.
\begin{figure}[!htb]
\begin{center}
\includegraphics[width=9cm]{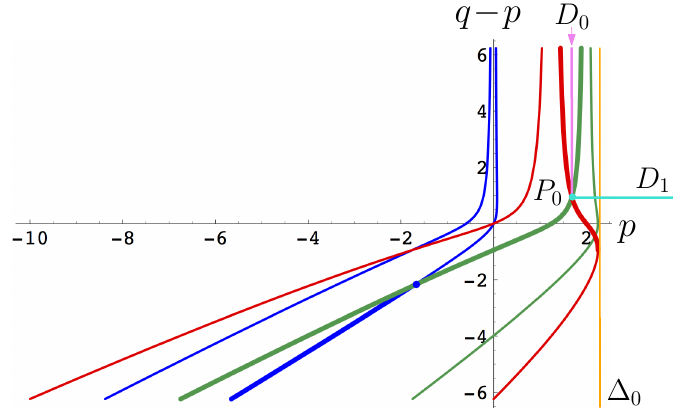}
\end{center}
\caption{Non-linear mapping \eqref{nltp-q} of the red and green parabolae and blue quartic of Ref. \cite{DHLZ} (here $\kappa=2$, $a=1$).}
\label{Fig4-5}
\end{figure} 
\begin{figure}[!htb]
\begin{center}
\includegraphics[width=9cm]{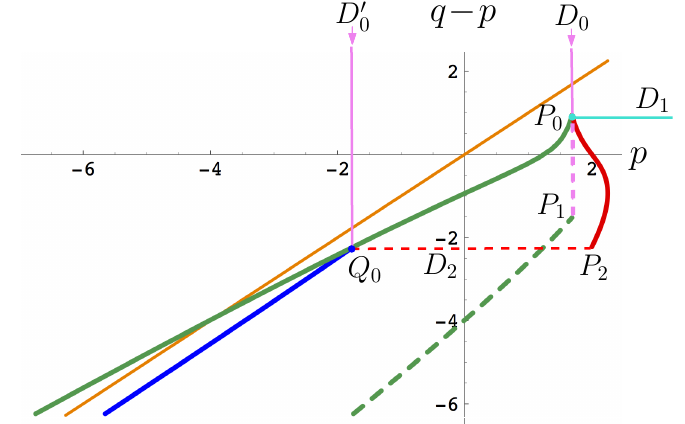}
\end{center}
\caption{Non-linear mapping \eqref{nltp-q} of the domains of validity of the proofs,  as shown in Fig. \ref{fig:raquette} from \cite{DHLZ} (here $\kappa=2$, $a=1$).}
\label{Fig6}
\end{figure}
\subsubsection{Phase diagram}
\begin{figure}[!ht]
\begin{center}
\includegraphics[width=9cm]{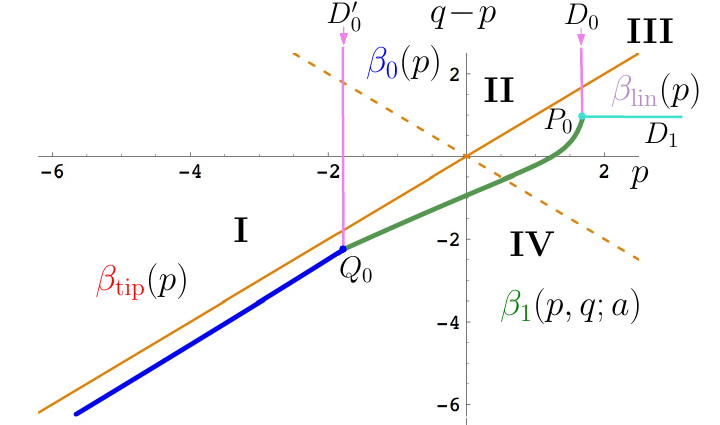}
\end{center}
\caption{Phase transition lines for the generalized integral means spectrum of simple whole-plane SLE$_\kappa$ with drift (here $\kappa=2$, $a=1$). The first bisector with $q=2p$ (orange continuous line) corresponds to the standard integral means spectrum \emph{(ims)} for the exterior case which crosses only phases I, II and III, whereas the second bisector with $q=0$ (orange dotted line) yields the standard \emph{ims} for the interior case, which does enter phase IV with the $\beta_1$ spectrum.}
\label{Fig2ter}
\end{figure}
\begin{figure}[!ht]
\begin{center}
\includegraphics[width=9cm]{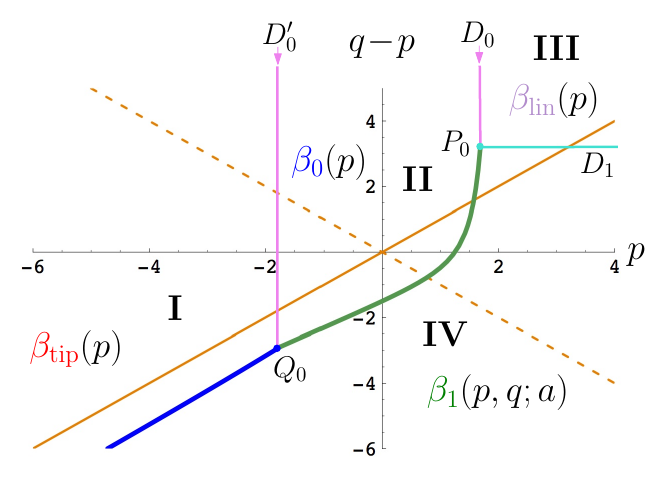}
\end{center}
\caption{Phase transition lines for the generalized integral means spectrum of simple whole-plane SLE$_\kappa$ with drift (here $\kappa=2$, $a=2$). The standard \emph{ims} in the $q=2p$ exterior case crosses all four phases in the order I, II, IV and III, while the $q=0$ standard \emph{ims} in the interior case crosses phases I, II and IV.}
\label{Fig7bis}
\end{figure}
\begin{figure}[!ht]
\begin{center}
\includegraphics[width=9cm]{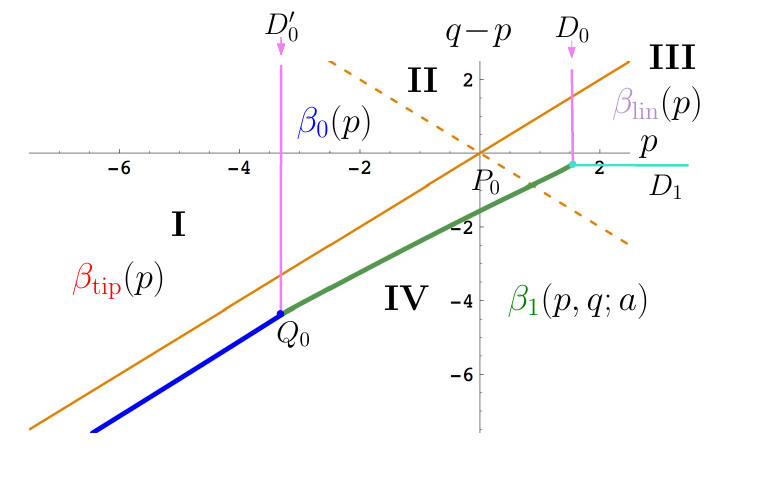}
\end{center}
\caption{Phase transition lines for the generalized integral means spectrum of non-simple whole-plane SLE$_\kappa$ with drift (here $\kappa=6$, $a=2$). The successive phase crossings of the two standard \emph{ims} bisector lines are analogous to those depicted in Fig. \ref{Fig2ter}.}
\label{Fig8bis}
\end{figure}
\begin{figure}[!htb]
\begin{center}
\includegraphics[width=9cm]{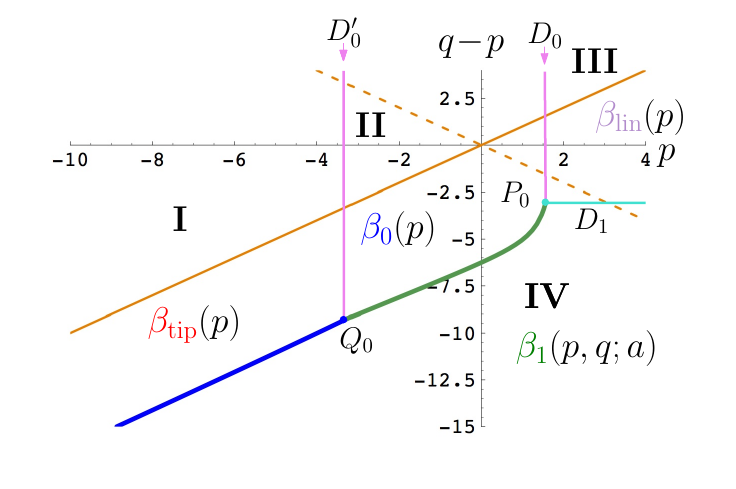}
\end{center}
\caption{Phase transition lines for the generalized integral means spectrum of non-simple whole-plane SLE$_\kappa$ with drift (here $\kappa=6$, $a=8$). The standard \emph{ims} in the $q=2p$ exterior case crosses phases I, II, III only, while the $q=0$ standard \emph{ims} in the interior case crosses all four phases in the order I, II, III and IV.}
\label{Fig9bis}
\end{figure}
Various cases, relative to the values of parameters $\kappa$ and $a$, and drawn thanks to the non-linear mapping \eqref{nltp-q}, are depicted in Figs. \ref{Fig2ter}, \ref{Fig7bis}, \ref{Fig8bis}, and \ref{Fig9bis}.  In these figures, it is especially interesting to focus on the standard integral means spectra in the $p,q-p$ plane, obtained for the whole-plane exterior version, along the line $q=2p$, hence $q-p=p$ (first bisector, golden continuous line), and for the whole-plane interior version along the line $q=0$, hence $q-p=-p$ (second bisector, golden dotted line).

Point $P_0$ \eqref{eq:P0} in the drift-less case  yields a value of $t_0:=p_0-q_0=\frac{\kappa^2-16}{32\kappa}$, so that $1+2\kappa t_0=\kappa^2/2$. The position of the translated point $P_0=(p_0,\tilde q_0)$ in the presence of drift $a$ is given by Eq. \eqref{nltp-q} as
\begin{equation}
\tilde q_0=p_0+\frac{16-\kappa^2}{32\kappa}\left(1+\frac{16 a^2}{\kappa^2}\right).
\end{equation}
To determine whether the first bisector enters region IV as in Fig. \ref{Fig7bis}, so that the \emph{exterior} standard whole-plane spectrum has a $\beta_1$ component, or avoids it as in Fig. \ref{Fig2ter}, we need to know the sign of $\tilde q_0-2p_0$. If positive, the first bisector passes below $P_0$ so that it successively traverses regions I, II, IV and III as in  Fig. \ref{Fig7bis}. Owing to  \eqref{eq:P0}, this happens for 
\begin{equation}\label{cond1}
\frac{a^2}{\kappa^2} \geq a_0(\kappa):=\frac{1}{4} \frac{2+\kappa}{4-\kappa},\,\,\,\kappa<4.
\end{equation} 
This phenomenon thus occurs only for simple SLE$_{\kappa<4}$ curves, and for a sufficiently strong drift term $a$. Otherwise, one is in the configuration of Figs. \ref{Fig2ter}, \ref{Fig8bis}, and \ref{Fig9bis} for the first bisector, and the $\beta_1$ spectrum does not appear in the standard {\it ims} of the exterior whole-plane SLE with drift.

To determine whether the second bisector enters region III and crosses all four phases as in Fig. \ref{Fig9bis}, so that the \emph{interior} standard whole-plane spectrum has a linear component $\beta_{\mathrm{lin}}$, or whether it avoids the linear phase III as in Fig. \ref{Fig8bis}, we need to know the position of $P_0$ with respect to that bisector, hence the sign of $\tilde q_0$. If negative, the second bisector passes above $P_0$, so that it successively traverses regions I, II, III and IV as in  Fig. \ref{Fig9bis}. This happens for 
\begin{equation}\label{cond2}
\frac{a^2}{\kappa^2} \geq \tilde a_0(\kappa):=\frac{1}{8} \frac{8+\kappa}{\kappa-4},\,\,\,\kappa>4.
\end{equation} 
This phenomenon thus occurs only for non-simple SLE$_{\kappa>4}$ curves, and for a sufficiently strong drift term $a$. Otherwise, one is in the configuration of Figs. \ref{Fig2ter}, \ref{Fig7bis}, and \ref{Fig8bis} for the second  bisector, and the $\beta_{\mathrm{lin}}$ spectrum does not appear in the standard {\it ims} of the interior whole-plane SLE with drift, which takes the successive forms $\beta_{\mathrm{tip}}, \beta_0, \beta_1$.
\begin{rkk}\label{dualitycd}
The two conditions on the \emph{reduced drift parameter}, $a/\kappa\geq a_0(\kappa)$, in fact obey \emph{SLE duality} \cite{Duplantier00,MR2112128,DupLH,MR2439609,dub_dual}. Defining  the dual SLE parameter $\kappa'=16/\kappa$, with $\kappa'>4$ and $\kappa< 4$, one checks that $\tilde a_0(\kappa')=\tilde a_0(16/\kappa)=a_0(\kappa)$. The occurence here of this reduced drift parameter $a/\kappa$ may seem natural, if one recalls that the quadratic variation  of $\sqrt{\kappa} B_t+at$ is $\kappa t$  and its mean $at$.
\end{rkk}
\subsection{Derivation of Claims \ref{claim1} and  \ref{claim2}}
\subsubsection{Discourse on the Method} We are going to use here a \emph{Liouville quantum gravity} (LQG) approach,  which historically gave the first derivation of the standard SLE multifractal spectrum  \cite{Duplantier00}, which was later confirmed by a standard mathematical  approach \cite{BS,MR3638311,gwynne2018}. It is based on the celebrated Knizhnik-Polyakov-Zamolodchikov (KPZ) relation \cite{MR947880,MR981529,MR1005268} between scaling exponents in the Euclidean plane, and their counterparts under a random LQG measure that gives the scaling limit of the area measure on a random planar map. The KPZ relation is now mathematically proved \cite{springerlink:10.1007/s00222-010-0308-1bis,rhodes-2008,DRSV12}.  Although  the LQG method, which originates in theoretical physics, is \emph{heuristic} and not  fully rigorous, it often offers the quickest and most natural path to the derivation of scaling exponents and multifractal spectra. It is also intimately related to the recently developed and rigorous \emph{wedge-welding} theory in Liouville quantum gravity \cite{10.1214/15-AOP1055,DMS14} (See in particular Appendix B in \cite{DMS14} for a mathematically precise description of the KPZ interpretation.)
\subsubsection{Derivation of Claim \ref{claim1}}
Let us first recall that in the original work on whole-plane SLE \cite{DNNZ}, the novel integral means spectrum, $\beta_1(p;\kappa):=\beta_1(p;q=0;\kappa, a=0)$, derived there  for $p\in \mathbb R$, was related to some Liouville quantum gravity results obtained in \cite{MR2112128}. (See \cite[Section 1.3]{DNNZ}.)  It was found that the related \emph{packing spectrum}, defined as,  
\begin{align}\label{packing}
s_1(p;\kappa):=\beta_1(p;\kappa)-p+1,
\end{align}
is given by 
\begin{align} \label{s}
 s_1(p)=s_1(p;\kappa):=2p+\frac{1}{2}-\frac{1}{2}\sqrt{1+2\kappa p}.
\end{align} 
When seen as a function of $p$, it has for inverse in terms of  $s=s_1$, 
\begin{equation*}
p=\frac{s}{2}+ \frac{\kappa}{8} {\mathcal U}^{-1}_\kappa(s),
\end{equation*}  
where we defined
\begin{align*}
&{\mathcal U}^{-1}_\kappa(x):= \frac{1}{2\kappa}\left(\kappa-4+\sqrt{(4-\kappa)^2+16\kappa x}\right),\\
&{\mathcal U}_\kappa(x):=\frac{1}{4}x\left(\kappa x+4-\kappa\right),\\
&{\mathcal V}_\kappa(x):={\mathcal U}_\kappa\left(\frac{1}{2}\left(x+1-\frac{\kappa}{4}\right)\right)=\frac{1}{16\kappa}\left[\kappa^2 x^2-(4-\kappa)^2\right].
\end{align*}  
Here ${\mathcal U}_\kappa$ is the KPZ function of Liouville quantum gravity adapted to SLE$_\kappa$, while ${\mathcal V}_\kappa$ is an associated function that relates boundary scaling dimensions to bulk ones \cite{MR2112128,DupLH}.  Here we generalize methods introduced in \cite{Duplantier00,PhysRevLett.89.264101} and  expounded in \cite{MR2112128,DupLH}, and use notations similar to those of \cite{MR2112128}, Section 8. For simplicity, we first implicitly assume SLE paths to be simple, i.e., with $\kappa \leq 4$, since the quantum gravity composition rules differ for the simple and non-simple phases of SLE \cite{MR2112128,DupLH}.  Nevertheless, the results obtained also hold for $\kappa >4$. One has the set of identities,
\begin{align}\label{pxs} 
&p=x_1(s)-x_1, \\ \label{x1s}
& x_1(s):=2{\mathcal V}_\kappa\left[{\mathcal U}^{-1}_\kappa(s)+{\mathcal U}^{-1}_\kappa(\tilde x_1)\right],\\ \label{x1} 
&x_1:=x_1(0)=\frac{1}{8\kappa}(6-\kappa)(2-\kappa),\\ \label{x1tilde}
&\tilde x_1:=\frac{6-\kappa}{2\kappa}, \,\,\,{\mathcal U}^{-1}_\kappa(\tilde x_1)=\frac{2}{\kappa}.
\end{align}  
 
The scaling exponent $x_1(s)$ geometrically corresponds to a configuration where the SLE \emph{tip} is locally avoiding a bunch of $s$ \emph{independent Brownian paths}. The tip here should be understood as the so-called SLE `second tip'  at the origin \cite{MR3638311}, after inversion of unbounded (interior) whole-plane SLE \cite{DNNZ}, as in Beliaev and Smirnov's  bounded (exterior) version of whole-plane SLE \cite{BS,MR3638311}. 

In the LQG approach,  $s$ independent Brownian paths \emph{avoiding} an SLE path near its tip are conformally equivalent to a certain number $k(s)$ of \emph{mutually-avoiding} SLEs in a star configuration,  given by 
  \begin{equation}\label{ks}
  k(s)=1+\frac{{\mathcal U}^{-1}_\kappa(s)}{{\mathcal U}^{-1}_\kappa(\tilde x_1)},
  \end{equation}
 such that $x_1(s)=2{\mathcal V}_\kappa(2 k(s)/\kappa)$. When $p\in \mathbb C$, its imaginary part $\tilde t:=\Im p$ corresponds to exponentially weighting by $\exp (\tilde t \arg \mathcal C)$ the mutually-avoiding SLE-Brownian path configurations $\mathcal C$, with local \emph{winding angle} $\arg \mathcal C$ around the tip. One can then show by Coulomb gas arguments \cite{PhysRevLett.89.264101,MR2112128,2008NuPhB.802..494D} that the new scaling exponent associated with the tip is 
  \begin{equation}\label{x1sti}
  \hat x_1(s,\tilde t):= x_1(s) -\frac{\kappa}{2}\frac{\tilde t^{\,2}}{k^2(s)}.
  \end{equation}
The average logarithmic spiral rotation rate $a$ near the tip is then obtained by Legendre transformation as \cite{PhysRevLett.89.264101,MR2112128},
 \begin{equation}\label{a}
 a=\frac{\partial}{\partial \tilde t} \, \hat x_1(s,\tilde t).
 \end{equation} 
On the other hand, the real part of $p$, $t:=\Re p$, is now given by the generalization of \eqref{pxs},
 \begin{equation}\label{tx1sti}
  \Re p=t=\hat x_1(s,\tilde t)-x_1,
  \end{equation} 
whereas the packing spectrum for complex $p$, $s=s_1(p;\kappa)=\beta_1(p;\kappa)-\Re p+1$, is still given by \eqref{s}, but now in terms of the reduced variable $\tau$, 
\begin{align}\label{tau}
& \tau:=x_1(s)-x_1,\\ \label{s1}
& s=s_1(\tau)=2\tau+\frac{1}{2}-\frac{1}{2}\sqrt{1+2\kappa \tau}.
  \end{align}
  From Eqs. \eqref{x1s}, \eqref{ks}, we find the simple identity \cite{PhysRevLett.89.264101,MR2112128}
   \begin{equation}\label{ks2}
  \frac{1}{2\kappa}{k^2(s)}= x_1(s) +b,\,\,\, b=\frac{(4-\kappa)^2}{8\kappa}.
  \end{equation}
 We thus find for \eqref{x1sti}  the simple formula,
 \begin{equation}\label{tx1stibis}
 \hat  x_1(s,\tilde t)=x_1(s)-\frac{1}{4}\frac{{\tilde t}^2}{x_1(s)+b},
  \end{equation}
  from which \eqref{tx1sti} gives,
 \begin{align}\nonumber
  t=\hat x_1(s,\tilde t)-x_1&=x_1(s)-x_1-\frac{1}{4}\frac{{\tilde t}^2}{x_1(s)-x_1+c}\\ \label{ttau}
  &=\tau-\frac{1}{4}\frac{{\tilde t}^2}{\tau+c},\,\,\,
  c:=b+x_1=\frac{1}{2\kappa}.
  \end{align}
  Eq. \eqref{ttau} is then inverted into 
  \begin{align}\label{taut}
  \tau=\frac{1}{2}\left(t-c\pm \sqrt{(t+c)^2+\tilde t^2}\right),
  \end{align}
  which can be recast as
  \begin{align}\label{taukt}
  1+2\kappa \tau=\frac{1}{2}(1+2\kappa t)\pm \frac{1}{2}\sqrt{(1+2\kappa t)^2+4\kappa^2\tilde t^2}.
 \end{align} 
  For $\tilde t=0$, we have $\tau=t$, which selects the (+)-branch in \eqref{taukt}, and recalling that $t=\Re p, \tilde t=\Im p$, we obtain
  \begin{equation} 
 1+2\kappa \tau=\frac{1}{2}(1+2\kappa \Re p)+ \frac{1}{2}\sqrt{(1+2\kappa \Re p)^2+4\kappa^2\Im p^2}, 
\end{equation}
which is the announced complex formula \eqref{t} for $p\in \mathbb C, q=0$.  When $q\neq 0$, we invoke the general validity of the observation made in Ref. \cite{DHLZ} that  the generalized packing spectrum, $s_1(p,q;\kappa,a=0)=\beta_1(p,q;\kappa,a=0)-\Re p+1$,  solely depends on the reduced variable $p-q$, hence $s_1(p,q;\kappa,0)=s_1(p-q,0;\kappa,0)$. \qed
\subsubsection{Derivation of Claim \ref{claim2}}
When $a\neq 0$, we modify the above aproach as follows. In the absence of Brownian paths, $s=0$, \eqref{tx1stibis} becomes, since $x_1(0)=x_1$,
\begin{equation}\label{tx10tibis}
  \hat x_1(0,\tilde t)=x_1(0)-\frac{1}{4}\frac{{\tilde t}^2}{x_1(0)+b}=x_1-\frac{\kappa}{2}{\tilde t}^2.
  \end{equation} 
  The spiral rotation rate $a$  then corresponds via \eqref{a} to a parameter $\tilde t_0$ such that,
  \begin{equation}
  a=-\kappa \tilde t_0,\,\,\, \hat x_1(0,\tilde t_0)=x_1-\frac{a^2}{2\kappa}.
  \end{equation}
 Re-centering around the spiralling rate $a$,  we define, instead of \eqref{tx1stibis},
  \begin{equation}\label{tx1stiter}
  \widehat x_1(s,\tilde t):=x_1(s)-\frac{1}{4}\frac{{(\tilde t-\tilde t_0)}^2}{x_1(s)+b},
  \end{equation}
  and substitute to \eqref{tx1sti}, \eqref{ttau}
 \begin{align}\nonumber
  t=\widehat x_1(s,\tilde t)-\hat x_1(0,\tilde t_0)&=x_1(s)-x_1+\frac{a^2}{2\kappa}-\frac{1}{4}\frac{{(\tilde t-\tilde t_0)}^2}{x_1(s)-x_1+c}\\ \label{ttaubis}
  &=\tau +\frac{a^2}{2\kappa}-\frac{1}{4}\frac{{(\tilde t-\tilde t_0)}^2}{\tau+c},\,\,\,
  c=\frac{1}{2\kappa},\,\,\,\tilde t_0=-\frac{a}{\kappa}.
  \end{align}
  Thus, instead of \eqref{taut} we find
 \begin{align}\label{tautbis}
  \tau=\frac{1}{2}\left(t-c-\frac{a^2}{2\kappa}\pm \sqrt{\left(t+c-\frac{a^2}{2\kappa}\right)^2+(\tilde t-\tilde t_0)^2}\right).
  \end{align} 
  By again selecting the $(+)$-branch, and recalling that $t=\Re p$, $\tilde t=\Im p $, this can finally be written as 
  \begin{align}\nonumber  
1+2\kappa \tau&=\frac{1}{2}(1-a^2+2\kappa t) + \frac{1}{2}\sqrt{(1-a^2+2\kappa t)^2+4(a+\kappa\tilde t)^2},\\
\label{tterbis}
&=\frac{1}{2}\left\{\Re\left[(1+ia)^2+2\kappa p\right]+\left |(1+ia)^2+2\kappa p\right |\right\}.
\end{align}
This is the announced result \eqref{tter} for $p\in \mathbb C, q=0$. Again, for $q\neq 0$, we invoke the fact \cite{DHLZ} that  the generalized packing spectrum, $s_1(p,q;\kappa,a)=\beta_1(p,q;\kappa,a)-\Re p+1$,  solely depends on the reduced variable $p-q$. \qed




\section{Integrable probability for drifted whole-plane SLE}
In order to anticipate the next section,  let us put the computations in a more general setting. 
\subsection{Some background on L\'{e}vy processes}
\begin{df}
A L\'evy process is a stochastic process $(L_t)_{t\geq 0}$ such that
\begin{itemize}
\item[(1)] $L_0=0$ (a.s);
\item[(2)] For any discrete ordered set $\left\{t_i,\,i\in I_n:=\{0,\cdots,n\}\right\}$, such that $t_0=0$ and $0\leq  t_i<t_{i+1},\, \forall i\in I_{n-1}$, the successive increments ,$L_{t_{i+1}}-L_{t_i},\,\,i\in I_{n-1}$, 
are all mutually independent; 
\item[(3)] For any $0\leq s\leq t$, $L_t-L_s$ has the same law as $L_{t-s}$.
\item[(4)] $L_t$ is continuous in probability, $\lim_{t\to 0} \mathbb P(|L_t-L_0|>\varepsilon)= 0,\,\,\,\forall \varepsilon >0$, which rules out fixed discontinuities of the path $t\mapsto L_t$.
\end{itemize}
\end{df}
Notice that Brownian motion is a special L\'evy process, and  a general difference with Brownian motion is that \emph{random jumps} are allowed. The characteristic function of a L\'evy process $L_t$ has the form
\begin{equation}\label{levysymbol}
\mathbb E[e^{i\xi L_t}]=e^{-t\eta(\xi)},
\end{equation}
where $\eta$, called the \textbf{L\'evy symbol}, is a continuous complex function of $\xi\in\R$, satisfying $\eta(0)=0$ and $\eta(-\xi)=\overline{\eta(\xi)}$. If $\eta(-\xi)=\eta(\xi)$,  $L_t$ is a \textbf{symmetric L\'evy process}. For Brownian motion, the L\'evy symbol is $\eta(\xi)=\frac{\xi^2}{2}$. More generally, the function
$$
\eta(\xi)=\frac{|\xi|^\alpha}{2},\alpha\in(0,2],
$$
is the L\'evy symbol of the so-called $\alpha-$stable process.

\subsection{Derivation of the PDE}
The inner whole-plane Loewner process is defined as the solution of the ODE in $\mathbb C$ 
\begin{align}\label{gt}
 \begin{cases}
 \partial_t g_t(z)=g_t(z)\frac{g_t(z)+\lambda(t)}{g_t(z)-\lambda(t)}, \,\,\, t\geq 0,\\
 \lim\limits_{t\rightarrow +\infty}e^{t}g_t(z)=z,\,\,\,\,\forall z\in\mathbb{C},
 \end{cases}
 \end{align}
with driving function $\lambda(t)=e^{iL_t}$ where $L_t$ is real-valued; $g_t$ is a conformal mapping from a simply connected domain $\Omega_t \subset \mathbb C$ onto $\mathbb{D}$, where $\Omega_t$ is  defined as $\Omega_t:=\left\{z\in\mathbb{C}: \tau_z>t\right\}$, where    \begin{equation*}
    \tau_z:=\sup\left\{t\in\mathbb{R}: \inf_{s\leq t}|g_s(z)-\lambda(s)|>0\right\}.
    \end{equation*}
Its inverse function $f_t:=g_t^{-1}$ obeys the PDE \eqref{eq:Loewner}, 
\begin{align}
  \begin{cases}\label{eq:PDE}
 \partial_t f_t(z)=zf_t'(z)\frac{\lambda(t)+z}{\lambda(t)-z}\\
 \lim\limits_{t\rightarrow+\infty}f_t(e^{-t}z)=z,\,\,\,\,\forall z\in\mathbb{D},
 \end{cases}
 \end{align}
where $f_t$ is now a mapping from $\mathbb D$ to the domain $\Omega_t=\mathbb C\setminus K_t$, where the connected set $K_t$ is the hull of the Loewner process.  In this section, we will assume $L_{t}, t \geq 0$, to be a \emph{L\'evy process}. The (complex) average integral means spectrum of the conformal map $f=f_0$, where $f_t$ is defined by \eqref{eq:PDE}, describes the singular behavior of the expectation, 
\begin{equation}\label{E-fpp}
\mathbb{E}[|f'(z)^p|]=\mathbb{E}\left[f'(z)^{\frac{p}{2}}\overline{f'(z)}^{\frac{\bar{p}}{2}}\right], \quad p\in\mathbb{C}.
\end{equation}
Similarly to the method used in \cite{DHLZ}, we shall consider the L\'evy-Loewner evolution (LLE) {\it two-point function} for $z_1,z_2\in \mathbb D$, defined as,
\begin{equation}\label{eq:GL}
G(z_1,\overline{z}_2):=\mathbb{E}\left[z_1^{\frac{q}{2}}\frac{f'(z_1)^{\frac{p}{2}}}{f(z_1)^{\frac{q}{2}}}\overline{z}_2^{\frac{\bar{q}}{2}}\frac{\overline{f'(z_2)}^{\frac{\bar{p}}{2}}}{\overline{f(z_2)}^{\frac{\bar{q}}{2}}}\right], \quad p,q\in\mathbb{C}.
\end{equation}
The moment \eqref{E-fpp} is the value $G(z,\bar z)$ at coinciding points $z_1=z_2=z$, for the case $q=\bar q=0$. Following essentially the same approach as was introduced in \cite{MR2153402,BS,MR3638311,DHLZ}, 
we aim at finding a partial differential equation satisfied by $G$.\\




Since $f_t$ obeys a PDE instead of an ODE, the use of It\^o calculus is problematic. A way to overcome this difficulty \cite{BS} is to consider the ODE \eqref{gt} for negative times, and then compare 
the \emph{reverse} function $g_{-t}$ to the \emph{inverse} $g_t^{-1}$. The details are as follows. 

For any fixed $s\geq 0$, define the auxiliary function $g_t^{(s)}$ such that:  $ g_t^{(s)}(z)=e^{-t}z$ for $t>s$, while for $t\leq s$, $g_t^{(s)}$ is the solution to the differential equation \eqref{gt} with the initial (continuity) condition $g_s^{(s)}(z)=e^{-s}z$,
\begin{equation}
\begin{cases}\label{3}
\partial_tg_t^{(s)}(z)=g_t^{(s)}(z)\frac{g_t^{(s)}(z)+\lambda(t)}{g_t^{(s)}(z)-\lambda(t)}, \,\,\, t\geq 0,\\
g_s^{(s)}(z)=e^{-s}z.
\end{cases}
\end{equation}
The family of conformal maps $(g_t^{(s)}(z))_{t \geq 0}$ is illustrated in Fig.\ref{fig:gst}.
\begin{figure}[!htb]
\begin{center}
\includegraphics[width=0.8\linewidth]{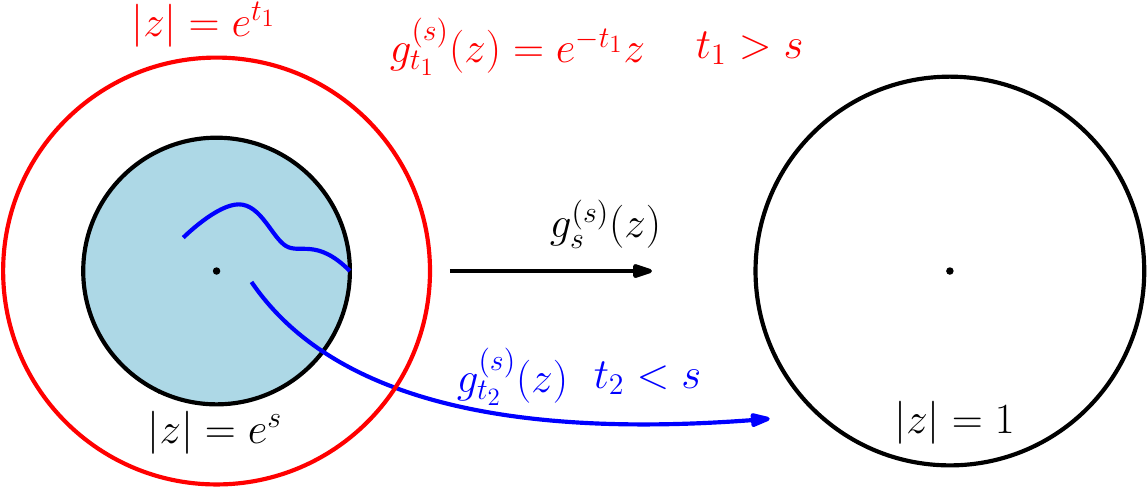}
\end{center}
\caption{Family of auxiliary conformal maps $(g_t^{(s)})_{t \geq 0}$ for fixed $s\geq 0$.  When $0\leq s \leq t$, $g_{t}^{(s)}(z)=e^{-t}z$ maps $e^t\mathbb D$ onto $\mathbb D$, 
whereas for $0\leq t\leq s$, $g_{t}^{(s)}(z)$ maps 
$e^s \mathbb D \setminus K_{t}$ onto $\mathbb D$.}
\label{fig:gst}
\end{figure}
\begin{lemma}\label{lem-gts} With $g_t$ and $g_t^{(s)}$ defined as above, we have, for any $t\geq 0$,
$$\lim\limits_{s\rightarrow+\infty}g_t^{(s)}(z)=g_t(z).$$
\end{lemma}
This lemma is just the interior version of the following result by Lawler \cite{lawler} for the exterior whole-plane case.
\begin{lemma}\cite[Def. 6.28, Prop. 4.21]{lawler} \label{lem-gtstilde}
Let $\tilde{g}_t(z)$ be the solution of the differential equation,
\begin{equation}\label{5}
\begin{cases}
\partial_t \tilde{g}_t(z)=\tilde{g}_t(z)\frac{\tilde{\lambda}(t)+\tilde{g}_t(z)}{\tilde{\lambda}(t)-\tilde{g}_t(z)},\,\,\, t\leq 0,\\
\lim\limits_{t\rightarrow-\infty}e^t\tilde{g}_t(z)=z,\,\,\,\forall z\in\mathbb{C}\backslash\{0\}.
\end{cases}
\end{equation}
For any fixed $s\geq 0$, define $\tilde{g}_t^{(s)}(z)$ as: $\tilde{g}_t^{(s)}(z)=e^{-t}z$ if $t\leq -s$; for $t\geq -s$, $\tilde{g}_t^{(s)}(z)$ is the solution 
of the above differential equation with initial value $\tilde{g}_{-s}^{(s)}(z)=e^sz$. Then for $t\leq 0$, $\lim\limits_{s\rightarrow+\infty} \tilde{g}_t^{(s)}(z)=\tilde{g}_t(z)$.
\end{lemma}

In order to prove that Lemma \ref{lem-gts} follows from \ref{lem-gtstilde}, one applies complex inversion and time reversal so as to define 
for $t\leq 0$, $\tilde{g}_t^{(s)}(z):=1/g_{-t}^{(s)}(1/z)$ and $\tilde{\lambda}(t):=1/\lambda(-t)$, where $g_t^{(s)}(z)$ is defined by \eqref{3}. Then $\tilde{g}_t^{(s)}$  is as in Lemma \ref{lem-gtstilde} and for $s\rightarrow +\infty$, it converges to the limit $\tilde{g}_t$ obeying \eqref{5}. 
It then finally suffices to check that $g_{t}$, defined for $t\geq 0$ as $g_t(z):=1/{\tilde{g}_{-t}({1}/{z})}$, satisfies  \eqref {gt}.\\

 We then define  a reversed radial LLE, as the solution to the ODE in the unit disk $\mathbb D$,
\begin{equation}\label{f-tilde}
\partial_t\tilde{f}_t(z)=\tilde{f}_t(z)\frac{\tilde{f}_t(z)+\lambda(t)}{\tilde{f}_t(z)-\lambda(t)},
\quad\tilde{f}_0(z)=z,\,\,\,  \forall z\in\D,\,\,\, t\geq 0.
\end{equation}
\begin{lemma}\label{lem::relation}
For $f_t$ as defined in \eqref{eq:PDE} and $\tilde f_t$ as defined in \eqref{f-tilde}, we have the equivalence in law,
\begin{equation}\label{eq:equiv}
\lim_{t\rightarrow+\infty} e^t\tilde{f}_t(z)\stackrel{\rm (law)}=f_0(z).
\end{equation}
\end{lemma}
\begin{proof}
For any fixed $s\geq0$, let $g_t^{(s)}$ be as above. Then  we have $g_t^{(s)}(z)\stackrel{\rm (law)}=\tilde{f}_t\big(g_0^{(s)}(z)\big)$, because both obey \eqref{3}, and they coincide at $t=0$ because of the initial condition in \eqref{f-tilde}. 
We then have,  
$e^t\tilde{f}_t(z)\stackrel{\rm (law)}=e^t g_t^{(s)}\big((g_0^{(s)})^{-1}(z)\big)$. Letting $s=t$, we get
\begin{equation}\label{eq:etfg}
e^t\tilde{f}_t(z)\stackrel{\rm (law)}=(g_0^{(t)})^{-1}(z),
\end{equation}
and if we let $t\rightarrow+\infty$, by Lemma \ref{lem-gts} we have $\lim\limits_{t\rightarrow+\infty} e^t\tilde{f}_t(z)\stackrel{\rm (law)}=g_0^{-1}(z)=f_0(z)$.
\end{proof}

Let us define the auxiliary, time-dependent, radial variant of the LLE two-point function $G(z_1,\bar z_2)$ \eqref{eq:GL},
\begin{eqnarray}\nonumber
\widetilde G(z_1,\bar{z}_2,t)&:=&\mathbb{E}\left[z_1^{\frac{q}{2}}\frac{\tilde{f}'_t(z_1)^{\frac{p}{2}}}{\tilde{f}_t(z_1)^{\frac{q}{2}}}\overline{z}_2^{\frac{\bar{q}}{2}}\frac{\overline{\tilde{f}'_t(z_2)}^{\frac{\bar{p}}{2}}}{\overline{\tilde{f}_t(z_2)}^{\frac{\bar{q}}{2}}}\right]\\ \label{eq:tildeG}
      &=&\mathbb{E}\left[z_1^{\frac{q}{2}}X_t(z_1)\overline{z}_2^{\frac{\bar{q}}{2}}Y_t(\bar{z}_2)\right],
\end{eqnarray}
where $\tilde f_t$ is the reversed radial Loewner process \eqref{f-tilde}, together with the shorthand notations,
$$X_t(z):=\frac{\tilde{f}'_t(z)^{\frac{p}{2}}}{\tilde{f}_t(z)^{\frac{q}{2}}},\,\,\, Y_t(\bar z):=\overline{X_t(z)}=\frac{\overline{\tilde{f}'_t(z)}^{\frac{\bar{p}}{2}}}{\overline{\tilde{f}_t(z)}^{\frac{\bar{q}}{2}}}.$$


By using \eqref{eq:etfg}, where the r.h.s. and its derivative are locally uniformly bounded by the Koebe distorsion theorem, 
the two-point function $G(z_1,\bar{z}_2)$ \eqref{eq:GL} is, by Lebesgue's dominated convergence theorem, the limit
\begin{align}\label{tendstoinfinity}
\lim_{t\rightarrow+\infty} e^{\Re(p-q)t}\widetilde{G}(z_1,\bar{z}_2,t)=G(z_1,\bar{z}_2).
\end{align}
\begin{rkk}\label{hol}The same argument implies that $\widetilde{G}(z_1,\bar{z}_2,t)$ and $G(z_1,\bar{z}_2)$ are holomorphic with respect to both $z_1$ and $\bar{z}_2$. 
\end{rkk}
As explained in \cite{MR2153402,BS}, the idea is then to construct a martingale $\mathcal M_s$ related to $\widetilde G$. The vanishing of the drift term in its It\^o derivative then yields a partial differential equation obeyed by $\widetilde G$. 

For $s\leq t$, define the two-point martingale $(\mathcal M_s)_{t\geq s\geq 0}$ with 
$$\mathcal{M}_s:=\mathbb{E}[X_t(z_1)Y_t({\bar z}_2)|\mathcal{F}_s],$$ 
where the random variable is integrable for fixed $z_1$ and $z_2$, and where $\mathcal{F}_s$ is the $\sigma$-algebra generated by the L\'evy process filtration $\{L_u, u\leq s\}$. By the Markov property of the L\'evy process, we know that for any $s\leq t$,
\begin{align}
\tilde{f}_t(z)\stackrel{\rm (law)}=\lambda(s)\Ti f/_{t-s}(\Ti f/_s(z)/\lambda(s)).
\end{align}
Therefore, 
\begin{equation}\mathcal{M}_s=X_s(z_1) Y_s(\bar{z}_2) \widetilde G(z_{1,s},\bar{z}_{2,s};\tau),\,\,\,\tau:=t-s, \forall \tau >0,
\end{equation} 
where
$$
z_{1,s}:=\frac{\tilde{f}_s(z_1)}{\lambda(s)},\,\,\, \bar z_{2,s}:=\overline{\left(\frac{\tilde{f}_s(z_2)}{\lambda(s)}\right)}=\overline{\tilde{f}_s(z_2)}\lambda(s).
$$
 In order to prepare for It\^o calculus, we have \cite[Section 4, Eqs. (47-49)]{DHLZ}
\begin{align}\label{eq:s-derivatives1}
dX_s(z_1)=& X_s(z_1)\left[\frac{p}{2}-\frac{q}{2}-\frac{p}{(1-z_{1,s})^2}+\frac{q}{1-z_{1,s}}\right]ds,\\[3mm]\label{eq:s-derivatives2}
dY_s(\bar z_2)=& Y_s(\bar z_2)\left[\frac{\bar{p}}{2}-\frac{\bar{q}}{2}-\frac{\bar{p}}{(1-\bar{z}_{2,s})^2}+\frac{\bar{q}}{1-\bar{z}_{2,s}}\right]ds,\\[3mm] \label{eq:s-derivatives3}
\frac{\partial z_{1,s}}{\partial s}{\big\vert_{\lambda(s)}}=& z_{1s}\frac{z_{1,s}+1}{z_{1,s}-1},\,\,\,\frac{\partial\bar{z}_{2,s}}{\partial s}\big\vert_{\lambda(s)}=\bar{z}_{2,s}\frac{\bar{z}_{2,s}+1}{\bar{z}_{2,s}-1}.
\end{align}

Let us write $\mathcal M_s$  as a formal function of two variables, 
$$H(s,L_s):=\mathcal M_s=X_s(z_1) Y_s(\bar{z}_2) \widetilde{G}(z_{1,s},\bar{z}_{2,s},t-s).$$
It is a (local) martingale for all $s\leq t$, thus by It\^o calculus its total $s$-derivative vanishes,
$$
\Lambda H(s,L_s)+\partial_sH(s,L_s)=0,
$$
where $\Lambda$ is the generator of the L\'evy process $L_s$.

We have from Eqs. \eqref{eq:s-derivatives1}, \eqref{eq:s-derivatives2}, \eqref{eq:s-derivatives3},
\begin{align*}
&\partial_s H=H\left[\frac{p}{2}-\frac{q}{2}-\frac{p}{(1-z_{1,s})^2}+\frac{q}{1-z_{1,s}}\right]\\[3mm]
+&H\left[\frac{\bar{p}}{2}-\frac{\bar{q}}{2}-\frac{\bar{p}}{(1-\bar{z}_{2,s})^2}+\frac{\bar{q}}{1-\bar{z}_{2,s}}\right]\\[3mm]
-&X_s(z_1)Y_s(\bar{z}_2)\partial_\tau \widetilde{G}(z_{1,s},\bar{z}_{2,s},t-s)\\[3mm]
+&X_s(z_1)Y_s(\bar{z}_2)\partial_{z_1}\widetilde{G}(z_{1,s},\bar{z}_{2,s},t-s)z_{1s}\frac{z_{1,s}+1}{z_{1,s}-1}\\[3mm]
+&X_s(z_1)Y_s(\bar{z}_2)\partial_{\bar{z}_2}\widetilde{G}(z_{1,s},\bar{z}_{2,s},t-s)\bar{z}_{2,s}\frac{\bar{z}_{2,s}+1}{\bar{z}_{2,s}-1},
\end{align*}
where $\tau:=t-s$.
Since neither $X_s(z)$ nor its complex conjugate $Y_s(z)$ vanish in $\mathbb D$, we deduce that
\begin{align}\label{usedusedherehere}
-\Lambda \widetilde{G}(z_{1,s},\bar{z}_{2,s},t-s)&=\widetilde{G}(z_{1,s},\bar{z}_{2,s},t-s)\left[\frac{p}{2}-\frac{q}{2}-\frac{p}{(1-z_{1,s})^2}+\frac{q}{1-z_{1,s}}\right]\nonumber\\
&+\widetilde{G}(z_{1,s},\bar{z}_{2,s},t-s)\left[\frac{\bar{p}}{2}-\frac{\bar{q}}{2}-\frac{\bar{p}}{(1-\bar{z}_{2,s})^2}+\frac{\bar{q}}{1-\bar{z}_{2,s}}\right]\nonumber\\
&-\partial_\tau \widetilde{G}(z_{1,s},\bar{z}_{2,s},t-s)+\partial_{z_1}\widetilde{G}(z_{1,s},\bar{z}_{2,s},t-s)z_{1s}\frac{z_{1,s}+1}{z_{1,s}-1}\nonumber\\
&+\partial_{\bar{z}_2}\widetilde{G}(z_{1,s},\bar{z}_{2,s},t-s)\bar{z}_{2,s}\frac{\bar{z}_{2,s}+1}{\bar{z}_{2,s}-1}.
\end{align}

Notice that by \eqref{tendstoinfinity}, it holds that, as $t\to +\infty$,
\begin{equation*}
\Re (p-q)\exp\left[\Re (p-q)t\right]\widetilde{G}(z_{1},\bar{z}_2,t) 
+\exp\left[\Re (p-q)t\right]\partial_t \widetilde{G}(z_{1},\bar{z}_2,t)\to 0,
\end{equation*}
so that 
\begin{equation}\label{eq:limG}
\lim_{t\to +\infty}\exp\left[\Re (p-q)t\right]\partial_t \widetilde{G}(z_{1},\bar{z}_2,t) = - \Re (p-q)\, G(z_1,z_2).
\end{equation}
Multiplying both sides of \eqref{usedusedherehere} by $\exp\left[\Re (p-q)(t-s)\right]$, and letting $t \to +\infty$, we  get
\begin{align}
-\Lambda G(z_{1},\bar{z}_2)&=G(z_{1},\bar{z}_2)\left[\frac{p}{2}-\frac{q}{2}-\frac{p}{(1-z_{1})^2}+\frac{q}{1-z_{1}}\right]\nonumber\\
&+G(z_{1},\bar{z}_2)\left[\frac{\bar{p}}{2}-\frac{\bar{q}}{2}-\frac{\bar{p}}{(1-\bar{z}_{2})^2}+\frac{\bar{q}}{1-\bar{z}_{2}}\right]\nonumber\\
&+\Re (p-q)G\left(z_{1},\bar{z}_2\right)\nonumber\\
&+\partial_{z_1}G(z_{1},\bar{z}_2)z_{1}\frac{z_{1}+1}{z_{1}-1}+\partial_{\bar{z}_2}G(z_{1},\bar{z}_2)\bar{z}_2\frac{\bar{z}_2+1}{\bar{z}_2-1}.
\end{align}
We finally get that $G(z_1,\bar{z}_2)$ satisfies $\mathcal{P}(D)\,G(z_1,\bar{z}_2) = 0$, where
\begin{align}\label{generaloperator}
\mathcal{P}(D)&:=\Lambda+z_{1}\frac{z_{1}+1}{z_{1}-1}\partial_{z_1}+\bar{z}_2\frac{\bar{z}_2+1}{\bar{z}_2-1}
\partial_{\bar{z}_2}+p-q+\bar{p}-\bar{q}\nonumber\\[3mm]
&-\frac{p}{(1-z_{1})^2}+\frac{q}{1-z_{1}}-\frac{\bar{p}}{(1-\bar{z}_{2})^2}+\frac{\bar{q}}{1-\bar{z}_{2}}.
\end{align}
Recall the definition of $\Lambda$ acting on a $C^{\infty}(\mathbb R^2)$ function $u$,
$$
 \Lambda u(x)=\lim\limits_{t\downarrow 0}\frac{1}{t}\left({\mathbb E}^x[u(L_t)]-u(x)\right).
$$
For $k,l\in\Z$, we have for $z=re^{i\theta}$,
\begin{align}\label{levysymbolspecial}
\Lambda(z^k\bar{z}^l)&=r^{k+l}\Lambda (e^{i\theta(k-l)})
=r^{k+l}\lim\limits_{t\downarrow 0}\frac{1}{t}\left({\mathbb E}^\theta[e^{i(k-l)L_t}]-e^{i(k-l)\theta}\right)\nonumber\\[3mm]
&=r^{k+l}\lim\limits_{t\downarrow 0}\frac{1}{t}(e^{-t\eta(k-l)}-1)\, e^{i(k-l)\theta}
=-\eta(k-l)z^k\bar{z}^l,
\end{align}
where $\eta$ is the L\'evy symbol of $L_t$.


\subsection{Drifted Brownian motion}
In this section, we consider the special L\'evy process
$L_t=at+\sqrt{\kappa}B_t$, where $a\in\R, \kappa\geq 0$ and $B_t$ is standard one-dimensional Brownian motion. These processes are the most general L\'evy processes with a.s. continuous trajectories. By definition of the  L\'evy symbol
$$
\expect[e^{i\xi L_t}]=\expect[e^{i\xi(at+\sqrt{\kappa}B_t)}]=e^{ia\xi t-t\frac{\kappa}{2}\xi^2}=e^{-t\eta(\xi)}.
$$
So
$$
\eta(\xi)=\frac{\kappa}{2}\xi^2-ia\xi.
$$
By \eqref{levysymbolspecial}, we have
$$
\Lambda (z^k\bar{z}^l)=-\eta(k-l)z^k\bar{z}^l=\left(-\frac{\kappa}{2}(k-l)^2+ia(k-l)\right)z^k\bar{z}^l,
$$
so that the L\'evy generator in the Brownian drift case is explicitly
$$
\Lambda=-\frac{\kappa}{2}(z\partial_z-\bar{z}\partial_{\bar{z}})^2+ia(z\partial_z-\bar{z}\partial_{\bar{z}}).
$$
The operator in \eqref{generaloperator} thus becomes
\begin{align}\label{eq:PD}
\mathcal{P}(D)=&-\frac{\kappa}{2}(z_1\partial_{z_1}-\bar{z}_2\partial_{\bar{z}_2})^2+z_{1}\left(\frac{z_{1}+1}{z_{1}-1}+ia\right)\partial_{z_1}\nonumber\\[3mm]
+&\bar{z}_2\left(\frac{\bar{z}_2+1}{\bar{z}_2-1}-ia\right)
\partial_{\bar{z}_2}+p-q+\bar{p}-\bar{q}\nonumber\\[3mm]
-&\frac{p}{(1-z_{1})^2}+\frac{q}{1-z_{1}}-\frac{\bar{p}}{(1-\bar{z}_{2})^2}+\frac{\bar{q}}{1-\bar{z}_{2}}.
\end{align}
\subsubsection{Algebraic solutions}
We want to find some solutions to the PDE
\begin{equation}\label{eq:PDEG}
\mathcal{P}(D)G(z_1,\bar{z}_2)=0,\,\,\,G(0,0)=1,
\end{equation}
and follow the method of Ref. \cite{DNNZ}, by looking for solutions of the form,
\begin{equation}\label{eq:Gfactor}
G(z_1,\bar{z}_2)=(1-z_1)^\alpha(1-\bar{z_2})^{\bar{\alpha}}P(z_1\bar{z}_2),\,\,\,P(0)=1.
\end{equation}
The action of the partial differential operator $\mathcal P(D)$ \eqref{eq:PD} readily gives
\begin{align*}
&\mathcal{P}(D)[(1-z_1)^\alpha(1-\bar{z_2})^{\bar{\alpha}}P(z_1\bar{z}_2)]\\
&=z_1\bar{z}_2(1-z_1)^{\alpha-1}(1-\bar{z}_2)^{\bar{\alpha}-1}(\kappa \alpha\bar{\alpha}P(z_1\bar{z}_2)+2(z_1\bar{z}_2-1)P'(z_1\bar{z}_2))\\[3mm]
&+[\mathcal{P}(\partial)(1-z_1)^\alpha](1-\bar{z}_2)^{\bar{\alpha}}P(z_1\bar{z}_2)
+[\mathcal{P}(\bar{\partial})(1-\bar{z}_2)^{\bar{\alpha}}](1-z_1)^\alpha P(z_1\bar{z}_2),\\[3mm]
\end{align*}
where
\begin{align}\label{Pdelta}
\mathcal{P}(\partial) &:= -\frac{\kappa}{2}(z_1\partial_{z_1})^2+\left(\frac{z_1+1}{z_1-1}+ia\right)z_1\partial_{z_1}+p-q+\frac{q}{1-z_1}-\frac{p}{(1-z_1)^2},\\[3mm] \nonumber
\mathcal{P}(\bar{\partial})&:=-\frac{\kappa}{2}(\bar{z}_2\partial_{\bar z_2})^2+\left(\frac{\bar{z}_2+1}{\bar{z}_2-1}-ia\right)\bar{z}_2\partial_{\bar z_2}+\bar{p}
-\bar{q}+\frac{\bar{q}}{1-\bar{z}_2}-\frac{\bar{p}}{(1-\bar{z}_2)^2}.
\end{align}
Notice that as complex conjugates, $$ \forall z\in \mathbb D,\,\,\,\mathcal{P}(\partial)(1-z)^\alpha=0\Leftrightarrow\mathcal{P}(\bar{\partial})(1-\bar{z})^{\bar{\alpha}}=0.$$ So if we have,
\begin{equation}\label{eq:cond1}
\mathcal{P}(\partial)(1-z_1)^\alpha=0,
\end{equation}
then equation \eqref{eq:PDEG} reduces for \eqref{eq:Gfactor} to
\begin{align}\nonumber
&\kappa\alpha\bar\alpha P(z_1\bar{z}_2)+2(z_1\bar{z}_2-1)P'(z_1\bar{z}_2)=0,\,\,\,P(0)=1\\[3mm]
&\Leftrightarrow P(z_1\bar{z}_2)=(1-z_1\bar{z}_2)^{-\frac{\kappa}{2} \alpha \bar\alpha} \label{eq:sol}.
\end{align}

Let us now look for $\alpha$ such that Eq. \eqref{eq:cond1} is satisfied. A direct computation readily gives \cite{DNNZ},
 $$
\mathcal{P}(\partial)(1-z)^\alpha=(1-z)^\alpha A+(1-z)^{\alpha-1}B+(1-z)^{\alpha-2}C,
$$
where
\begin{align}\label{pararouge}
A&:=-\frac{\kappa}{2}\alpha^2+(1+ia)\alpha+p-q,\\[3mm]\label{pararouge2}
B&:=\kappa\alpha^2-\left(\frac{\kappa}{2}+3-ia \right)\alpha +q,\\[3mm]\label{pararouge3}
C&:=-\frac{\kappa}{2}\alpha^2+\left(2+\frac{\kappa}{2}\right)\alpha-p.
\end{align}

Notice that $A+B+C=0$.  For any $\alpha \in \mathbb C$,  the choice of  $p, q$ such that $B=0$ and $C=0$, yields a solution to \eqref{eq:cond1}, hence together with \eqref{eq:sol} a solution \eqref{eq:Gfactor} to \eqref{eq:PDEG}. 

We thus get the identity for drifted SLE,
\begin{equation}\label{eq:G12}
G(z_1,\bar z_2)=\mathbb E\left[z_1^{\frac{q}{2}}\frac{f'(z_1)^{\frac{p}{2}}}{f(z_1)^{\frac{q}{2}}}\bar{z}_2^{\frac{\bar{q}}{2}}\frac{\overline{f'(z_2)}^{\frac{\bar{p}}{2}}}{\overline{f(z_2)}^{\frac{\bar{q}}{2}}}\right]=(1-z_1)^\alpha(1-\bar{z}_2)^{\bar{\alpha}}(1-z_1\bar{z}_2)^{-\frac{\kappa}{2} \alpha\bar\alpha},
\end{equation}
 where the quadratic equations $B=0$, $C=0$ yield $p$ and $q$ in terms of $\alpha \in \mathbb C$ under the parametric form,
 \begin{align}\label{prouge}
& p=-\frac{\kappa}{2}\alpha^2+\left(2+\frac{\kappa}{2}\right)\alpha, \quad\,\,\, \alpha\in \mathbb C,\\ \label{qrouge}
& q=-\kappa\alpha^2+\left(\frac{\kappa}{2}+3-ia \right)\alpha . 
 \end{align}
These results generalize those found for real $p,q$ and $a=0$ in \cite{DHLZ}. The complex $p,q$ case, still for  $a=0$, has been thoroughly studied in Ref. \cite{ho:tel-01581324}. These equations generalize in the complex $p,q$ case, hence in four-dimensional space, the so-called \emph{red parabola} of the real $(p,q)$-plane described in \cite{DHLZ}. 
By remark \ref{hol}, and \cite[Lemma 3.1]{DNNZ} the space of holomorphic solutions in $z_1$,  $\bar{z}_2$ to the linear PDE in \eqref{eq:PDEG} is one-dimensional. As a consequence, we have proven the following 
\begin{theo}\label{thmrouge}
Let $f(z)=f_0(z)$ where $f_t$ is the drifted whole-plane Loewner process driven by $\lambda(t)=e^{i(at+\sqrt{\kappa}B_t)}, a \in \mathbb R$. For $(p,q)\in \mathbb C^2$, let the complex `red parabola' $\mathcal R$ be defined as the two-dimensional manifold, 
\begin{align} \label{Red} 
p=-\frac{\kappa}{2}\alpha^2+\left(2+\frac{\kappa}{2}\right)\alpha,\quad\,\,\,
q-p=-\frac{\kappa}{2}\alpha^2+(1-ia)\alpha,\quad\,\,\, \alpha\in \mathbb C.
\end{align}
For $(p,q)\in \mathcal R, \,\, z_1, z_2\in \mathbb D$, we identically have 
\begin{equation*}
G(z_1,\bar z_2)=\mathbb E\left[z_1^{\frac{q}{2}}\frac{f'(z_1)^{\frac{p}{2}}}{f(z_1)^{\frac{q}{2}}}\bar{z}_2^{\frac{\bar{q}}{2}}\frac{\overline{f'(z_2)}^{\frac{\bar{p}}{2}}}{\overline{f(z_2)}^{\frac{\bar{q}}{2}}}\right]=(1-z_1)^\alpha(1-\bar{z}_2)^{\bar{\alpha}}(1-z_1\bar{z}_2)^{-\frac{\kappa}{2} \alpha\bar\alpha}.
\end{equation*}
In particular, for $z_1=z_2=z$,
\begin{equation}\label{Gzbarz}
G(z,\bar z)=\mathbb E\left[\left|z^{q}\right|\left|\frac{f'(z)^{p}}{f(z)^{q}}\right|\right]=(1-z)^\alpha(1-\bar{z})^{\bar{\alpha}}(1-z\bar{z})^{-\frac{\kappa}{2} \alpha\bar\alpha}.
\end{equation}
\end{theo}
Hence, in the case of the complex red parabola \eqref{prouge}, \eqref{qrouge}, we find that the complex generalized \emph{bulk} spectrum is simply given by 
\begin{equation}\label{betaalpha}
\beta(p,q;\kappa,a)=\frac{1}{2}\kappa|\alpha|^2.
\end{equation}
\begin{rkk}\label{tip}
\emph{Tip spectrum.}  When $2\Re \alpha +1\leq 0$, the presence in \eqref{Gzbarz} of the singular factor $|(1-z)^\alpha|^2$, besides that of the bulk singular one, brings in an extra singular contribution to the integral means near $z=1$.  
This yields the new complex generalized \emph{tip} spectrum \cite{BS,MR3638311,DHLZ} along the red parabola \eqref{Red},
\begin{equation}\label{betatip}
\beta(p,q;\kappa,a)=\frac{1}{2}\kappa|\alpha|^2-2\Re \alpha -1,\,\,\,2\Re \alpha\leq -1.
\end{equation}
\end{rkk}
Let us now turn to the case of \emph{real} points along the complex red  parabola $\mathcal R$ \eqref{Red}.
 \begin{cor}\label{mainthm}
 Let $f(z)=f_0(z)$ where $f_t$ is the drifted whole-plane Loewner process driven by $\lambda(t)=e^{i(at+\sqrt{\kappa}B_t)}$.  If $p,q$ take the following values:
 $$
p=p(\kappa,a)=\frac{(4+\kappa)^2}{8\kappa}\left(1+\frac{4a^2}{(2+\kappa)^2}\right),\;
q=q(\kappa,a)=\frac{4+\kappa}{2\kappa}\left(1+\frac{4a^2}{(2+\kappa)^2}\right);
$$
 then the  generalized integral means spectrum $\beta(p,q)$ of $f$ is equal to $p$.
 \end{cor}
\begin{proof}
Let us look for exponents $p,q\in\mathbb{R}$, as parameterized by \eqref{prouge} and \eqref{qrouge}, with $\alpha=\alpha_1+i\alpha_2$ and $\alpha_1,\alpha_2 \in \mathbb{R}$. 
The condition $\Im{p}=0$ gives
$$
\alpha_2\left(-\kappa\alpha_1+2+\frac{\kappa}{2}\right)=0,
$$
hence either $\alpha_2=0$ or $\alpha_1={(4+\kappa)}/{2\kappa}$.
The condition $\Im{q}=0$ yields 
$$2\kappa\alpha_1\alpha_2-\left(\frac{\kappa}{2}+3\right)\alpha_2+a\alpha_1=0.$$
So if $\alpha_2=0$, we have either $\alpha_1=0$ or $a=0$. The first case is trivial, while the second one  is the driftless case  studied in   \cite{DHLZ}. So, assuming $a\neq 0$, we obtain
$$
\alpha_1=\frac{4+\kappa}{2\kappa},\quad \alpha_2=-\frac{a(4+\kappa)}{\kappa(2+\kappa)},
$$
and
$$
\alpha=\frac{4+\kappa}{2\kappa}\left(1-i\frac{2a}{2+\kappa}\right),
$$
which in turn gives
\begin{align}
p &= p(\kappa,a):=\frac{(4+\kappa)^2}{8\kappa}\left(1+\frac{4a^2}{(2+\kappa)^2}\right),\label{preal} \\
q &=q(\kappa,a):=\frac{4+\kappa}{2\kappa}\left(1+\frac{4a^2}{(2+\kappa)^2}\right) \label{qreal}.
\end{align}
Notice the further identity  $\frac{\kappa}{2}\alpha\bar \alpha=p$.
So for these special real values of $p$ and $q$ we have
$$
\mathbb E\left[|z|^q\frac{|f'(z)|^p}{|f(z)|^q}\right]=\frac{|(1-z)^\alpha|^2}{(1-|z|^2)^{\frac{\kappa|\alpha|^2}{2}}}.
$$
Notice also that $\Re{\alpha}>0$, so that the singularity at $z=1$ does not contribute to the circle integral 
$$
 \int_{|z|=r<1}\frac{|(1-z)^\alpha|^2}{(1-|z|^2)^{\frac{\kappa|\alpha|^2}{2}}}|dz|\asymp_{r\to 1^-} (1-r)^{-\frac{\kappa|\alpha|^2}{2}}.
$$
So in the case \eqref{preal} \eqref{qreal} the averaged generalized spectrum is simply $\beta(p,q)=p$.
\end{proof}
\subsection{Check of integral means spectra on the integrable complex `red parabola'}\label{Redcheck}
  
As in \cite[Section 5.2.1]{DHLZ}, we will find that along the `red parabola' $\mathcal R$, a succession of explicit complex integral means spectra reproduces the result $\beta:=\kappa|\alpha|^2/2$ of Theorem \ref{thmrouge}. In addition to formulae \eqref{beta1complexe}, \eqref{tbis}, \eqref{beta1spiralecomplexe}, \eqref{tter} for the complex generalized spectrum $\beta_1$ of (drifted) whole-plane SLE, we shall need the SLE complex \emph{bulk} spectrum $\beta_0(p), p\in \mathbb C$, and some extensions of both $\beta_0$ and $\beta_1$ \cite[Section 5.1]{DHLZ}.
 
\subsubsection{SLE complex bulk spectrum}The SLE complex bulk spectrum $\beta_0(p), p\in \mathbb C$, can be obtained from the results of  \cite{PhysRevLett.89.264101,2008NuPhB.802..494D}, and reads \cite{BDone},
\begin{align}\label{beta0}
&\beta_0(p)=s_0(p)+\Re p -1,\\
&s_0(p)=s_0(t,\tilde t),\,\,\, t:=\Re p,\,\,\,\tilde t:=\Im p,
\end{align}
where the expression for $s_0(t,\tilde t)$ is
\begin{align}
\label{s0}
&s_0(t,{\tilde t})=1+b'-t+\sqrt{(b'-t)^2+{\tilde t}^2} -(2b')^{\frac{1}{2}}
\left[b'-t+\sqrt{(b'-t)^2+{\tilde t}^2}\right]^{\frac{1}{2}},\\ \label{bprime}
&s_0(t,{0})=1+2(b'-t) -2\sqrt{b'}\sqrt{b'-t},\quad \,\,\,
b':=\frac{(4+\kappa)^2}{8\kappa}.
\end{align}
By introducing the variables,
\begin{eqnarray}
\label{tau0}
\tau_0&:=&b'-t,\\
\label{taubar}
\bar{\tau}&:=&\frac{1}{2}\left[b'-t+\sqrt{(b'-t)^2+{\tilde t}^2}\right]
 =\frac{1}{2}\left(\tau_0+\sqrt{\tau_0^2+{\tilde t}^2}\right),
\end{eqnarray}
the function $s_0(t,\tilde t)$ \eqref{s0} can then be recast as a function of the \emph{single} variable
 $\bar\tau$, as
\begin{eqnarray}
\label{staubar}
 s_0(t,{\tilde t})&=&s(\bar \tau):=1+2\bar{\tau} -2\sqrt{b'}\sqrt{\bar{\tau}},\\ \nonumber
s_0(t,{0})&=&s(\tau_0)=1+2{\tau_0} -2\sqrt{b'}\sqrt{{\tau_0}}.
\end{eqnarray} 
\subsubsection{Extensions of complex spectra $\beta_0$ and $\beta_1$} As in Refs.\cite{DNNZ,DHLZ} it is natural to define auxiliary \emph{pseudo}-integral means spectra, 
which help in understanding phase transitions that are mediated by overlaps between various analytic expressions of the spectra. They are obtained by restoring the usual sign indeterminacy in front of square root operations \cite[Section 4.2]{DNNZ}, \cite[Section 5.1]{DHLZ}. Let us define the auxiliary functions,
\begin{align}\label{beta0pm}
\beta^{\pm}_0(p)&:=s^{\pm}_0(p)+\Re p -1,\,\,\, p\in \mathbb C,\\ \label{s0pm}
s^{\pm}_0(p)&=s^{\pm}_0(t,\tilde t)=s^{\pm}(\bar \tau):=1+2\bar{\tau} \pm 2\sqrt{b'}\sqrt{\bar{\tau}},\\\nonumber 
\bar{\tau}&=\frac{1}{2}\big(\Re (b'-p) +|b'-p|\big),\,\,\,\,\,\,\quad \quad b'=\frac{(4+\kappa)^2}{8\kappa}\,\,\,\\ \nonumber
&=\frac{1}{2}\left(b'-t+\sqrt{(b'-t)^2+{\tilde t}^2}\right),\,\, t:=\Re p,\,\,\tilde t:=\Im p,
\end{align}
such that the complex bulk integral means spectrum \eqref{beta0} is given by the $(-)$-branch, $\beta_0\equiv\beta_0^{-}$.
Similarly, we define 
\begin{align}\label{beta1spiralecomplexeext}
&\beta^{\pm}_1(p,q;\kappa,a):= s_1^{\pm}(p-q;\kappa,a)+\Re p-1,\,\,\, p,q\in \mathbb C,\\ \label{s1pm}
&s_1^{\pm}(p-q;\kappa,a)=s_1^{\pm}(\tau):=2\tau+\frac 12\mp \frac 12\sqrt{1+2\kappa \tau},\\ \label{tterext} 
&1+2\kappa \tau=\frac{1}{2}\left\{\Re\left[(1+ia)^2+2\kappa (p-q)\right]+\left |(1+ia)^2+2\kappa(p-q)\right |\right\},
\end{align}
such that the complex generalized spectrum \eqref{beta1spiralecomplexe} associated with spiral whole-plane SLE is given by the $(+)$-branch, $\beta_1\equiv\beta_1^{+}$.
\subsubsection{Complex spectra $\beta_1^{\pm}$ along $\mathcal R$}
From parameterization \eqref{Red}, we first find the identity along the red parabola $\mathcal R$,
$$(1+ia)^2+2\kappa (p-q)=(1+ia-\kappa \alpha)^2.$$  
Using the general identity, 
\begin{equation}\label{rez2} \frac{1}{2}\left[\Re (z^2)+|z^2|\right]=(\Re z)^2, z\in \mathbb C,\end{equation}
we find for \eqref{tterext},
$$1+2\kappa \tau=[\Re (1+ia-\kappa \alpha)]^2=(1-\kappa \Re \alpha)^2,$$
so that $s_1^{\pm}$ \eqref{s1pm} reads
$$s_1^{\pm}(\tau)= \frac{1}{\kappa}\left[(1-\kappa \Re \alpha)^2-1\right]+\frac{1}{2}\mp\frac{1}{2}|1-\kappa \Re \alpha|.$$
We simultaneously have from \eqref{Red} and \eqref{rez2}, 
\begin{equation}\label{rep}
\Re p=-\frac{\kappa}{2}\Re (\alpha^2)+\left(2+\frac{\kappa}{2}\right)\Re \alpha=\frac{\kappa}{2}|\alpha^2|-\kappa(\Re \alpha)^2+\left(2+\frac{\kappa}{2}\right)\Re \alpha.
\end{equation}
Combining the last two equations gives
$$s_1^{\pm}(\tau)+\Re p-1=\frac{\kappa}{2}|\alpha^2|+\frac{\kappa}{2}\Re \alpha-\frac{1}{2}\mp\frac{1}{2}|1-\kappa \Re \alpha|.$$
Therefore, we get for \eqref{beta1spiralecomplexeext} the branch-dependent identity,
\begin{align}
\label{beta1pmred}\beta^{\pm}_1(p,q;\kappa,a)=s_1^{\pm}(\tau)+\Re p-1=\frac{\kappa}{2}|\alpha^2|, \,\,\,\kappa \Re \alpha \gtreqqless 1.
\end{align}
This shows that the result of Theorem \ref{thmrouge} for spiral whole-plane SLE is recovered for $\Re \alpha \geq 1/\kappa$ by the `physical' branch $\beta_1^+$ of the generalized complex spectrum, and for $\Re \alpha \leq 1/\kappa$ by its `unphysical' branch $\beta_1^-$, in a way entirely similar to the real case studied in \cite[Section 5.2.1]{DHLZ}.
\subsubsection{Complex spectra $\beta^{\pm}_0$ along $\mathcal R$}
From parameterization \eqref{Red} we first get the identity,
$$b'-p=\frac{\kappa}{2}\left(\alpha-\frac{4+\kappa}{2\kappa}\right)^2,$$
from which we deduce with the help of \eqref{rez2},
$$\bar \tau=\frac{1}{2}\big(\Re (b'-p) +|b'-p|\big)=\frac{\kappa}{2} \left(\Re \alpha-\frac{4+\kappa}{2\kappa}\right)^2.$$
This in turn gives  
\begin{align*}
s_0^{\pm}(\bar \tau)=1+\kappa \left(\Re \alpha-\frac{4+\kappa}{2\kappa}\right)^2\pm \left(2+\frac{\kappa}{2}\right)\left|\Re \alpha-\frac{4+\kappa}{2\kappa}\right|,
\end{align*}
which, together with \eqref{rep} yields
 \begin{align*}
s^{\pm}_0(\bar \tau)+\Re p -1=\frac{\kappa}{2}|\alpha^2|-\left(2+\frac{\kappa}{2}\right)\Re \alpha +\kappa\left(\frac{4+\kappa}{2\kappa}\right)^2\pm \left(2+\frac{\kappa}{2}\right)\left|\Re \alpha-\frac{4+\kappa}{2\kappa}\right|.
\end{align*}
Thus we find the branch-dependent identity,
\begin{align}
\label{beta0pmred}\beta^{\pm}_0(p;\kappa)=s_0^{\pm}(\bar{\tau})+\Re p-1=\frac{\kappa}{2}|\alpha^2|, \,\,\,\kappa \Re \alpha \gtreqqless 2+\frac{\kappa}{2}.
\end{align}
We thus see that the result of Theorem \ref{thmrouge} for spiral whole-plane SLE is recovered for $\Re \alpha \leq 2/\kappa+1/2$ by the `physical' branch $\beta_0^-$ of the standard complex spectrum, and for $\Re \alpha \geq 2/\kappa+1/2$ by its `unphysical' branch $\beta_0^+$, in a way again similar to the real case studied in \cite[Section 5.2.1]{DHLZ}. We thus arrive at
\begin{prop}\label{overlap}Along the red parabola $\mathcal R$ \eqref{Red}, the integral means spectrum of the drifted whole-plane SLE is successively given by
\begin{align*}
&\beta^{-}_0(p;\kappa)=\beta^{-}_1(p,q;\kappa,a)=\frac{\kappa}{2}|\alpha^2|,\,\,\, \Re \alpha \in (-\infty,{1}/{\kappa}],\\
&\beta^{-}_0(p;\kappa)=\beta^{+}_1(p,q;\kappa,a)=\frac{\kappa}{2}|\alpha^2|,\,\,\, \Re \alpha \in [1/\kappa,2/\kappa+1/2],\\
&\beta^{+}_0(p;\kappa)=\beta^{+}_1(p,q;\kappa,a)=\frac{\kappa}{2}|\alpha^2|,\,\,\, \Re \alpha \in [2/\kappa+1/2,+\infty).
\end{align*}
\end{prop}
\begin{rkk}The two `physical' integral means spectra $\beta_0^-$ \eqref{beta0} and $\beta_1^+$ \eqref{beta1spiralecomplexe} overlap along the red parabola $\mathcal R$ \eqref{Red} in the interval $\Re \alpha \in [1/\kappa,2/\kappa+1/2]$,  a result which can be directly compared to \cite[Eqs. (93)-(95)]{DHLZ}. The integral means spectrum is always given by one of those two `physical' spectra, which coincides with the 'physical' branch of the other spectrum in the preceding overlap interval, or with its `unphysical' branch outside the said interval.
\end{rkk}
This corresponds  to the presence of a two-dimensional ``overlap ribbon" on the red parabola, where the complex generalized integral means spectrum takes both the $\beta_0^-$ and $\beta_1^+$ forms. In the 4-dimensional $(p,q)$ space, there exists a larger \emph{phase-transition} manifold, that is defined by the single condition that these two spectra are equal. This three-dimensional manifold must  intersect the above overlap ribbon on a certain phase-transition \emph{line}. The study of such phase-transition manifolds is left to a future work.
\section{General L\'evy processes with special symbols}
In this section, we generalize the results in \cite{DHLZ}, \cite{DNNZ},\ \cite{LY4}, \cite{LY1},\cite{LY2} and \cite{LY3} to the generalized integral means spectrum: in other words, we investigate the  values of $(p, q)$ for which the generalized integral means spectrum for L\'evy-Loewner evolution has an exact form.

 For this purpose, we assume in this section that  $G(z,\bar{z})$ \eqref{eq:GL} may be written as
$$
G(z,\bar{z})=(1-z)(1-\bar{z})h(z,\bar{z}),
$$
where $h(z,\bar z)$ is \emph{separately analytic} with respect to $z$ and $\bar z$ and satisfies the boundary condition $h(0,0)=1$. By applying \eqref{generaloperator}, we get
\begin{align*}
&\Lambda \left[(1-z)(1-\bar{z})h\right]\\
+&\left[\bar{q}(1-z)-\bar{p}\frac{1-z}{1-\bar{z}}+q(1-\bar{z})-p\frac{1-\bar{z}}{1-z}+(\bar{p}-\bar{q}+p-q)(1-z)(1-\bar{z})\right]h\\
&+\frac{1-z}{1-\bar{z}}\bar{z}(1+\bar{z})(h+(\bar{z}-1)\partial_{\bar{z}}h)+\frac{1-\bar{z}}{1-z}z(1+z)(h+(z-1)\partial_z h)=0
\end{align*}
 The  coefficient of $h=h(z,\bar{z})$ in this equation is the sum of a polynomial in $z$, $\bar z$, and of the polar part,
 \begin{equation*}
 \frac{1-\bar{z}}{1-z}\left[-p+z(1+z)\right]+\frac{1-z}{1-\bar{z}}\left[-\bar p+\bar{z}(1+\bar{z})\right].
 \end{equation*}
The latter clearly becomes pole free, i.e., a polynomial in $z$, $\bar z$ if and only if $p=\bar p=2$,
 which we shall hereafter assume. Under this condition, the above equation becomes
 \begin{multline}
\Lambda[(1-z)(1-\bar{z})h]+(z+1)(\bar{z}-1)z\partial_zh+
(\bar{z}+1)(z-1)\bar{z}\partial_{\bar{z}}h\\+ \left[(z-1)(3-\bar{q})\bar{z}+(\bar{z}-1)(3-q)z\right]h=0.
\end{multline}
Besides the restriction to $p=2$, we shall also assume that $q\in\R$ and that the L\'evy process $L_t$ is symmetric. We then get
\begin{multline}\label{pdebyh}
\Lambda[(1-z)(1-\bar{z})h]+(z+1)(\bar{z}-1)z\partial_zh+
(\bar{z}+1)(z-1)\bar{z}\partial_{\bar{z}}h\\+(3-q)(2z\bar{z}-z-\bar{z})h=0.
\end{multline}
In order to analyze this equation, we use the Fourier expansion of $t\mapsto h(re^{it},re^{-it})$:
\begin{equation} \label{fourier}
h(z,\bar z)=\sum_{n=-\infty}^{+\infty}\theta_n(\xi)z^n,\;\xi:=z\bar z,\,z=re^{it}.
\end{equation}
When replacing $h$ by this expansion in the equation, we get a recursion formula between the $\theta_n$'s for $n\in \mathbb Z$. 
More precisely, by writing that the $n$'th Fourier coefficient of the left side of \eqref{pdebyh} vanishes, we obtain for $n \in \mathbb Z$,
\begin{align} \nonumber
&2\xi(\xi-1)\theta_n'(\xi)-\Big(\eta_n+n+(\eta_n+2q-n-6)\xi\Big)\theta_n\\ \label{equn}
&+\xi\Big(\eta_n+n+q-2\Big)\theta_{n+1}(\xi)+(\eta_n-n+q-2)\theta_{n-1}(\xi)=0.
\end{align}

Note  that the assumption that $L_t$ is symmetric implies that $h(z,\bar z)$ is symmetric w.r.t. $z$ and $\bar z$, which translates into,
\begin{equation}\label{-nn}
\theta_{-n}(\xi)=\xi^n\theta_n(\xi),
\end{equation}
from which we may simply recast expansion \eqref{fourier} above as
\begin{equation}\label{thetaseries} h(z,\bar z)=\theta_0(\xi)+\sum_{n=1}^\infty \theta_n(\xi)(z^n+\bar{z}^n)=\theta_0(\xi)+\sum_{n=1}^\infty 2\theta_n(\xi)r^n\cos{nt}.\end{equation}
Before continuing, let us recall that we are looking for the integral means, i.e., the angular integrals,
$$I(r)=\int_{0}^{2\pi}G(re^{it},re^{-it})dt,$$
that can be easily expressed in terms of $\theta_j$'s as 
\begin{equation}\label{eqn:I}
\frac{I(r)}{2\pi}=(1+r^2)\theta_0(r^2)-2r^2\theta_1(r^2),
\end{equation}
so that we only need to compute $\theta_0$ and $\theta_1$. For later purposes, let us also mention that $\theta_0(0)=1$. 
We thus focus on the equations for  $n=0$ and $n=1$ (recall that $\eta_0=0$ and that $\theta_{-1}(\xi)=\xi\theta_1(\xi)$),
\begin{align}\label{equ0}
&(\xi-1)\theta_0'(\xi)-(q-3)\theta_0(\xi)+(q-2)\theta_1(\xi)=0,\\ \nonumber
&2\xi(\xi-1)\theta_1'(\xi)-\left[\eta_1+1+(\eta_1+2q-7)\xi\right]\theta_1(\xi)+(\eta_1+q-1)\xi\theta_2(\xi)\\ \label{eqn:1}
&+(\eta_1+q-3)\theta_0(\xi)=0.
\end{align}
There are two simple cases where we can explicitly compute $\theta_0$ and $\theta_1$.
\begin{enumerate}
\item The first case is when the coefficient of the $\theta_0$-term in the second equation vanishes, i.e., when 
\begin{equation}\label{eta3-q}\eta_1=3-q,
\end{equation}
 which requires $q<3$. 
In this case we may take $\theta_1=0$ (and actually $\theta_n=0$ for $n\geq 1$) and $\theta_0$ to be the solution to 
$$(\xi-1)\theta_0'(\xi)-(q-3)\theta_0(\xi)=0,\;\theta_0(0)=1.$$
This gives 
\begin{equation}\label{3-q}h(z,\bar z)=\theta_0(\xi)=\frac{1}{(1-\xi)^{3-q}},
\end{equation}
and
$$G(z,\bar z)=\frac{(1-z)(1-\bar z)}{(1-z\bar z)^{3-q}},$$ 
so that $\beta(2,q)=3-q >0$. 
\item The second case is by letting the coefficient of the $\theta_2$-term vanish in the second equation, i.e., by taking 
\begin{equation}\label{eta1-q}\eta_1=1-q,
\end{equation}
which requires $q\leq 1$. 
We then get a system of coupled ODEs for $\theta_0$ and $\theta_1$, which we must solve with initial data $\theta_0(0)=1$, $\theta_1(0)$ finite. We find
\begin{equation}\label{1-q} 
\begin{cases}
\theta_0(\xi)=(1+\xi)(1-\xi)^{-(4-q)}\\[4mm]
\theta_1(\xi)=-\frac{2}{2-q}(1-\xi)^{-(4-q)},
\end{cases}
\end{equation}
from which we deduce  that $\beta(2,q)=4-q (>3)$.
\end{enumerate}
Eqs. \eqref{3-q} and \eqref{1-q}  generalize results of Ref. \cite{LY2}  to $q\neq 0$. 

As noticed in \cite{LY3}, the preceding method generalizes: more precisely for any $n\geq 1$, if we let the coefficient of the $\theta_{n-1}$-term vanish in the $n$th equation, i.e., by taking $\eta_n=2-q+n$, then the solution of the system is $\theta_p=0$ for $p\geq n$, while $\theta_0,..\theta_{n-1}$ are the solutions of the $n$ first equations, with the initial data $\theta_0(0)=1$.\\
Another possible generalization is by letting vanish, again in the $n$th equation, the coefficient of the $\theta_{n+1}$-term, i.e., by taking $\eta_n=2-q-n$; then the $n$ first equations allow us to compute $\theta_0,..\theta_{n-1}$, which is more than needed since we only need to know $\theta_0$ and $\theta_1$. Having dealt in the last section with the $n=1$ case, let us now investigate the $n=2$ case.
\subsection{The $\eta_2=4-q$ case}\label{4-qcase}
We take $\theta_{n\geq 2}=0$, and have to solve the following system of differential equations:
\begin{align*}
\begin{cases}
(x-1)\theta_0'(x)+(3-q)\theta_0(x)+(q-2)\theta_1(x)=0\\
2x(x-1)\theta_1'(x)-\left[1+\eta_1+(\eta_1+2q-7)x\right]\theta_1(x)+(\eta_1-3+q)\theta_0(x)=0.
\end{cases}
\end{align*}
We can assume the existence of a real number $\delta$ (to be determined later) and of two functions $f_0, f_1:\,[0,1]\to \R$ such that $\theta_0$ and $\theta_1$ have the following form,
\begin{equation*}
\theta_0(x)=(1-x)^{-\delta}f_0(x), \quad \theta_1(x)=(1-x)^{-\delta}f_1(x).
\end{equation*}
The ODE system satisfied by $f_0$ and $f_1$ is then
\begin{align}\label{eqn::a10}
\begin{cases}
(x-1)f_0'+(3-q-\delta)f_0+(q-2)f_1=0\\
2x(x-1)f_1'-\left[1+\eta_1+(\eta_1+2q-7+2\delta)x\right]f_1+(\eta_1-3+q)f_0=0.
\end{cases}
\end{align}
The form of the equation coefficients suggests we define $\delta'$ as,
\begin{equation}
\label{dd'}
\delta':=\delta+q-3,
\end{equation}
so that 
\begin{align}\label{eqn::a1}
\begin{cases}
(x-1)f_0'-\delta' f_0+(q-2)f_1=0\\
2x(x-1)f_1'-\big(1+\eta_1+(\eta_1-1+2\delta')x\big)f_1+(\eta_1-3+q)f_0=0.
\end{cases}
\end{align}
From the first equation we extract $f_1$ as a function of $f_0$ and $f_0'$ for $q\neq 2$,
\begin{equation}\label{f1}
f_1(x)=\frac{x-1}{2-q}f_0'(x)-\frac{\delta'}{2-q}f_0(x),\quad q\neq 2.
\end{equation}
Substituting into \eqref{eqn::a1}, we obtain the following degree two differential equation satisfied by $f_0$, 
\begin{align}\label{eqn::a2}
&\frac{2x(x-1)^2}{2-q}f_0''(x)+\left(2x(x-1)\frac{1-\delta'}{2-q}-\frac{x-1}{2-q}\left(1+\eta_1+(\eta_1-1+2\delta')x\right)\right)f_0'(x)+\nonumber\\
&\left(\eta_1+q-3+\frac{\delta'}{2-q}\left(1+\eta_1+(\eta_1-1+2\delta')x\right)\right)f_0(x)=0.
\end{align}
We want to find $\delta'$ such that this equation reads 
\begin{align}\label{eqn::hp}
x(x-1)f_0''(x)+((a+b+1)x-c)f_0'(x)+abf_0(x)=0,
\end{align}
so that $f_0(x)=\,_2F_1(a,b,c;x)$, the hypergeometric function.
This  identification shows that $\delta'$ must obey the following relation, 
\begin{align}\label{eqn::b2}
E(\delta',\eta_1,q):=2\delta'(\delta'+\eta_1)+(2-q)(\eta_1+q-3)=0, 
\end{align}
so that \eqref{eqn::a2} simplifies into
\begin{align}\label{eqn::eqn::a8}
&x(x-1)f_0''(x)+\left(\frac{3-\eta_1-4\delta'}{2}x-\frac{1+\eta_1}{2}\right)f_0'(x)\\
&-\delta'\,\frac{1-\eta_1-2\delta'}{2}f_0(x)=0.
\end{align}
This yields
\begin{align}\label{eqn:param}
a=-\delta',\quad b=\frac{1-\eta_1}{2}-\delta', \quad c=\frac{1+\eta_1}{2},
\end{align}
and we can choose $f_0(x)= 
\hypgeo{2}{1}(a,b;c;x)$ and $\theta_0(x)=(1-x)^{-\delta} \hypgeo{2}{1}(a,b;c;x)$.  
The other independent solution to the hypergeometric equation is $$x^{1-c}\hypgeo{2}{1}(1+a-c,1+b-c;2-c;x),$$ which is \emph{non-analytic} at the origin $x=0$, hence is discarded as a candidate for $f_0$.\\

Before continuing, let us consider general symmetric L\'evy processes: what are the possible couples $(\eta_1,\eta_2)$? We know that we must have $\eta_2\geq 0$ and the L\'evy-Khinchine formula implies that $\eta_1\geq \eta_2/4$. It happens  that every couple $(\eta_1,\eta_2)$ such that $\eta_1\geq\eta_2/4\geq 0$ actually corresponds to some (symmetric) L\'evy process \cite{applebaum2009} \footnote{R\'emy Rhodes, private communication.}. Moreover the case $\eta_1=\eta_2/4$ exactly corresponds to an SLE process, while in the case  $\eta_2=0$, $L_t$ is a pure jump process with jumps equal to $k\pi$, $k$ odd (notice that the other case is similar with $k$ even, yielding a continuous process on the circle).

For $\eta_2=4-q$, the preceeding constraints on $(\eta_1,\eta_2)$ become $q\leq 4$ and $\eta_1\geq\eta_2/4=1-q/4$, with equality for SLE$_\kappa$, with $\kappa=2-q/2$. Let us then define 
\begin{equation}\label{D4-q}
D_{4-q}:=\{(q,\eta_1)\in\R^2;\,q\leq 4 ,\,\eta_1\geq 1- q/4\}.
\end{equation}
Let us return to Eq. \eqref{eqn::b2}. It can be shown (see below) that $\eta_1^2-2(2-q)(\eta_1+q-3)\geq0$ in $D_{4-q}$; we may thus extract $\delta'$ from \eqref{eqn::b2}:
\begin{align} \nonumber
\delta_\pm&=\delta'_\pm+3-q\\ \label{eqn:delta'}
\delta'_\pm&=\delta'_\pm(q, \eta_1):=\frac{1}{2}\left(-\eta_1\pm\sqrt{\eta_1^2-2(2-q)(\eta_1+q-3)}\right).
\end{align}
Replacing in \eqref{eqn:param} $\delta'$ by its value in terms of $\eta_1,q$, we get:
\begin{align}\label{apm}
&a_{\pm}=a_{\pm}(\eta_1, q):=\frac{1}{2} \eta_1 \mp\frac{1}{2}\sqrt{\eta_1^2-2(2-q)(\eta_1+q-3)}\\ \label{bpm}
& b_{\pm}=b_{\pm}(\eta_1, q):=\frac{1}{2} \mp \frac{1}{2} \sqrt{\eta_1^2-2(2-q)(\eta_1+q-3)}\\ \label{cnpm}
&c=c(\eta_1, q):=\frac{1+\eta_1}{2}.
\end{align}
We can now compute $f_1$ as,
\begin{align}
f_1(x)&=\frac{x-1}{2-q}f_0'(x)+\frac{-\delta'}{2-q}f_0(x)\nonumber\\
&=\frac{a}{2-q}\hypgeo{2}{1}(a,b;c;x)+\frac{ab(x-1)}{c(2-q)}\hypgeo{2}{1}(a+1,b+1;c+1;x)\label{f_1}.
\end{align}
Using expression \eqref{eqn:I} for the integral means, we get
\begin{align}\label{Ir}
\frac{1}{2\pi}I(r)=(1-r^2)^{-\delta}\left[(1+r^2)f_0(r^2)-2r^2f_1(r^2)\right].
\end{align}
\subsubsection{The $(+)$-branch}\label{+branch}
Let us first consider the case where $$a=a_+,\, b=b_+.$$
 We then have $c-a_+ - b_+=\sqrt{{\eta_1^2-2(2-q)(\eta_1+q-3)}}>0$, so that 
 \begin{equation}\label{Gamma1} 
 0<\hypgeo{2}{1}(a,b;c;1)=\frac{\Gamma(c)\Gamma(c-a-b)}{\Gamma(c-a)\Gamma(c-b)}<\infty.
 \end{equation}
If furthermore, $a+b+1< c$, then $\hypgeo{2}{1}(a+1,b+1;c+1;1)<\infty$, and
\[
\Delta:=\lim_{r\rightarrow 1} \left((1+r^2)f_0(r^2)-2r^2f_1(r^2)\right)=2\left(1-\frac{a}{2-q}\right)\hypgeo{2}{1}(a,b;c;1)<\infty.
\]
  We have thus proven that if $\Delta> 0$, the spectrum is
\[
\beta(2, q)=3-q+\delta_{+}'(\eta_1,q)=3-q+\frac{1}{2}\left(\sqrt{\eta_1^2-2(2-q)(\eta_1+q-3)}-\eta_1\right).
\]
One can check that the coefficient $A:=1-a_+/(2-q)>0$ for $q<3$, whereas for $q\geq 3$, it vanishes for $\eta_1=1-q$ or $\eta_1=3-q$.  In $D_{4-q}$, $A$ is thus positive since $\eta_1 \geq 1-q/4$ is located outside the non-positive interval $[1-q,3-q]$.  So we get that the spectrum  is equal to $\delta_+$ on the subset of $D_{4-q}$ of points for which $a+b+1<c$.  
This condition is just $Z>1$, with 
\begin{equation}
\label{Z}
Z=Z(\eta_1,q):={\eta_1^2-2(2-q)(\eta_1+q-3)}.
\end{equation} Note that $Z$ can be written as 
\begin{align}\label{Zbis}
Z&=(q-3)^2+(\eta_1+q-2)^2-1\\ 
\label{Zter}
&=X^2+Y^2-1, \quad X:=q-3, \quad Y:=\eta_1+q-2,
\end{align}
so that the condition $Z>1$ corresponds to the \emph{exterior} of the \textcolor{blue}{\bf blue ellipse} of equation $Z=1$,  i.e., $X^2+Y^2=2$ in the $(q,\eta_1)$-plane (Fig. \ref{qeta1}). 

Notice that the set $Z=0$ is the co-centered \textcolor{ao(english)}{\bf green ellipse} of equation $X^2+Y^2=1$, which only intersects $D_{4-q}$ at the tangency point $(\frac{12}{5},\frac{2}{5})$ with the $(\eta_1=1-q/4)$-line, implying that $Z\geq 0$ in $D_{4-q}$, as mentioned above.

In the \emph{interior} of the blue ellipse, we instead have $a+b< c< a+b+1$, and we use the Euler transformation,
\begin{equation}\label{euler}
\hypgeo{2}{1}(a,b;c;x)=(1-x)^{c-a-b}\hypgeo{2}{1}(c-a,c-b;c;x),
\end{equation}
so that 
\begin{equation}\label{euler+}
\hypgeo{2}{1}(a+1,b+1;c+1;x)=(1-x)^{c-a-b-1}\hypgeo{2}{1}(c-a,c-b;c+1;x).
\end{equation}
We then get from \eqref{f_1} 
\begin{align}
f_1(x)
=\frac{a}{2-q}\hypgeo{2}{1}(a,b;c;x)-\frac{ab}{c}\frac{1}{(2-q)}(1-x)^{c-a-b}\hypgeo{2}{1}(c-a,c-b;c+1;x)\label{f_1+}.
\end{align}
We now have 
$\hypgeo{2}{1}(a,b;c;1)<\infty$ and $\hypgeo{2}{1}(c-a,c-b;c+1;1)<\infty$, so the second term in \eqref{f_1+} still vanishes as $x\to 1$. This again yields 
\[
0<\Delta=2\left(1-\frac{a}{2-q}\right)\hypgeo{2}{1}(a,b;c;1)<\infty,
\]
i.e., the same result as found outside the blue ellipse. We thus find for the $(+)$-branch, 
\begin{equation} \label{Ireq}\frac{1}{2\pi}I(r)\sim 2 (1-r^2)^{-\delta_{+}}\left[1-\frac{a_+}{2-q}\right] \frac{\Gamma(c)\Gamma(c-a_+-b_+)}{\Gamma(c-a_+)\Gamma(c-b_+)},\quad r\to 1^-,\quad q\neq 2.
\end{equation}
\subsubsection{The $(-)$-branch}
Let us now consider the other possible choice, $$a=a_{-},\, b=b_{-}.$$
We then use  both \eqref{euler} and \eqref{euler+} in \eqref{f_1}, 
\begin{align}
f_1(x)
=(1-x)^{c-a-b}\left[\frac{a}{2-q}\hypgeo{2}{1}(c-a,c-b;c;x)-\frac{ab}{c}\frac{1}{(2-q)}\hypgeo{2}{1}(c-a,c-b;c+1;x)\right]\label{f_1-},
\end{align}
where now $c-a-b<0$, and where 
\begin{align*}&\hypgeo{2}{1}(c-a,c-b;c;1)=\frac{\Gamma(c)\Gamma(a+b-c)}{\Gamma(a)\Gamma(b)}<\infty,\\
&\hypgeo{2}{1}(c-a,c-b;c+1;1)=\frac{\Gamma(c+1)\Gamma(a+b-c+1)}{\Gamma(a+1)\Gamma(b+1)}<\infty.
\end{align*}
From the well-known identity $\Gamma(x+1)=x\Gamma(x)$, we finally get for \eqref{f_1-}
$$f_1(x)=(1-x)^{c-a-b}\frac{c-b}{2-q} \frac{\Gamma(c)\Gamma(a+b-c)}{\Gamma(a)\Gamma(b)}.$$
$I(r)$ \eqref{Ir} is now equivalent for $r^2=x\to 1$ to
$$\frac{1}{2\pi}I(r)\sim 2 (1-x)^{-\delta_{-}+c-a-b}\left[1-\frac{c-b}{2-q}\right] \frac{\Gamma(c)\Gamma(a+b-c)}{\Gamma(a)\Gamma(b)}.$$
where we recall that $a=a_{-}, b=b_{-}$. We now use the \emph{duality} formulae
$$c-a_{-}=b_+,\,\,c-b_{-}=a_+,\,\,a_{-}+b_{-}-c=c-a_+-b_+,$$
and $$\delta_{\pm}=3-q+\delta'_{\pm},\,\,\delta'_{-}-\delta'_+=c-a_{-}-b_{-},$$
so that for $r\to 1^-$,
$$\frac{1}{2\pi}I(r)\sim 2 (1-r^2)^{-\delta_{+}}\left[1-\frac{a_+}{2-q}\right] \frac{\Gamma(c)\Gamma(c-a_+-b_+)}{\Gamma(c-a_+)\Gamma(c-b_+)},$$
which is exactly the same as the result \eqref{Ireq} for the $(+)$-choice in Section \ref{+branch}. 
\begin{rkk}{\bf The $q=2$ case.}\label{q=2}
Up to now, we have assumed that $q\neq 2$. The solution for $q=2$, thus  $\eta_2=2, \eta_1\geq 1/2$, can be obtained by continuity as the $q\to 2$ limit of $f_0$ and $f_1$ \eqref{f1}. Eq. \eqref{apm} gives 
\begin{align*}
&a_+ =\frac{1}{2\eta_1}(\eta_1-1)(q-2) +\mathcal O\left((q-2)^2\right),\\
&b_+=\frac{1}{2}(1-\eta_1)+\mathcal O(q-2),\quad c=\frac{1}{2}(1+\eta_1),
\end{align*}
so that $f_0=1$, and \eqref{f1} has a finite limit when $q\to 2$,
\begin{equation}\label{f1q=2} 
f_1(x)=\frac{\eta_1-1}{2\eta_1}\left[1-\frac{1-\eta_1}{1+\eta_1}(1-x)\hypgeo{2}{1}(1,\tfrac{1}{2}(3-\eta_1);\tfrac{1}{2}(3+\eta_1);x)\right],
\end{equation}
together with $\delta'_+=0$ and $\delta_+=1$. Eq. \eqref{Ireq} simply becomes 
$\frac{1}{2\pi}I(r)\sim  (1-r^2)^{-1}\left(3-1/\eta_1\right)$, as $r\to 1^-$, so that $\beta(p=2,q=2)=1$.
\end{rkk}
We therefore proved the following
\begin{theo}
For a L\'evy process, with symbols $\eta_2=4-q,\,\,\eta_1\geq \eta_2 /4$ for $q \leq 4$, the generalized integral means spectrum of the corresponding LLE is 
\begin{equation}\label{beta2q}\beta(2, q)=3-q+\frac{1}{2}\left(\sqrt{\eta_1^2-2(2-q)(\eta_1+q-3)}-\eta_1\right).\end{equation}
In particular, if $q=2,\eta_2=2$, the standard integral means spectrum at $p=2$ of the logarithm of the LLE is independent of $\eta_1 \geq 1/2$, and equal to $\beta(2,2)=1$. 
\end{theo}
In the $q=0$ case, Eq. \eqref{beta2q} recovers a result of \cite{LY3}.
\begin{figure}[h!]
\begin{center}
\includegraphics[width=9cm]{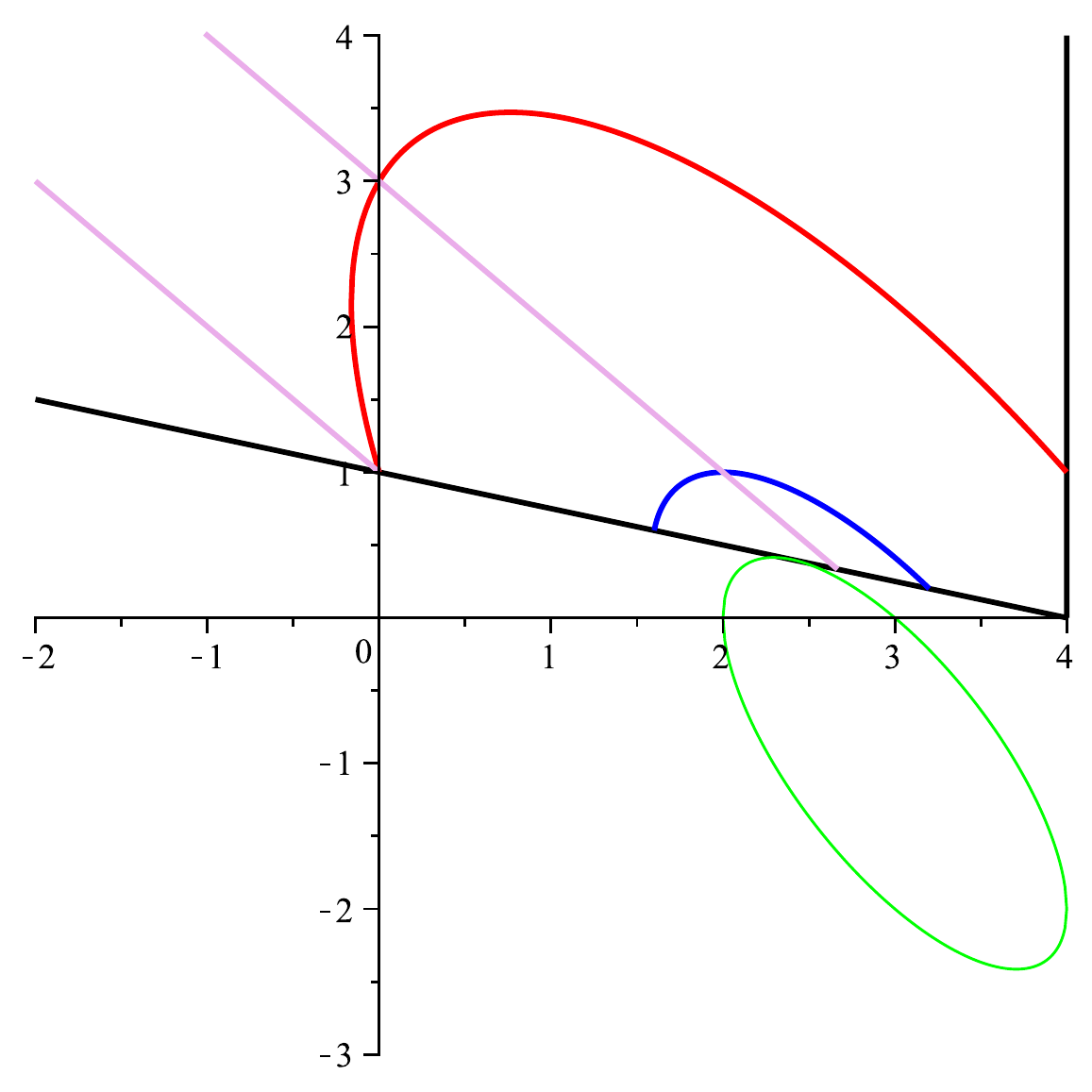}
\end{center}
\caption{Domain lines in the $(q,\eta_1)$ plane, for the $\eta_2=4-q$ case. The two purple straight lines have for equations, $\eta_1=1-q$ and $\eta_1=3-q$. The \textcolor{ao(english)}{\bf green}, \textcolor{blue}{\bf blue} and \textcolor{red}{\bf red} ellipses have for respective equations $X^2+Y^2=1,2,10$, with $X=q-3,Y=\eta_1+q-2$, with same center located at $(q,\eta_1)=(3,-1)$.}
\label{qeta1} 
\end{figure}
\subsubsection{Algebraic solutions}As a transition to the next section, let us look for purely algebraic solutions of the form,
$$\theta_j(x)=(1-x)^{-\alpha} f_j(x),\quad f_j(x)=A_j^0+A_j^1(1-x),\,j=0,1.$$
where $A_j^k,k=0,1$ are fixed coefficients and with the understanding that $\theta_j=0$ for $j\geq 2$. 
Recall that the hypergeometric function $\hypgeo{2}{1}$ is given by the well-known series expansion,
\begin{equation}\label{hypergeo}
\hypgeo{2}{1}(a,b;c;x)=\sum_{n=0}^\infty \frac{(a)_n(b)_n}{(c)_n} \frac{x^n}{n!}=1+\frac{a b}{c} \frac{x}{1!}+\frac{a(a+1)b(b+1)}{c(c+1)} \frac{x^2}{2!}+\cdots,
\end{equation} 
with 
\begin{align*}
(a)_n=
\begin{cases}
1\quad \quad \quad  \quad  \quad  \quad \quad \quad  \quad  \quad \,\, n=0\\
a(a+1)\cdots (a+n-1) \quad  n\geq 1,
\end{cases}
\end{align*}
so that the fact that $f_0(x)=\hypgeo{2}{1}(a,b;c;x)$ is at most linear in $x$ implies that either: $a=a_+=0$; $a_+=-1$; $b=b_+=0$; or $b_+=-1$. The function $f_1$ \eqref{f1} is then also linear, assuming for now that $q\neq 2$. For later convenience, let us write \eqref{apm} \eqref{bpm} as
\begin{align}\label{a+b+}
&a_+=\tfrac{1}{2}\eta_1- \tfrac{1}{2}{Z}^{1/2},\quad
b_+=\tfrac{1}{2}- \tfrac{1}{2}{Z}^{1/2},\\ \nonumber
&Z=\eta_1^2-2(2-q)(\eta_1+q-3).
\end{align}\\
$\bullet$ {\bf The $a_+=0$ case}.
The equation $Z=\eta_1^2$ gives for $q\neq 2$, $\eta_1=3-q$, which recovers the algebraic solutions, Eqs. \eqref{eta3-q} and \eqref{3-q}. For $q=2$, Remark \eqref{q=2} and Eq. \eqref{f1q=2} yield for $\eta_1=1$, $f_0=1,f_1=0$, in agreement with \eqref{3-q}.\\
 
\noindent$\bullet$ {\bf The $a_+=-1$ case}. The equation $Z=(\eta_1+2)^2$ yields $(4-q) (\eta_1+q-1)=0$. Hence we first recover the algebraic case $\eta_1=1-q$, as in Eqs. \eqref{eta1-q} and \eqref{1-q}. 
The other case, $q=4$, is the vertical boundary line for $D_{4-q}$ (Fig \ref{qeta1}), where $\eta_2=0$, and for which Eq. \eqref{beta2q} gives $\delta_+=\beta(p=2,q=4)=0$, so that $\theta_j=f_j,\, j=0,1$. One further finds $b_+=-c=-\frac{1}{2}(1+\eta_1)$, so that one gets the polynomial solutions, $f_0(x)=\hypgeo{2}{1}(-1,b_+;c;x)=1+x$, and  
from \eqref{f1} $f_1(x)=\frac{1}{2}(1+x)$.\\

\noindent$\bullet$ {\bf The $b_+=0$ case}. From \eqref{a+b+} we get the condition $Z=1$, which in parameterization \eqref{Zbis} of $Z$ \eqref{Z} is just that defining the \textcolor{blue}{\bf blue ellipse} as $X^2+Y^2=2$. Its solution is given by $\eta_1=2-q\pm\sqrt{-q^2+6q-7}$. The condition $(q,\eta_1)\in D_{4-q}$ \eqref{D4-q} selects the $(+)$-branch only,  and restricts the range of parameter $q$ to $q\in[\frac{8}{5},\frac{16}{5}]$ (see Fig. \ref{qeta1}). This yields a first line of algebraic solutions,
\begin{align*}
\begin{cases}
\eta_1&=2-q+\sqrt{1+(q-2)(4-q)}\\
\beta(2,q)&=\frac 12\left(5-q-\sqrt{1+(q-2)(4-q})\right)\\
q&\in [\frac{8}{5},\frac{16}{5}].
\end{cases}
\end{align*}
\\

\noindent$\bullet$ {\bf The $b_+=-1$ case}. From \eqref{a+b+} we get the condition $Z=9$, which in parameterization \eqref{Zbis} is defining a \textcolor{red}{\bf red ellipse} as $X^2+Y^2=10$. Its solution is given by $\eta_1=2-q\pm\sqrt{1+6q-q^2}$. The condition $(q,\eta_1)\in D_{4-q}$ \eqref{D4-q}  allows for both $(\pm)$-branches,  but restricts the range of parameter $q$ to $[3-\sqrt{10},4]$ for the 
$(+)$-branch, and to $[3-\sqrt{10},0]$ for the $(-)$-branch (see Fig. \ref{qeta1}). This finally gives algebraic solutions for,
\begin{align*}
\begin{cases}
\eta_1 &=2-q+\sqrt{1+6q-q^2}\\
\beta(2,q) &=\frac 12(7-q-\sqrt{1+6q-q^2})\\
q&\in[3-\sqrt{10},4],
\end{cases}
\end{align*}
and
\begin{align*}
\begin{cases}
\eta_1 &=2-q-\sqrt{1+6q-q^2}\\
\beta(2,q) &=\frac 12(7-q+\sqrt{1+6q-q^2})\\
q&\in[3-\sqrt{10},0].
\end{cases}
\end{align*}

From \eqref{hypergeo} and for $n\geq 2$, one further finds a whole series of algebraic solutions where the $f_j$'s are \emph{polynomials of degree} $n$, when either $(a)_{n+1}=0$ or $(b)_{n+1}=0$, i.e., either $a_+=-n$ or $b_+=-n$.\\

\noindent$\bullet$ {\bf The $a_+=-n$ case}. From \eqref{a+b+}, one finds 
\begin{equation}\label{eta1na}
\eta_1=\frac{(2-q)(3-q)-2n^2}{2-q+2n}.
\end{equation}
One recovers the two linear cases seen above,  $n=0$,  $\eta_1=3-q$ with $q\leq 8/3$ and  $n=1$, $\eta_1=1-q$ with $q\leq 0$. For $n \geq 2$, the L\'evy symbol condition $\eta_1\geq 1-q/4$ requires that $q\leq 2-2n$. Eq. \eqref{eta1na} can be recast as 
$$(2-q+2n)(\eta_1+q-3+2n)=2n(n-1).$$ This corresponds to a branch of a {\bf hyperbola} $\mathcal H_n$ in the $(q,\eta_1)$ plane, defined in affine coordinates by
\begin{equation}\label{Hn}
X_nY_n=2n(n-1),\,\,\, X_n:=2-q+2n,\,\,Y_n:=\eta_1+q-3+2n,\,\,\, q\leq 2-2n ,\,\, n\geq 2.
\end{equation}\\
\noindent$\bullet$ {\bf The $b_+=-n$ case}. From \eqref{a+b+} one readily finds
\begin{equation}\label{eta1nb}
Z=(2n+1)^2 ,
\end{equation}
which in parameterization \eqref{Zbis} is defining the {\bf ellipse} $\mathcal E_n$ by the equation $$X^2+Y^2=1+(2n+1)^2.$$This in turn yields
\begin{equation}\label{eta1nbis}
\eta_1=2-q\pm\sqrt{(2-q)(q-4)+(2n+1)^2},
\end{equation}
together with 
\begin{align}\label{beta1nb}
\beta(2, q)&=3-q+\frac{1}{2}\left(2n+1-\eta_1\right)\\ \nonumber
&=\frac{1}{2}\left(2n+5-q\mp \sqrt{(2-q)(q-4)+(2n+1)^2}\right).
\end{align}
The condition $(q,\eta_1)\in D_{4-q}$ \eqref{D4-q}  allows for both $(\pm)$-branches, but restricts for the 
$(+)$-branch the range of parameter $q$ to $[3-\sqrt{(2n+1)^2+1},4]$, and for the $(-)$-branch to $[3-\sqrt{(2n+1)^2+1},\tfrac{8}{5}(1-n)]$. For $n=0,1$ one recovers the blue  and red ellipses of Fig. \ref{qeta1}.
\\

Let us finally investigate what happens on the boundary of $D_{4-q}$.
\begin{enumerate}
\item  {\it On the $\eta_1=1-q/4$ line:}\\
 As we have already seen, this case occurs when $\eta_2=4\eta_1$ and the process is an SLE$_\kappa$ with $\kappa=2\eta_1=2-q/2$, for which we get from the above,
\begin{align}\label{SLEspec}
\begin{cases}
\beta(2,q)&=1-\frac{1}{4}q,\;q\geq \frac{12}{5},\\
&=4-\frac{3}{2}q,\;q\leq\frac{12}{5}.\\
\end{cases}
\end{align}
 It is known that the SLE generalized spectrum has several phases \cite{DHLZ}, among which, the standard `bulk' spectrum \cite{BS}, 
\begin{equation}\label{BS}
\beta_0(p;\kappa):=-p+\frac{(4+\kappa)^2}{4\kappa}-\frac{\kappa+4}{4\kappa}\sqrt{(4+\kappa)^2-8\kappa p},
\end{equation}
and the `unbounded whole-plane' one \cite{DHLZ}, 
\begin{equation}\label{eq:DHLZ}\beta_1(p,q;\kappa):=3p-2q-\frac 12-\frac 12\sqrt{1+2\kappa(p-q)}.
\end{equation}
We thus have 
\begin{align*}
\begin{cases}
\beta_0(2;\kappa=2-q/2)&=1-\frac{1}{4}q,\;q\geq -4,\\
&=\frac{16}{4-q},\;q\leq -4.\\
\beta_1(2,q;\kappa=2-q/2)&=4-\frac{3}{2}q,\;q\leq3,\\
&=7-\frac{5}{2}q,\;q\geq3.
\end{cases}
\end{align*}
Hence, on the SLE boundary line,  a \emph{phase transition} takes place at $q=12/5$ in the spectrum \eqref{SLEspec}, in the sense that 
$\beta(2,q)=\beta_1(2,q;\kappa=2-q/2)$ if $q\leq12/5$ and $\beta(2,q)=\beta_0(2;\kappa=2-q/2)$ if $q\geq 12/5$. One can check that the phase transition point $(p=2,q=12/5)$ is located on the so-called `green parabola' that  delineates the respective domains of validity of $\beta_0$ and $\beta_1$ for whole-plane SLE$_{\kappa=4/5}$ \cite[Sec. 5.2.2]{DHLZ}. 
\item {\it On the  $q=4$ line:}\\
Here, $\eta_2=0$ and the spectrum is $\beta(2,4)=0, \forall \eta_1$. This case corresponds to a pure jump L\'evy process, whereas the equality $q=2p=4$ corresponds to a bounded LLE process \cite{DHLZ}; the above result then agrees with \cite{Chen2008SchrammLoewnerED}.
\end{enumerate}
 Fig. \ref{qeta1} summarizes the results of this section, showing 
\begin{itemize}
\item the domain $D_{4-q}$, domain of validity of the hypergeometric analysis; 
\item the domain of definition of the square root involved in the expression of the spectrum $\beta(2,q)$ \eqref{beta2q}, which is the exterior of the \textcolor{ao(english)}{\bf green} ellipse (thus containing $D_{4-q}$);
\item the $\eta_1=1-q$ and $\eta_1=3-q$ lines, corresponding to degenerate hypergeometric solutions with $a_+=-1,0$;  
\item the special solutions with $b_+=-1,0$, for which the hypergeometric $f_0,f_1$ are degree $1$ polynomials, which respectively correspond to the two  \textcolor{red}{\bf red} and \textcolor{blue}{\bf blue} ellipses (intersected with $D_{4-q}$). Note that the phase-transition point $(q=12/5,\eta_1=2/5)$ is the intersection point of the boundary of $D_{4-q}$ with the \textcolor{ao(english)}{\bf green} ellipse (a single tangency point).
\end{itemize}
\bigskip
\subsection{The $\eta_2=-q$ case}\label{-qcase}
Here the points $(q,\eta_1)$ must belong to
\begin{equation}\label{D-q} 
D_{-q}=\{(q,\eta_1):\,q\leq 0,\,\eta_1\geq-q/4\}.
\end{equation}
In this case, the first three equations \eqref{equn} together with \eqref{-nn} form a system of coupled ODEs with unknowns $\theta_j$,$j=0,1,2$,
\begin{align}\label{equ2bis}
\begin{cases}
(\xi-1)\theta_0'(\xi)-(q-3)\theta_0(\xi)+(q-2)\theta_1(\xi)=0,\\ 
2\xi(\xi-1)\theta_1'(\xi)-\left[\eta_1+1+(\eta_1+2q-7)\xi\right]\theta_1(\xi)+(\eta_1+q-3)\theta_0(\xi)\\ 
+(\eta_1+q-1)\xi\theta_2(\xi)=0,\\ 
2\xi(\xi-1)\theta_2'(\xi)-\left[2-q+(q-8)\xi\right]\theta_2(\xi)-4\theta_1(\xi)=0.
\end{cases}
\end{align}
\subsubsection{Polynomial Ansatz.}Let us now consider for $n\geq 0$ the following {\it Ansatz}, 
\begin{equation}\label{ansatzseries}
\theta_j(\xi)=(1-\xi)^{-\alpha}f_j(\xi),\quad f_j(\xi)=\sum_{k=0}^{n}A_j^kP_k(\xi),\quad P_k(\xi):=(1-\xi)^k,\;j=0,1,2.
\end{equation}
Eqs. \eqref{equ2bis} give 
\begin{align}\label{equ2-quater}
\begin{cases}
[1]\bullet \,\, (\xi-1)f_0'(\xi)-(q-3+\alpha)f_0(\xi)+(q-2)f_1(\xi)=0,\\ 
[2]\bullet \,\, 2\xi(\xi-1)f_1'(\xi)-\left[\eta_1+1+(\eta_1+2\alpha+2q-7)\xi\right]f_1(\xi)\\ 
\,\,\,\,\,\,\,\,\,\,\,\,\, +(\eta_1+q-3)f_0(\xi)+(\eta_1+q-1)\xi f_2(\xi)=0,\\ 
[3]\bullet \,\, 2\xi(\xi-1)f_2'(\xi)-\left[2-q+(2\alpha+q-8)\xi\right]f_2(\xi)-4f_1(\xi)=0.
\end{cases}
\end{align}
Consider then in each left-hand side of the three equations [1], [2], [3] in \eqref{equ2-quater}, the contributions arising for a \emph{fixed} $k$ from the monomials $A_j^kP_k$ in $f_j, j=0,1,2$ \eqref{ansatzseries}. Because of the universal presence of factors $(\xi-1)$ or $\xi(\xi-1)$ in front of derivatives $P_k'(\xi)=-kP_{k-1}(\xi)$, and of polynomials of degree at most 1 in $\xi$ in front of $P_k(\xi)$, only $P_k$ and $P_{k+1}$ monomials will result from $P_k$.  We explicitly find for the three lines the resulting contributions,
\begin{align}\label{PkPk+1}
\begin{cases}
[1]\bullet \left\{[k-(q-3+\alpha)]A_0^k+(q-2)A_1^k\right\} P_k,\\ 
[2]\bullet\left\{(\eta_1+q-3)A_0^k+2\left[k-(\eta_1+\alpha+q-3)\right]A_1^k +(\eta_1+q-1)A_2^k\right\}P_k\\
\,\,\,\,+\left\{(-2k+\eta_1+2\alpha+2q-7)A_1^k -(\eta_1+q-1)A_2^k\right\}P_{k+1},\\ 
[3]\bullet\left\{-4A_1^k+\left[2k-(2\alpha-6)\right]A_2^k\right\}P_k+\left(-2k+2\alpha+q-8\right)A_2^k P_{k+1}.
\end{cases}
\end{align}
When summing up over $k\in \{0,\cdots,n\}$ to reconstruct the $f_j$'s in \eqref{equ2-quater},  each monomial $P_k$ must get an overall vanishing coefficient in order to satisfy the equations. Therefore, collecting all coefficients of terms $P_k$, we are led to the recursions,
\begin{align}\label{AkAk-1}
\begin{cases}
[1] \bullet\,\, [k-(q-3+\alpha)]A_0^k+(q-2)A_1^k=0,\\ 
[2] \bullet\, \,(\eta_1+q-3)A_0^k+2\left[k-(\eta_1+\alpha+q-3)\right]A_1^k +(\eta_1+q-1)A_2^k\\
\,\,\,\,\,=[2(k-1)-(\eta_1+2\alpha+2q-7)]A_1^{k-1} +(\eta_1+q-1)A_2^{k-1},\\ 
[3] \bullet\,\, -4A_1^k+2\left[k-(\alpha-3)\right]A_2^k=\left[2(k-1)-(2\alpha+q-8)\right]A_2^{k-1}.
\end{cases}
\end{align}
Note  that in the  $k=0$ case, there are \emph{no} $A_1^{k-1}$ and $A_2^{k-1}$ terms,  which are thus set by convention equal to $0$.

This system is best written under a matricial form, by defining successively,
\begin{align*}
&D_k(\alpha):=
&\begin{pmatrix}k-(\alpha+q-3)&q-2&0\\
\eta_1+q-3&2k-2(\eta_1+\alpha +q-3)&\eta_1+q-1\\
0&-4&2k-2(\alpha-3)
\end{pmatrix}
\end{align*}
and 
\begin{align*}
&C_{k-1}(\alpha):=
&\begin{pmatrix}0&0&0\\
0&2(k-1)-(\eta_1+2\alpha +2q-7)&\eta_1+q-1\\
0&0&2(k-1)-(2\alpha+q-8)
\end{pmatrix}.
\end{align*}
Let us finally define the column vectors,
\begin{align*}
{\bf A}^k:=\left(
\begin{array}{c}
A_0^k\\
A_1^k\\
A_2^k\\
\end{array}
\right),\quad k\geq 0,
\end{align*}
so that recursions \eqref{AkAk-1} become for $k\geq 1$
\begin{align}\label{DkAk}
D_k(\alpha){\bf A}^k=C_{k-1}(\alpha) {\bf A}^{k-1},\quad k\geq 1,
\end{align}
together with the initial condition,
\begin{align}\label{D0A0}
D_0(\alpha){\bf A}^0={\bf 0},
\end{align}
and the closure relation
\begin{align}\label{An+1}
{\bf A}^{n+1}={\bf 0},
\end{align}
such that ${\bf A}^{k}={\bf 0},\,\, \forall k\geq n+1$. 

Equation \eqref{D0A0} shows that for a non-trivial solution to exist, one must have 
$$\det D_0(\alpha)=0,$$
with ${\bf A}^0$ an eigenvector of vanishing eigenvalue.
The determinant of $D_0$ is
\begin{align}\label{detD0}
\det D_0(\alpha)&=2(1-\alpha)E_0(\alpha)\\ \nonumber
E_0(\alpha)&:=2(\alpha+q-3)(\eta_1+\alpha+q-5)-q(\eta_1+q-3).
\end{align}
 Let us look for the solutions to
\begin{align}\label{alphapm0bisbis}
E_0(\alpha)&=0\Leftrightarrow \alpha=\alpha_0^{\pm}:= 3-q+\frac 12\left(2-\eta_1\pm\sqrt{\hat Z}\right),\\ \label{hatZnew} 
\hat Z&=\hat Z(\eta_1,q):= (\eta_1-2)^2+2q(\eta_1+q-3),
\end{align}
which yields the set of non-trivial zeroes of $\det D_0$. $\hat Z$ can also be written as 
\begin{align}\nonumber
\hat Z&=(q-1)^2+(\eta_1+q-2)^2-1\\ 
\label{hatZterter}
&=\hat X^2+Y^2-1, \quad \hat X:=q-1, \quad Y:=\eta_1+q-2.
\end{align}
Thus $\hat Z$ is non-negative outside the \textcolor{ao(english)}{\bf green} ellipse $\hat Z=0$ (Fig \ref{Fig-q}), and since $q\leq 0$ in $D_{-q}$, expression \eqref{hatZterter} is clearly non-negative there, vanishing only for $q=0,\eta_1=2$, so that  $\alpha_0^{\pm}$ \eqref{alphapm0bisbis} is defined and real in $D_{-q}$.
 Observe also that a translation maps $Z$ \eqref{Zbis} to $\hat Z$ \eqref{hatZterter},
\begin{equation}\label{hatZZ}
\hat Z(\eta_1,q)=Z(\eta_1-2,q+2).
\end{equation}
The null eigenvectors $\bf A^0(\alpha)$ of $D_0(\alpha)$ with vanishing eigenvalue are given, either for $\alpha=\alpha_0^{\pm}$  or for $\alpha=1$,  by  the one-dimensional space
\begin{align}\label{A0}
\bf A^0=\bf A^0(\alpha)=\begin{cases}
A_0^0\in \mathbb R\\
A^0_1=\frac{1}{q-2} (\alpha+q-3)A_0^0\\
A^0_2=-\frac{2}{\alpha -3} A^0_1.
\end{cases}
\end{align} 
As it will appear shortly, the value of the generalized spectrum is given by the root  $\alpha_0^{+}$ in \eqref{alphapm0bisbis}.
To check this, consider the first integrability line, $\eta_1=3-q$, where \eqref{alphapm0bisbis} gives, 
\begin{align*}
\begin{cases}
\alpha_0^{+}&=3-q,\,\,\,\quad \alpha_0^{-}=2,\,\,\,\quad \quad q\leq 1\\&=2\,\,\,\,\,\,\,\,\,\,\, \quad \quad \quad =3-q,\,\,\,\,\,q\geq 1.
\end{cases}
\end{align*}
Therefore, the choice of root $\alpha_0^+$ reproduces for $q\leq 1$, hence $q\leq 0$ in $D_{-q}$  the expected spectrum \eqref{3-q} $\beta(2,q)=3-q$.
 For the second integrability line, $\eta_1=1-q$, one similarly finds 
 \begin{align*}
\begin{cases}
\alpha_0^{+}&=4-q,\,\,\,\quad \alpha_0^{-}=3,\,\,\,\quad \quad q\leq 1\\&=3\,\,\,\,\,\,\,\,\,\,\, \quad \quad \quad =4-q,\,\,\,\,\,q\geq 1.
\end{cases}
\end{align*}
The condition $\eta_1\geq 0$ requires that $q\leq 1$, thus the choice of root $\alpha_0^+$ again gives for $q\leq 0$ in $D_{-q}$ the expected spectrum \eqref{1-q} $\beta(2,q)=4-q$. 
\subsubsection{Recursion} Assume for the time being that for $k\geq 1$, $\det D_k(\alpha)\neq 0$, so that $D_k(\alpha)$ is invertible. We immediately get from \eqref{DkAk}
\begin{align}\label{Ak}
{\bf A}^k&=M_k\, {\bf A}^{0}\\\nonumber M_k&=\prod_{\ell=0}^{k-1} D^{-1}_{k-\ell}(\alpha) C_{k-\ell-1}(\alpha) \\  \label{Mk}
&:=D^{-1}_{k}(\alpha) C_{k-1}(\alpha) D^{-1}_{k-1}(\alpha) C_{k-2}(\alpha)\cdots D^{-1}_{1}(\alpha) C_{0}(\alpha).
\end{align}
Notice that $D_k$ and $C_k$ obey a simple \emph{shift} relation,
\begin{equation}\label{shift0}
D_{k}(\alpha)=D_{0}(\alpha-k), \quad C_{k}(\alpha)=C_{0}(\alpha-k).
\end{equation}
Result \eqref{Mk} can then be rewritten simply as,
\begin{align}\label{Mkl}
M_k&=\prod_{\ell=0}^{k-1} D^{-1}_{0}(\alpha-k+\ell) C_0(\alpha-k+\ell+1)\\ \nonumber
&:=D^{-1}_{0}(\alpha -k) C_{0}(\alpha-k+1) D^{-1}_{0}(\alpha-k+1) C_{0}(\alpha-k+2)\cdots D^{-1}_{0}(\alpha-1) C_{0}(\alpha).
\end{align}
\subsubsection{Polynomial solutions}\label{ellipses} Requiring the $f_j$'s to be polynomials of given degree $n\geq 1$ is equivalent to requiring that ${\bf A}^{n+1}={\bf 0}$ \eqref{An+1}. From \eqref{DkAk}, we get $$D_{n+1}(\alpha) {\bf A}^{n+1}=C_{n} {\bf A}^{n}={\bf 0},$$ i.e., 
\begin{align}\label{eq7}
&[2n-(\eta_1+2\alpha +2q-7)]A^n_1+(\eta_1+q-1) A^n_2=0,\\ \label{eq8} 
&[2 n-(2\alpha+q-8)] A^n_2=0.
\end{align} 
If $2 n-(2\alpha+q-8)\neq 0$, then $A^n_2=0$, and for a non-vanishing solution to exist, one needs the condition $2n-(\eta_1+2\alpha +2q-7)=0$ to hold for $\alpha=\alpha_0^{\pm}$. 
This gives  $\pm \sqrt{\hat Z}=2n-1$, with $n\geq 1$. This selects the $(+)$-branch $\alpha=\alpha_0^+$, together with
\begin{equation}\label{hatZk}
\hat Z=(2n-1)^2,\,\,n\geq 1.
\end{equation}
Because of \eqref{hatZterter}, one thus finds from \eqref{hatZk} a set of {\bf ellipses} $\hat {\mathcal E}_n$ in the  $(q,\eta_1)$ plane, satisfying the equation 
\begin{equation}\label{hatelln}
\hat X^2+Y^2=(2n-1)^2+1.
\end{equation}
It is interesting to note that it is far from obvious that  condition \eqref{hatZk}, while necessary,  is also \emph{sufficient} to obtain that $A_2^n=0$ at level $n$ of recursion \eqref{Ak}, when starting from eigenvector ${\bf A}^0(\alpha)$ \eqref{A0} for $\alpha=\alpha_0^+$. We checked with {\bf Mathematica\textregistered}  that this is indeed the case, but despite repeated attempts, a combinatorial-like proof has eluded us. 
{\rkk Degeneracy of $D_k(\alpha).$} Let us finally consider the degenerate case when $\det D_k(\alpha)=\det D_0(\alpha-k)=0$. Since $\alpha=\alpha_0^+$, this  requires that either $\alpha_0^+ -k=\alpha_0^-$ or $\alpha_0^+ -k=1$.\\
$\bullet$ {\bf Case} $\alpha_0^+ -k=\alpha_0^-$. This gives $\hat Z=k^2$, and since $\hat Z=(2n-1)^2$, we get $k=2n-1$. Since the recursion stops at level $n+1$, we are only interested in the cases where $k\leq n+1$, hence $n\leq 2$. \\
When $n=1$, a direct computation gives the solution at level $k=n=1$ with the necessary condition $A^1_2=0$, as  $A_1^1=-\frac{1}{4}(\eta_1+q-1) A_2^0$, and $A_0^1=\frac{q-2}{\alpha+q-4}A_1^1.$ Adding to it any null eigenvector of $D_1(\alpha_0^+)=D_0(\alpha_0^-)$ would also provide a solution to the recursion at level $k=1$, but not its closure at level $n+1=2$. Taking into account the boundary condition $\theta_0(0)=f_0(0)=A^0_0+A_0^1=1$ yields the explicit solution on ellipse $\hat{\mathcal E}_1$,
\begin{align*}
&\theta_j(\xi)=(1-\xi)^{-\alpha_0^+}f_j(\xi),\,\,\,f_j(\xi)=A_j^0+A_j^1(1-\xi),\,\,\,j=0,1,2,\\
&j=0:\,\,\,A^0_0=2\frac{2-q}{\eta_1+1},\,\,\,A_0^1=\frac{\eta_1+2q-3}{\eta_1+1},\\ 
&j=1:\,\,\,A_1^0=\frac{\eta_1-3}{\eta_1+1},\,\,\, A_1^1=\frac{\eta_1+q-1}{2-q}\frac{\eta_1+q-3}{\eta_1+1},\\
&j=2:\,\,\,A_2^0=\frac{4}{q-2}\frac{\eta_1+q-3}{\eta_1+1},\,\,\,A_2^1=0.
\end{align*}\\
 When $n=2$, we get $k=n+1=3$. In that case, ${\bf A}^3 ={\bf 0}$ is a trivial solution yielding  ${\bf A}^k ={\bf 0}, \,\, \forall k\geq 3$, and while adding to it any null vector of $D_3(\alpha_0^+)$  would still satisfy the recursion at level $3$, it would not close the latter at next levels. \\
$\bullet$ {\bf Case} $\alpha_0^+ -k=1$. From \eqref{alphapm0bisbis} and \eqref{hatZk}, one finds $\eta_1=5+2(n-k)-2q$,  and one has to consider 
the intersection in $D_{-q}$ of this straight line with ellipse $\hat {\mathcal E}_n$.  One finds two solutions, 
$q_{n,k}^{\pm}:=n-k+2\pm \sqrt{\Delta},\,\,\, \Delta=(n-1)^2+(k-1)(2n+1-k).$ As before, we are only interested in recursion levels $1\leq k\leq n+1$, so that $\Delta \geq 0$. One also has $\Delta=(n-k+2)^2+2(k-2)(2n+1-k)$, so that for $k\geq 2$, there is one admissible root, $q_{n,k}^-\leq0$, whereas $q_{n,k}^+\geq 0$, and for $k=1$, $q_{n,1}^-=2,\,q_{n,1}^+=2n$ which are not in $D_{-q}$. By \emph{continuity}, at points $q_{n,k}^-$ with $2\leq k\leq n+1$ on $\hat {\mathcal E}_n$, the generalized integral means spectrum is still $\beta(2,q)=\alpha_0^+$.\\

Eq. \eqref{hatelln} gives for ellipse $\hat{\mathcal E}_n$ the equations in Cartesian coordinates $(q,\eta_1)\in D_{-q}$,
\begin{align}\label{ellk}
\begin{cases}
\eta_1&=2-q+\sqrt{q(2-q)+(2n-1)^2}\\
\alpha&=\alpha_0^+=\beta(2,q)=\frac 12(5-q+2k-\sqrt{q(2-q)+(2n-1)^2})\\
q&\in [1-\sqrt{1+(2n-1)^2},0],
\end{cases}
\end{align}
and
\begin{align}\label{redell2}
\begin{cases}
\eta_1&=2-q-\sqrt{q(2-q)+(2n-1)^2}\\
\alpha&= \alpha_0^+=\beta(2,q)=\frac 12(5-q+2k+\sqrt{q(2-q)+(2n-1)^2})\\
q&\in [1-\sqrt{1+(2n-1)^2},\inf\{-\frac 85(n-\frac 32),0\}].
\end{cases}
\end{align}
In the latter set, the $q=-\frac 85 (n-\frac 32)$ point for $n\geq 2$ corresponds to SLE$_\kappa$ with $\kappa=\frac 45 (n-\frac 32)$ and $\alpha=\beta\left(2,-\frac 85 (n-\frac 32)\right)
=\beta_1\left(2,-\frac 85 (n-\frac 32) ;\frac 45 (n-\frac 32)\right)$ in \eqref{eq:DHLZ}.\\
\begin{figure}[!ht]\begin{center}
\includegraphics[width=8cm]{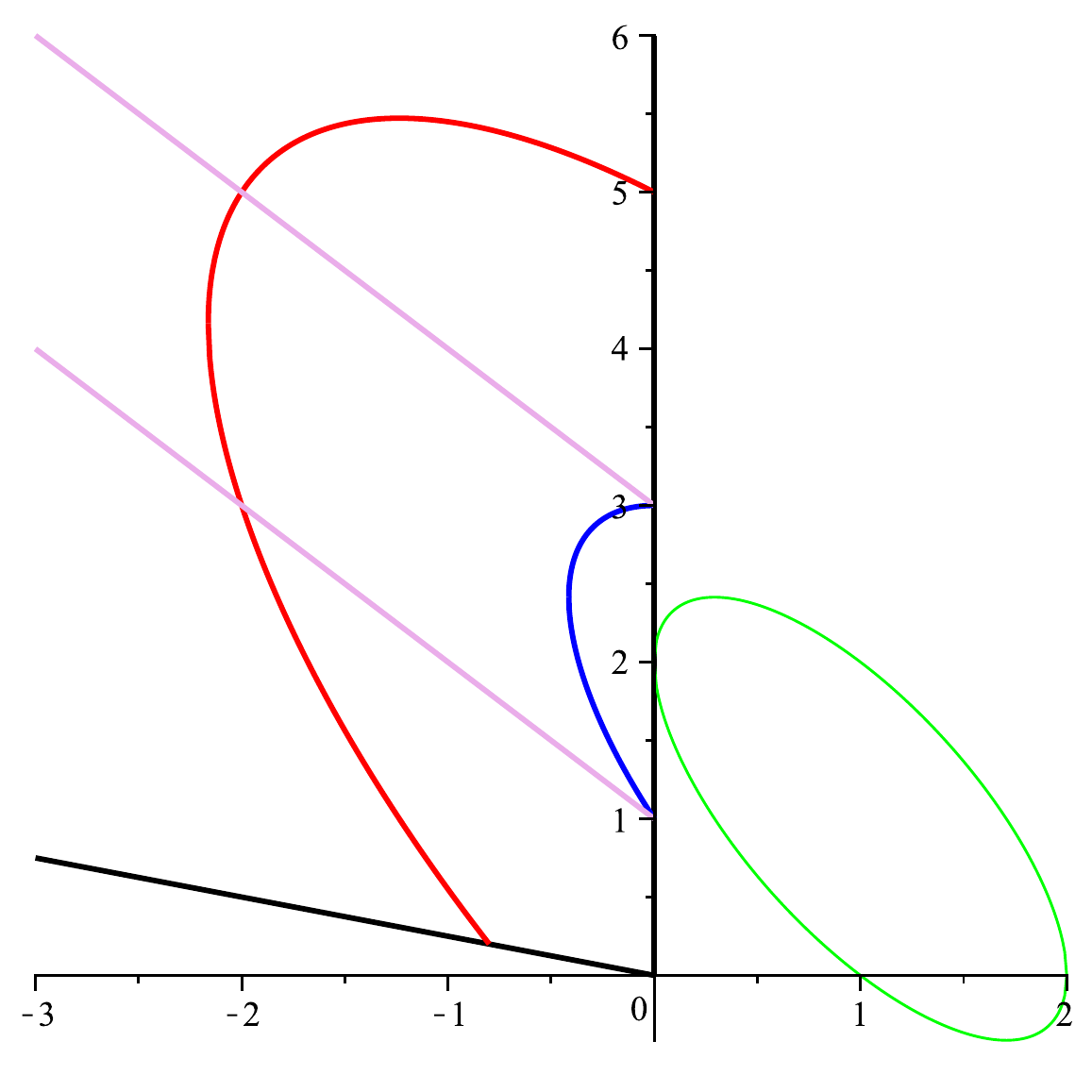}
\caption{{Domain lines in the $(q,\eta_1)$ plane, for the $\eta_2=-q$ case. The two purple straight lines have for equations, $\eta_1=1-q$ and $\eta_1=3-q$. The \textcolor{ao(english)}{\bf green}, \textcolor{blue}{\bf blue} and \textcolor{red}{\bf red} ellipses have for respective equations $\hat X^2+Y^2=1,2,10$, with $\hat X=q-1,Y=\eta_1+q-2$, with a common center located at $(q,\eta_1)=(1,1)$.}}
\label{Fig-q}
\end{center}
\end{figure}
The results are summarized in Fig. \ref{Fig-q}. The \textcolor{blue}{\bf blue} and \textcolor{red}{\bf red} ellipses in Fig.  \ref{Fig-q} above respectively correspond to the first $n=1$ and $n=2$ cases.  
Because of \eqref{hatZZ}, these ellipses are images by the $(q\to q-2,\eta_1\to \eta_1+2)$ translation of the corresponding same color ellipses of Fig. \ref{qeta1} in Section \ref{4-qcase}, with a common center now located at $(q,\eta_1)=(1,1)$.  
Note also that if the generalized spectrum is given in the whole region $D_{-q}$ by $\alpha_0^+$ in Eq. \eqref{alphapm0bisbis}, then on the SLE$_{\kappa}$-line where $\kappa=2\eta_1=\eta_2 /2=-q/2$ it coincides with the spectrum $\beta_1(2,q;\kappa=-q/2)=5-3q/2$, and there is no phase transition along that line, in contrast to the $(\eta_2=4-q)$-case. 
The phase transition now takes place on the $q=0=\eta_2$ line, at its $\eta_1=2$ contact point with the green ellipse, since the predicted spectrum $\alpha_0^+$ is there equal to $5-\eta_1$ for $0\leq \eta_1\leq 2$ and to $3$ for $\eta_1\geq 2$. This agrees with results \eqref{1-q}, $\beta(2,0)=4$ for $\eta_1=1$, and  \eqref{3-q}, $\beta(2,0)=3$ for $\eta_1=3$.
{\rkk Alternative condition.} For completeness, let us also mention that the alternative condition in \eqref{eq8}, $2 n-(2\alpha+q-8)= 0$ with $\alpha=\alpha_0^{\pm}$,  yields $2n+q+\eta_1=\pm \sqrt{\hat Z}$, 
leading to  $\eta_1+q-1=\frac{q(q-2)}{4(n+1)}-n$ and to the selection of the $(+)$-branch. From \eqref{eq7}, one further needs to check that $(\eta_1+q+1)A_1^n=(\eta_1+q-1)A_2^n$ upon starting the recursion \eqref{Ak}  from the null-eigenvector \eqref{A0}. Using again  {\bf Mathematica\textregistered} shows  this equality \emph{not} to hold, thus leading to no further solution.
\subsubsection{General solution to the Fuchsian system}\label{Fuchs}
 From the perspective of Fuchsian systems \cite{MR2363178,LY3}, the initial equations \eqref{equ2bis} for the vector function, 
 $\vartheta(\xi):= (\theta_0(\xi),\theta_1(\xi),\theta_2(\xi))^t$ can be written under the matrix form, 
$$\vartheta'(\xi)=  \frac{\mathbb A}{\xi}\vartheta(\xi)+\frac{ \mathbb B}{1-\xi} \vartheta(\xi),$$
where matrices $\mathbb A$ and $\mathbb B$ do not depend on $\xi$. They are simply given here by
\begin{align}\label{AA}
&\mathbb A=\frac{1}{2}
\begin{pmatrix}
0&0&0\\
\eta_1+q-3&-(\eta_1+1)&0\\
0&-4&q-2
\end{pmatrix},\\ \label{BB}
&\mathbb B=\frac{1}{2}
\begin{pmatrix}
2(3-q)&2(q-2)&0\\
\eta_1+q-3&-2(\eta_1+q-3)&\eta_1+q-1\\
0&-4&6
\end{pmatrix}.
\end{align}
The discriminant of $\mathbb B$ is 
\begin{align*}\det (\mathbb B-\alpha \mathds{1})&=(1-\alpha)\left[\frac{1}{2}(\eta_1+q-3)(2\alpha+q-6)+(\alpha-2)(\alpha+q-3)\right]\\
&=\frac{1}{2}(1-\alpha)E_0(\alpha)=\frac{1}{4}\det D_0(\alpha),
\end{align*}
so that $\mathbb B$ has for eigenvalues the zeroes of $\det D_0$, $\alpha_0^\pm$ \eqref{alphapm0bisbis} and $1$.

For $\xi\to 1^-$, the vector functions $\vartheta(\xi)= (\theta_0(\xi),\theta_1(\xi),\theta_2(\xi))^t$, have for asymptotic behavior,
  \begin{equation}\label{principe}\vartheta(\xi)=(1-\xi)^{-\alpha_0^+}{\bf f}^{+}(\xi)+(1-\xi)^{-\alpha_0^-}{\bf f}^-(\xi)+(1-\xi)^{-1}{\bf f}^1(\xi),\end{equation}
 where ${\bf f}^{\pm}, {\bf f}^1$  are vector functions with Taylor series expansions in $\xi$.  In the generic \emph{non-resonant}  case, where the eigenvalues do not differ by integer numbers, these functions converge at $\xi=1$ towards the  eigenvectors of $\mathbb B$ corresponding to their respective eigenvalue powers 
 $\alpha_0^{\pm},1$. In the \emph{resonan}t case, they can be polynomials in $\xi$, or can involve  polynomials in $-\log(1-\xi)$, which then dominate the limit when $\xi \to 1^-$.   
 
In the resonant case of the ${\mathcal E}_n$ ellipses of Section \ref{ellipses}, we have $\alpha_0^-=\alpha_0^+-(2n-1)$ with $n\geq 1$.  
 These vector functions are then simple polynomials, with $${\bf f}^+(\xi)=(f_0(\xi), f_1(\xi), f_2(\xi))^t,$$ and ${\bf f}^-(\xi)={\bf f}^1(\xi)={\bf 0}$. 
 At $\xi=1$, ${\bf f}^+(\xi)$ becomes the eigenvector ${\bf f}^+(\xi=1)={\bf A}_0(\alpha_0^+)$, as given in \eqref{A0}.
 
 Let us conclude with the following Theorem. 
{\theo The generalized integral means spectrum $\beta(2,q)$ of a whole-plane L\'evy-Loewner  process with L\'evy symbols $\eta_1$ and $\eta_2=-q$ is given in the whole $D_{-q}$  domain \eqref{D-q} by $\alpha_0^+$ in Eq. \eqref{alphapm0bisbis}.}
\begin{proof}
Because of \eqref{principe}, the generalized integral means spectrum must be equal to one of the three eigenvalues  $\alpha_0^+, \alpha_0^-,1$. 
 Eigenvalues $\alpha_0^+$ and $\alpha_0^-$ are equal only when $\hat Z=0$, i.e., on the green ellipse which lies outside $D_{-q}$, 
except for the point $P_0=(q=0, \eta_1=2)$. Therefore in $D^*_{-q}:=D_{-q}\setminus P_0$, we have $\alpha_0^+>\alpha_0^-$. One can also check that on $D_{-q}$, 
$\alpha_0^+ >1$, whereas the equality  $\alpha_0^- = 1$ is realized  in $D_{-q}$ 
on the branch of hyperbola of equation $(\eta_1+q-8)(q-4)=4$ with $q\leq 0$.  We know that 
$\beta(2,q)=3-q=\alpha_0^+$ on the half-line $\eta_1=3-q, q\leq 0$, as well as $\beta(2,q)=4-q=\alpha_0^+$ on the half-line $\eta_1=1-q, q\leq 0$. Because of the H\"older inequality, the 
generalized integral means spectrum is \emph{convex} in $(p,q)$ \cite{DHLZ}, hence \emph{continuous}. By continuity,  
$\beta(2,q)=\alpha_0^+$ on these integrability lines cannot jump to $\alpha_0^-< \alpha_0^+$ or $1<\alpha_0^+$ in $D^*_{-q}$,  hence $\beta(2,q)=\alpha_0^+$   
in the whole domain.  At the singular point $P_0$ , the spectrum is still $\alpha_0^+$, with a change of its analytic form.
\end{proof}

{\bf Acknowledgements:} This material is based upon work supported by the National Science Foundation under Grant No. DMS-1928930 while Bertrand Duplantier participated in the program ``Analysis and Geometry of Random Spaces",  hosted by the Mathematical Sciences Research Institute in Berkeley, California, during the Spring 2022 semester. The work by Yong Han is supported by the National Natural Science Foundation of China under Grant No.  12131016. B.D.  also wishes to warmly thank Emmanuel Guitter for his help with  {\bf Mathematica\textregistered} and the figures, and Thomas C. Halsey for a critical reading of the manuscript. 
\bibliographystyle{alpha}
\bibliography{biblio2}

\end{document}